\documentclass[pagesize,paper=a4]{article}
\usepackage[utf8]{inputenc}

\usepackage{natbib}
\usepackage{lscape}
\usepackage{pdflscape}
\usepackage{geometry}
\usepackage{float}
\usepackage{tikz}
\usetikzlibrary{arrows.meta}
\usepackage{amsmath,amsthm,amssymb}
\usepackage{tkz-graph}
\usepackage{tikz-qtree}
\usetikzlibrary{arrows}
\usepackage{amsmath}
\usepackage{multirow}
\usepackage{adjustbox}
\usepackage{subfigure}
\usepackage[ruled,vlined]{algorithm2e}
\usepackage{verbatim}
\tikzset{
  LabelStyle/.style = { rectangle, rounded corners, draw,
                        minimum width = 2em, fill = yellow!50,
                        text = red, font = \bfseries },
  VertexStyle/.append style = { inner sep=2pt,
                                font = \bfseries},
  EdgeStyle/.append style = {->,  left} }

\newtheorem{example}{Example}
\newtheorem{assumption}{Assumption}
\newtheorem{proposition}{Proposition}

\newtheorem{theorem}{Theorem}
\newtheorem{lemma}{Lemma}

\usepackage{array}
\newcolumntype{P}[1]{>{\centering\arraybackslash}p{#1}}

\usepackage{appendix}

\usepackage{setspace}
\usepackage{pdflscape}

\usepackage{xr}
\makeatletter
\newcommand*{\addFileDependency}[1]{
  \typeout{(#1)}
  \@addtofilelist{#1}
  \IfFileExists{#1}{}{\typeout{No file #1.}}
}
\makeatother
\newcommand*{\myexternaldocument}[1]{%
    \externaldocument{#1}%
    \addFileDependency{#1.tex}%
    \addFileDependency{#1.aux}%
}

\myexternaldocument{Supplementary}

\title{Bridging Parametric and Nonparametric Methods \\
in Cognitive Diagnosis 
\footnote{This research is partially supported by NSF CAREER SES-1846747, DMS-1712717, SES-1659328.}
}
\author{Chenchen Ma$^1$, Jimmy de la Torre$^2$, and Gongjun Xu$^1$\\
$^1$ University of Michigan\\
$^2$ University of Hong Kong}
\date{}
\begin{document}

\maketitle

\begin{abstract}
\noindent A number of parametric and nonparametric methods for estimating cognitive diagnosis models (CDMs) have been developed and applied in a wide range of contexts.
However, in the literature, a wide chasm exists between these two families of methods, and their relationship to each other is not well understood.
In this paper, we propose a unified estimation framework to bridge the divide between parametric and nonparametric methods in cognitive diagnosis to better understand their relationship.
We also develop iterative joint estimation algorithms and establish consistency properties within the proposed framework.
Lastly, we present comprehensive simulation results to compare different methods, and provide practical recommendations on the appropriate use of the proposed framework in various CDM contexts.

\end{abstract}

\section{Introduction}
\label{sec-intro}

Cognitive diagnosis models (CDMs), also known as diagnostic classification models, are typically used in conjunction with diagnostic assessments to determine fine-grained classifications of subjects' latent attribute patterns based on their observed responses to specifically-designed diagnostic items.
In educational assessments, the latent attributes can represent the mastery or lack of target skills \citep{de2011generalized,junker2001cognitive}.
Students' skill profiles, which are inferred from the their responses to test items, are used for subsequent learning interventions.
In psychiatric diagnosis, the latent attributes can be construed as the presence or absence of some underlying mental disorders \citep*{de2018analysis, templin2006measurement}.
Patients' responses to questionnaire items serve as the basis for identifying their mental disorder statuses, which in turn determines the appropriate treatments.

Several parametric models for cognitive diagnosis have been developed and widely applied in practice.
Popular examples include the deterministic input, noisy “and” gate (DINA) model \citep{junker2001cognitive}, the deterministic input, noisy “or” gate (DINO) model \citep{templin2006measurement}, the reduced reparameterized unified model  \citep[Reduced RUM;][]{hartz2002bayesian}, the general diagnostic model  \citep[GDM;][]{von2005general}, the log-linear CDM  \citep[LCDM;][]{henson2009defining}, and the generalized DINA model \citep[GDINA;][]{de2011generalized}.
To estimate these parametric models, estimators maximizing the marginal likelihood or joint likelihood functions have been employed \citep[e.g.,][]{chiu2016joint, de2009dina}.

Parametric CDMs, such as the DINA or DINO model, invoke certain parametric assumptions about the item response functions.
As pointed out in \cite{chiu2013nonparametric}, such assumptions may raise validity concerns about the assumed model and the underlying process.
As an alternative, some researchers have explored nonparametric methods for assigning subjects to latent groups without relying on parametric model assumptions.
For example, \cite{chiu2013nonparametric} proposed the nonparametric classification (NPC) method, where a subject is classified to its closet latent group by comparing the observed responses with ideal responses either from the DINA or DINO model.
Its generalization, the general NPC (GNPC) method proposed by \cite{chiu2018cognitive}, uses the weighted average of ideal responses from the DINA and DINO models to accommodate more general settings.
Consistency results for the NPC and the GNPC methods were established by \cite{wang2015consistency} and \cite{chiu2019consistency}, respectively.
Simulation results show that, compared to parametric methods, nonparametric methods tend to perform better when the sample sizes are not sufficiently large to provide reliable maximum likelihood estimates.

Even though the aforementioned parametric and nonparametric methods have been used in many CDM applications, the relationship between these two families of methods have not been explicitly discussed in the literature.
Although seemingly divergent from the surface, these frameworks are in fact closely related.
In this paper, we propose a unified estimation framework for cognitive diagnosis that subsumes both parametric and nonparametric methods.
In the proposed framework, we use a general loss function to measure the distance between a subject's responses and the centroid of a latent class.
By using different loss functions, the method can assume different parametric and nonparametric forms.
Under the general framework, we further develop a unified iterative joint estimation algorithm, as well as establish the consistency properties of the corresponding estimators.
 {Finally, we conduct comprehensive simulation studies to compare different parametric and nonparametric methods under a variety of settings, and provide relevant practical recommendations accordingly.}

The rest of the paper is organized as follows. Section \ref{sec-review} gives a brief review of both parametric and nonparametric methods  in cognitive diagnosis assessment.
Section \ref{sec-general} introduces the proposed general estimation framework with several illustrative examples. Section \ref{sec-analysis} presents the consistency results of the proposed method, and Section \ref{sec-simulation} presents the simulation results. Finally, Section 6 discusses some future extensions, whereas proofs of the main theorems are reported in the online Appendix.

\section{Parametric and Nonparametric Methods}
\label{sec-review}

Before introducing our proposed estimation framework, we give a brief review of both parametric and nonparametric methods that are widely used in the CDM literature.

\subsection{Parametric Methods}
\label{sec-para}

Parametric methods directly model item response functions under certain parametric model assumptions.
Most of CDMs are parametric models, where the item response probabilities are modeled as functions of item parameters and the latent attributes of subjects.
Specifically, in a CDM with $J$ items and $K$ latent attributes, two types of subject-specific variables are of interest.
One is the observed responses to $J$ items $\boldsymbol{x} = (x_1,\dots,x_J)\in\{0,1\}^J$, and the other is the mastery profile of $K$ latent attributes  $\boldsymbol{\alpha} = (\alpha_1,\dots,\alpha_K)\in \{0,1\}^K$.
 {There are $2^K$ possible latent patterns, }
and we use $\boldsymbol{p}=(p_{\boldsymbol{\alpha}}:\boldsymbol{\alpha}\in\{0,1\}^K)$ to denote the proportion parameters for the latent attribute patterns of subjects, which satisfies $p_{\boldsymbol{\alpha}}\in [0,1]$ and $\sum_{\boldsymbol{\alpha}\in\{0,1\}^K} p_{\boldsymbol{\alpha}}=1$.

Given a subject's latent attribute pattern $\boldsymbol{\alpha}$, the responses to $J$ items are assumed to be independent and follow Bernoulli distributions with parameters $\theta_{1,\boldsymbol{\alpha}},\dots,\theta_{J,\boldsymbol{\alpha}}$.
Specifically, $\theta_{j,\boldsymbol{\alpha}}:=\mathbb{P}(x_j=1|\boldsymbol{\alpha})$, which is the probability (item response function) of providing a positive response to item $j$ for latent class $\boldsymbol{\alpha}$. We write $\boldsymbol{\theta} = \big(\theta_{j, \boldsymbol{\alpha}}: j \in [J], \boldsymbol{\alpha} \in \{0,1\}^K\big)$ to denote the item  response probability matrix, with $[J]$ denoting the set $\{1,\ldots, J\}$. 
Then under the local independence assumption, the probability mass function of a subject's response vector $\boldsymbol{x} = (x_1,\dots,x_J)\in\{0,1\}^J$ takes the form
\begin{equation}
    \mathbb{P}(\boldsymbol{x}\mid \boldsymbol{\theta},\boldsymbol{p}) = \sum_{\boldsymbol{\alpha}\in\{0,1\}^K} p_{\boldsymbol{\alpha}}\prod_{j=1}^J \theta_{j,\boldsymbol{\alpha}}^{x_j}(1-\theta_{j,\boldsymbol{\alpha}})^{1-x_j}.
\end{equation}
To reflect the dependence between items and the latent attributes of subjects, a structural matrix, the so-called $Q$-matrix \citep{tatsuoka1983rule}, is used to impose constraints on item parameters. 
Specifically, $\boldsymbol{Q} \in \{0,1\}^{J\times K}$, where $q_{j,k}=1$ if item $j$ requires (or depends on) attribute $k$.
The $j$th row vector of $\boldsymbol{Q}$ denoted by $\boldsymbol{q}_j$ describes the full dependence of item $j$ on $K$ latent attributes.
Usually in applications such as cognitive diagnostic assessments, the matrix $\boldsymbol{Q}$ is pre-specified by domain experts \citep{george2015cognitive,  junker2001cognitive, von2005general} to reflect some scientific assumptions.
The structural matrix $\boldsymbol{Q}$ puts constrains on item parameters in certain ways under different model assumptions.
One important common assumption is that the item response function
 $\theta_{j,\boldsymbol{\alpha}}$ only depends on whether latent attribute pattern $\boldsymbol{\alpha}$ contains the required attributes by item $j$ (i.e., the attributes in the set $\mathcal{K}_j = \{k\in[K]:q_{j,k}=1\}$ with $[K]$ denoting the set $\{1,\ldots, K\}$).
Here we introduce three commonly used parametric models.

\begin{example}[DINA]
The DINA \citep{junker2001cognitive} model assumes a conjunctive relationship among attributes, where mastery of all the required attributes for an item is necessary for a subject to be deemed capable of providing a positive response (e.g., correct response, item endorsement), and possessing additional unnecessary attributes does not compensate for the lack of necessary attributes. In the DINA model, the so-called ideal response for each item $j$ and each latent attribute pattern $\boldsymbol{\alpha}$ is defined as,
\begin{equation}
    \eta_{j,\boldsymbol{\alpha}}^{\text{DINA}} = \prod_{k=1}^K \alpha_k^{q_{jk}}.
    \label{eq-ideal-DINA}
\end{equation}
The uncertainty is further incorporated by introducing the slipping and guessing parameters $s_j$ and $g_j$ for $j=1,\dots,J$.
For each item $j$, the slipping parameter is the probability of a negative response for capable subjects, and the guessing parameter is the probability of a positive response for incapable subjects, as in, $s_j = \mathbb{P}(x_{ij} = 0 \mid \eta_{j,\boldsymbol{\alpha}_i} = 1)$ and $g_j = \mathbb{P}(x_{ij} = 1 \mid \eta_{j,\boldsymbol{\alpha}_i} = 0)$, where $\boldsymbol{\alpha}_i$ is the latent pattern for the $i$th subject.
Therefore, we have
\begin{equation*}
    \theta_{j,\boldsymbol{\alpha}}^{\text{DINA}} = (1-s_j)^{\eta_{j, \boldsymbol{\alpha}}^{\text{DINA}}}
    g_j^{1-\eta_{j, \boldsymbol{\alpha}}^{\text{DINA}}}.
\end{equation*}
That is, the item response function is $1-s_j$ if the ideal response is 1,  and  $g_j$ otherwise.
\end{example}

\begin{example}[DINO]
The DINO \citep{templin2006measurement} model assumes a disjunctive relationship among attributes, where mastery of one of the required attributes of an item is necessary for a subject to be considered capable of providing a positive response.
In the DINO model, the ideal response is defined as,
\begin{equation}
    \eta_{j,\boldsymbol{\alpha}}^{\text{DINO}} = 1 - \prod_{k=1}^K(1- \alpha_k)^{q_{jk}}.
    \label{eq-ideal-DINO}
\end{equation}
For the DINO model, the slipping parameters and guessing parameters are defined in the similar way as the DINA model, as in, $s_j = \mathbb{P}(x_{ij} = 0 \mid \eta_{j, \boldsymbol{\alpha}_i} = 1)$ and $g_j = \mathbb{P}(x_{ij} = 1 \mid \eta_{j, \boldsymbol{\alpha}_i} = 0)$. Accordingly, we have
    $$\theta_{j,\boldsymbol{\alpha}}^{\text{DINO}} = (1-s_j)^{\eta_{j, \boldsymbol{\alpha}}^{\text{DINO}}}
    g_j^{1-\eta_{j,\boldsymbol{\alpha}}^{\text{DINO}}}.$$

\end{example}

\begin{example}[GDINA]
The GDINA \citep{de2011generalized} model is a more general CDM, where all the interactions among the required latent attributes by each item are considered. 
The item response function for the GDINA model is defined as
\begin{equation*}
    \theta_{j,\boldsymbol{\alpha}}^{\text{GDINA}} = f\Big(\sum_{S\subset \mathcal{K}_j}\beta_{j,S}\prod_{k\in S}\alpha_k \Big),
\end{equation*}
where $\mathcal{K}_j = \{k\in[K]:q_{j,k}=1\}$ is the set of required attributes by item $j$, and $f(\cdot)$ is a link function.
The link function is usually taken to be the identity, log, or logistic link.
In this work we use the identity link.
The coefficients can be interpreted as following: $\beta_{j,\emptyset}$ is the probability of a positive response for the most incapable subjects with $\boldsymbol{\alpha} = \mathbf{0}$; $\beta_{j,\{k\}}$ is the increase in the probability of a positive response for the subjects with $\alpha_{k} = 1$ compared to those with $\alpha_{k} = 0$; $\beta_{j,S}$ is the increase in the response probability for subjects with $\{\alpha_k = 1, k \in S\}$ compared to those missing one of the attributes in $S$.
By incorporating all the interactions among the required attributes, the GDINA model is one of the most general CDMs.
\end{example}

From a broader perspective, the aforementioned three CDMs belong to a general family of finite mixture models called restricted latent class models (RLCMs; \citealt{haertel1989using, xu2017identifiability}).
One common restriction is that all the capable subjects with all the required attributes have the same and highest item response parameters, that is,
\begin{equation}
    \label{eq-RLCM-A1}
    \underset{\boldsymbol{\alpha}:\boldsymbol{\alpha}\succeq \boldsymbol{q}_j}{\max} \theta_{j,\boldsymbol{\alpha}} = \underset{\boldsymbol{\alpha}:\boldsymbol{\alpha}\succeq \boldsymbol{q}_j}{\min} \theta_{j,\boldsymbol{\alpha}} \ \geq\ \theta_{j,\boldsymbol{\alpha}'} \ \geq\ \theta_{j,\boldsymbol{0}}, \text{ for any } \boldsymbol{\alpha}'\nsucceq \boldsymbol{q}_j,
\end{equation}
where we write $\boldsymbol{\alpha} \succeq \mathbf{q}_{j}$ if $\alpha_{k} \geq q_{j,k}$ for all $k=1,...,K$, and  $\boldsymbol{\alpha} \nsucceq \mathbf{q}_{j}$ otherwise.

To fit CDMs, popularly used parametric methods include marginal maximum likelihood estimation (MMLE) through EM algorithms \citep{de2009dina,de2011generalized} and MCMC techniques \citep{dibello2007review,von2005general}.
\cite*{chiu2016joint} also proposed a joint maximum likelihood estimation (JMLE) method for fitting CDMs.
The parametric estimation methods usually perform well when there are sufficiently large data.
However, as found in recent studies \citep*{chiu2019consistency, chiu2018cognitive}, they may either produce inaccurate estimates with small sample sizes or suffer from high computational costs.
This has lead researchers to consider nonparametric methods, which are reviewed in Section \ref{sec-nonpara}.

\subsection{Nonparametric Methods}
\label{sec-nonpara}

As the name suggests, nonparametric methods  no longer depend on  parametric model assumptions.
Instead of modeling item response functions, nonparametric methods directly classify the subjects to latent classes by minimizing the distance between subject's observed item responses and the centers of the latent classes.
Two popular examples of nonparametric methods are the NPC and the GNPC methods, which compare the subject's observed item responses to the so-called ideal response vectors of each proficiency-class.
Different CDMs define the ideal response vectors differently.
For example,
as specified in \eqref{eq-ideal-DINA} or \eqref{eq-ideal-DINO}, the ideal response in the DINA or DINO model will be 1 only if the subject possesses all the required attributes or one of the required attributes, respectively.
In the following, we give a brief introduction of the NPC and the GNPC methods. Please refer to \cite{chiu2019nonparametric}, \cite{chiu2013nonparametric} and \cite{chiu2018cognitive} for more details.

For the NPC method,  we use $M = 2^K$ to denote the total number of proficiency latent classes (i.e., attribute profiles), and for $m=1,\ldots, M$, we write $\boldsymbol{\eta}_m = (\eta_{1,m}, \eta_{2,m},\dots,\eta_{J,m})$ as the ideal response vector for the $m$th proficiency-class, where $\eta_{j,m}$ can be the DINA or DINO ideal response.
Given the ideal response vectors for each proficiency class, a subject is classified to the closest proficiency class that minimizes the distance between the subject's observed responses and the ideal responses:
\begin{equation*}
    \hat{\boldsymbol{\alpha}}_i = \underset{m\in\{1,2,\dots,M\}}{\arg\min} d(\boldsymbol{x}_i,\boldsymbol{\eta}_m),
\end{equation*}
where $d(\cdot)$ is a distance function.
For example, in \cite{chiu2013nonparametric}, they used the Hamming distance:
\begin{equation*}
    d_H (\boldsymbol{x}, \boldsymbol{\eta}) = \sum_{j=1}^J |x_j - \eta_j|.
\end{equation*}
In the NPC method, the ideal responses are either the DINA ideal responses or the DINO ideal responses, which are all binary; thus,
the absolute difference will be $0$ if the observed response is equal to the ideal response, and $1$ otherwise.
Moreover, because the observed and the ideal responses are all binary, the $L_2$ distance will lead to the same results as the Hamming distance in the NPC method.

Due to its dependence on the DINA or DINO model assumptions, which define two extreme relations between $\boldsymbol{q}$ and $\boldsymbol{\alpha}$, the NPC method may not be sufficiently flexible.
The GNPC method addresses this issue by considering a more general ideal response that represents a weighted average of the ideal responses of the DINA and DINO models, as in:
\begin{equation*}
	\eta_{j,m}^{(w)} = w_{j,m}\eta_{j,m}^{\text{DINA}} + (1-w_{j,m}) \eta_{j,m}^{\text{DINO}},
    \label{eq-gnpc}
\end{equation*}
where $w_{j,m}$ is the weight for the $j$th item and the $m$th proficiency class.
We use $\boldsymbol{\eta}_m^{(w)} = (\eta_{1,m}^{(w)},\dots,\eta_{J,m}^{(w)})$ to denote the weighted ideal response vector for the $m$th proficiency class in the GNPC method.
To get the estimates of the weights, \cite{chiu2018cognitive} proposed to minimize the $L_2$ distance between the responses to item $j$ and the weighted ideal responses $\eta_{j,m}^{(w)}$:
\begin{equation}
    d_{jm} = \sum_{i\in C_m} \big(x_{ij} - \eta_{j,m}^{(w)}\big)^2,
    \label{eq-weight-dist}
\end{equation}
where $\{C_m\}_{m=1}^M$ is the partition of the subjects into $M$ proficiency classes.
Minimizing \eqref{eq-weight-dist} leads to
\begin{equation*}
    \hat{w}_{j,m} = 1 - \Bar{x}_{j,C_m}, \quad \hat{\eta}_{j,m}^{(w)}=\Bar{x}_{j,C_m},
\end{equation*}
where $\Bar{x}_{j,C_m} = |C_m|^{-1} \sum_{i\in C_m} x_{ij}$,  the mean of the $j$th item responses for subjects in the $m$th proficiency class, and $|C_m|$ is the number of subjects in $C_m$.
Because the true memberships are unknown, they proposed to iteratively estimate the memberships and the ideal response vectors.
Specifically, starting with an initial partition of the subjects, the ideal response vectors are chosen to minimize the $L_2$ distance $\sum_{m=1}^M \sum_{i\in C_m} \sum_{j=1}^J (x_{ij} - \eta_{j,m}^{(w)})^2$.
The memberships of the subjects are then determined by minimizing the $L_2$ distance between the observed responses of a subject and the ideal response vectors estimated from the former step, as in, $\hat{\boldsymbol{\alpha}}_i = {\arg\min}_{m\in\{1,2,\dots,M\}}\ d\big(\mathbf{x}_i,{\hat{\boldsymbol{\eta}}_m}^{(w)}\big)$.

To implement the GNPC method, start with some initial values at $t=0$ step.
At the $(t+1)$th step, update the estimates as follows:
\begin{equation*}
	{\hat{\boldsymbol{\alpha}}_i}^{(t+1)}=\underset{m\in\{1,2,\dots,M\}}{\arg\min}
	d\big(\boldsymbol{x}_i, {\hat{\boldsymbol{\eta}}_m}^{{(w)}{(t)}}\big),
	\quad {\hat{\eta}_{j,m}}^{(w) (t+1)}= \ \Bar{\boldsymbol{x}}_{j,{\hat{C}_{m}^{(t+1)}}},
\end{equation*}
where $\hat{\boldsymbol{\eta}}_m^{(w) (t)}$ is the estimated centroids obtained in step $t$, and $\hat{C}_{m}^{(t+1)}$ is the partition of the subjects based on $\big\{\hat{\boldsymbol{\alpha}}_i^{(t+1)}\big\}_{i=1}^N$.
\cite{chiu2018cognitive} demonstrated through simulation studies that, compared to parametric methods, the nonparametric methods generally performed better in small-scale test settings.

\section{A General Estimation Framework}
\label{sec-general}

In this section, we propose a unified estimation framework that subsumes both the parametric and nonparametric models considered in Section \ref{sec-review}. This approach would facilitate a better statistical understanding of the relationship between the two families of CDM estimations.

For the parametric methods, we shall focus on the joint estimation of the subjects' latent classes $(\boldsymbol{\alpha}_i)_{i=1}^n$ and the model parameters.
Considering the joint maximum likelihood estimation for parametric CDMs and the nonparametric estimation approaches as introduced in Section \ref{sec-review}, we can see that the   $\boldsymbol{\theta}$ in the parametric models and the ideal response vectors $\boldsymbol{\eta}$ in the nonparametric methods are closely related, both denoting a certain ``centroid" of the responses of the latent classes under different model assumptions.
For instance, $\theta_{j,\boldsymbol{\alpha}} = P(x_j=1\mid \boldsymbol{\alpha})$ can be viewed as the statistical population average (center) of the responses to item $j$ of those subjects with attribute profile $\boldsymbol{\alpha}$, whereas $\eta_{j,\boldsymbol{\alpha}}$ corresponds to the nonparametric clustering center of the responses to item $j$ of those  in cluster $\boldsymbol{\alpha}$.
Therefore, similarly to the nonparametric clustering methods, the joint maximum likelihood estimation of parametric model can be viewed as minimization of some ``distance" function, introduced by the negative log-likelihood, between the observed responses and the ``centroid" responses $\boldsymbol{\theta}$.

Motivated by this observation, we propose a unified estimation framework for both the parametric and nonparametric methods.
Specifically, we let $\boldsymbol{A} = \big(\boldsymbol{\alpha}_i\big)_{i=1}^N$ denote a class membership matrix for $N$ subjects.
Based on the membership matrix $\boldsymbol{A}$, we can obtain a partition of $N$ subjects into $2^K$ proficiency classes, denoted by $\boldsymbol{C}(\boldsymbol{A} )=\big\{C_{\boldsymbol{\alpha}}(\boldsymbol{A} ):\boldsymbol{\alpha} \in \{0,1\}^K\big\}$, where $C_{\boldsymbol{\alpha}}(\boldsymbol{A} )$ denotes the set of subjects whose latent patterns are specified as $\boldsymbol{\alpha}$ by $\boldsymbol{A}$ .
For each latent class $\boldsymbol{\alpha}\in\{0,1\}^K$, we use $\boldsymbol{\mu}_{\boldsymbol{\alpha}}$ to denote the ``centroid'' parameters for both parametric and nonparametric methods.
Our proposed estimators for the latent attributes and centroid parameters are obtained by   minimizing a loss function of $(\boldsymbol{A}, \boldsymbol{\mu})$ as follows:

\begin{equation}
L(\boldsymbol{A},\boldsymbol{\mu}) := \sum_{\boldsymbol{\alpha}\in\{0,1\}^K} \sum_{i \in C_{\boldsymbol{\alpha}}(\boldsymbol{A} )}l(\boldsymbol{x}_i,\boldsymbol{\mu}_{\boldsymbol{\alpha}}),
 \label{eq-C}
\end{equation}
and the corresponding estimators are $(\hat{\boldsymbol{A}}, \hat{\boldsymbol{\mu}})=\underset{({\boldsymbol{A}, \boldsymbol{\mu}})}{\arg\min}L(\boldsymbol{A},\boldsymbol{\mu}).$
In \eqref{eq-C}, $l(\boldsymbol{x}_i,\boldsymbol{\mu}_{\boldsymbol{\alpha}})$ is a loss function that
measures the distance between the $i$th subject's response vector $\boldsymbol{x}_i$ and the centroid of latent class $\boldsymbol{\alpha}$.
Specifically, the loss function takes the additive form
 $l(\boldsymbol{x}_i,\boldsymbol{\mu}_{\boldsymbol{\alpha}}) = \sum_{j=1}^J l(x_{ij},\mu_{j,\boldsymbol{\alpha}})$,
where we abuse the notation $l(\cdot,\cdot)$  a little,  and when the loss function takes two vectors,
we use it to denote  the summation of the element-wise losses.
In this work, we also assume that $l(x_{ij},\mu_{j,\boldsymbol{\alpha}})$ is continuous in $\mu_{j,\boldsymbol{\alpha}}$.
Note that \eqref{eq-C} can also be expressed as
\begin{equation}\label{eqform}
	L(\boldsymbol{A},\boldsymbol{\mu}) = \sum_{\boldsymbol{\alpha}\in\{0,1\}^K} \sum_{i \in C_{\boldsymbol{\alpha}}(\boldsymbol{A} )}l(\boldsymbol{x}_i,\boldsymbol{\mu}_{\boldsymbol{\alpha}}) =  \sum_{i=1}^N \sum_{\boldsymbol{\alpha}\in\{0,1\}^K} I\{\boldsymbol{\alpha}_i = \boldsymbol{\alpha} \} \cdot l(\boldsymbol{x}_i,\boldsymbol{\mu}_{\boldsymbol{\alpha}}) = \sum_{i=1}^N l(\boldsymbol{x}_i,\boldsymbol{\mu}_{\boldsymbol{\alpha}_i}),
\end{equation}
which corresponds to a joint estimation of  $(\boldsymbol{A},\boldsymbol{\mu})$ under the loss function $l(\cdot,\cdot)$.
From the joint estimation perspective, we can show that, with appropriate loss functions (e.g., $L_1$, $L_2$, cross-entropy) and constraints on the centroids (e.g., centroids based on the ideal responses, weighted ideal responses, or specific CDM assumptions), the proposed framework can provide estimates for all the models discussed in Section \ref{sec-review}.
The examples below demonstrate how the NPC method, the GNPC method, and parametric estimation of the DINA and GDINA models can be derived from the proposed framework using various loss functions and centroid constraints.

\begin{example}[NPC]
In the proposed framework, let the ideal responses under the NPC method be the centroids, that is,
$\boldsymbol{\mu}_{\boldsymbol{\alpha}} = \boldsymbol{\eta}_{\boldsymbol{\alpha}}$.
If we use the $L_1$ loss function $l(x_{ij},\eta_{j,\boldsymbol{\alpha}}) = |x_{ij} - \eta_{j,\boldsymbol{\alpha}}|$,
then our proposed framework will become exactly the NPC method.
Recall that in the NPC method, the ideal response vectors $\boldsymbol{\eta}_{\boldsymbol{\alpha}}$ are determined by pre-specified model assumptions (either the DINA or the DINO); thus, we only need to classify each subject to the closest proficiency class.
\label{ex-NPC}
\end{example}

\begin{example}[GNPC]
Recall that in the GNPC method, the ideal response is defined as $\eta_{j,m}^{(w)} = w_{j,m}\eta_{j,m}^{\text{DINA}} + (1-w_{j,m}) \eta_{j,m}^{\text{DINO}}$, a weighted average of the DINA ideal response and the DINO ideal response.
 Note that for proficiency classes and items such that
$\boldsymbol{\alpha} \succeq \boldsymbol{q}_j$,  we have $\eta_{\boldsymbol{\alpha},j}^{\text{DINA}} = \eta_{\boldsymbol{\alpha},j}^{\text{DINO}} = 1$, and for $\boldsymbol{\alpha}\odot\boldsymbol{q}_j={\mathbf 0}$, where  $\odot$ denotes the elementwise multiplication of vectors,  we have $\eta_{\boldsymbol{\alpha},j}^{\text{DINA}} = \eta_{\boldsymbol{\alpha},j}^{\text{DINO}} = 0$.
In such cases, the weights in fact do not affect the weighted ideal responses since the DINA and the DINO models have the same ideal responses.
Therefore, if we constrain $\boldsymbol{\mu}_{\boldsymbol{\alpha}} =(\mu_{j,\boldsymbol{\alpha}}, j=1,\ldots, J)$ in \eqref{eq-C}, such that
$\mu_{j,\boldsymbol{\alpha}} = 1$ if $\boldsymbol{\alpha} \succeq \boldsymbol{q}_j,$
$\mu_{j,\boldsymbol{\alpha}} = 0$ if $\boldsymbol{\alpha} \odot \boldsymbol{q}_j={\mathbf 0}$, and $ \mu_{j,\boldsymbol{\alpha}} = \eta_{j,m}^{(w)}$ as defined in the GNPC
for the rest of the items,
while at the same time use the $L_2$ loss function $l(x_{ij},\eta_{j,\boldsymbol{\alpha}}) = (x_{ij} - \eta_{j,\boldsymbol{\alpha}})^2$,
then the criterion in \eqref{eq-C} is equivalent to the GNPC method.
\label{ex-GNPC}
\end{example}

\begin{example}[DINA]
Let's consider the cross-entropy loss (i.e., the negative log-likelihood function),
\begin{equation}
    l(x_{ij}, \mu_{j,\boldsymbol{\alpha}})
    = - \big(x_{ij}\log \mu_{j,\boldsymbol{\alpha}} + (1-x_{ij})\log(1-\mu_{j,\boldsymbol{\alpha}})\big).
    \label{eq-ce}
\end{equation}
In addition, if we constrain the centroids to satisfy the following conditions:
\begin{equation*}
	\underset{\boldsymbol{\alpha}:\boldsymbol{\alpha}\succeq \boldsymbol{q}_j}{\max} \mu_{j,\boldsymbol{\alpha}} = \underset{\boldsymbol{\alpha}:\boldsymbol{\alpha}\succeq \boldsymbol{q}_j}{\min} \mu_{j,\boldsymbol{\alpha}} \ \geq \ \underset{\boldsymbol{\alpha}:\boldsymbol{\alpha}\nsucceq \boldsymbol{q}_j}{\max} \mu_{j,\boldsymbol{\alpha}} = \underset{\boldsymbol{\alpha}:\boldsymbol{\alpha}\nsucceq \boldsymbol{q}_j}{\min} \mu_{j,\boldsymbol{\alpha}},
\end{equation*}
that is, all the capable subjects share the same higher item positive probabilities, whereas all the incapable subjects share the same lower item probabilities, then the proposed criterion \eqref{eq-C} becomes the JMLE criterion for the DINA model. Moreover, the centroids here correspond to item response parameters $\boldsymbol{\theta}$ for each latent class in the DINA model.
\label{ex-DINA}
\end{example}

\begin{example}[GDINA]
In Example \ref{ex-DINA}, we can put the following constraints on the centroids:
$\mu_{j,\boldsymbol{\alpha}} = \mu_{j,\boldsymbol{\alpha}'}, \text{ if } \boldsymbol{\alpha}_{\mathcal{K}_j} = \boldsymbol{\alpha}'_{\mathcal{K}_j}, $
where $\boldsymbol{\alpha}_{\mathcal{K}_j}= (\alpha_k)_{k\in \mathcal{K}_j}$ is the sub-vector of $\boldsymbol{\alpha}$ on the set $\mathcal{K}_j$, and $\mathcal{K}_j = \{k\in[K]:q_{j,k}=1\}$ is the set of required attributes by item $j$.
Equivalently, these constraints will result in the same centroid parameters for any two latent patterns sharing the same values on the required attributes of item $j$, which is a GDINA model assumption. Furthermore, if we take the same loss functions as in Example \ref{ex-DINA}, it will result in the JMLE criterion for the GDINA model.
Again, the centroids correspond to item response parameters $\boldsymbol{\theta}$ for each proficiency class.
\label{ex-GDINA}
\end{example}{}

As demonstrated in the above examples, by taking different loss functions and different constraints on the centroid of each latent class, our proposal \eqref{eq-C} provides a general estimation framework bridging both the parametric and nonparametric methods in the literature.
The parametric estimation approaches mostly use the cross-entropy loss (negative log-likelihood) function, whereas the nonparametric approaches use the $L_1$ or $L_2$ distance measures.
 The analogous roles of negative log-likelihood for a parametric CDM and the distance function for a nonparametric CDM were also noted in \cite{chiu2018cognitive}.

It can be noted that the proposed estimation criterion \eqref{eq-C} does not directly use the information pertaining to the population distribution of the latent attribute profiles, which differentiates it from marginal likelihood estimation.
As the population proportion of each latent class of attribute profiles may also provide useful information for the model estimation, we propose to further generalize \eqref{eq-C} by including the proportion parameters in the loss function as follows:
\begin{equation}
 L(\boldsymbol{A} ,\boldsymbol{\mu},\boldsymbol{\pi}) := \sum_{\boldsymbol{\alpha}\in\{0,1\}^K} \sum_{i \in C_{\boldsymbol{\alpha}} (\boldsymbol{A} )} \Big( l(\boldsymbol{x}_i,\boldsymbol{\mu}_{\boldsymbol{\alpha}}) + h(\pi_{\boldsymbol{\alpha}})\Big),
\label{eq-C2}
\end{equation}
where $l(\cdot,\cdot)$ is the loss function as in \eqref{eq-C}, and $h(\cdot)$ is a continuous nonincreasing regularization function of the proportion parameter $\pi_{\boldsymbol{\alpha}}$, which denotes the population proportion of latent class $\boldsymbol{\alpha}$.
As can be seen from \eqref{eq-C2}, the loss function $L$ depends on both the centroids and the class proportions, with one part measuring the distance between a subject's response $\boldsymbol{x}_i$ and the centroid of a latent class $\boldsymbol{\mu}_{\boldsymbol{\alpha}}$, and the other part involving a regularization of class proportions.

Implicitly, Examples \ref{ex-NPC}--\ref{ex-GDINA} take $h(\pi_{\boldsymbol{\alpha}})=0$. When we take the loss function $l(x_{ij}, \mu_{j,\boldsymbol{\alpha}})$ to be the cross-entropy loss function as in \eqref{eq-ce}, and let $h(\pi_{\boldsymbol{\alpha}}) = - \log \pi_{\boldsymbol{\alpha}}$, then \eqref{eq-C2} becomes
\begin{equation}
	L(\boldsymbol{A} ,\boldsymbol{\mu},\boldsymbol{\pi}) = \sum_{\boldsymbol{\alpha}\in\{0,1\}^K}\sum_{i\in C_{\boldsymbol{\alpha}}(\boldsymbol{A} ) } \Big(l(x_{ij}, \mu_{j,\boldsymbol{\alpha}})- \log \pi_{\boldsymbol{\alpha}}\Big)
	  = - \sum_{i=1}^N \log\Big\{\pi_{\boldsymbol{\alpha}_i}\times   Lik(\boldsymbol{x}_i;\boldsymbol{\mu}_{\boldsymbol{\alpha}_i})\Big\},
	\label{eq-CML2}
\end{equation}
where $Lik(\boldsymbol{x};\boldsymbol{\mu}_{\boldsymbol{\alpha}})=P(\boldsymbol{x} \mid \boldsymbol{\mu}_{\boldsymbol{\alpha}})$ is the likelihood function for latent class ${\boldsymbol{\alpha}}$ and observation $\boldsymbol{x}$, and  $\boldsymbol{\mu}_{\boldsymbol{\alpha}} =(\mu_{j,\boldsymbol{\alpha}}, j=1,\ldots, J)$ is the corresponding model parameters with $\mu_{j,\boldsymbol{\alpha}} = \theta_{j,\boldsymbol{\alpha}} =P(x_{ij}=1\mid \boldsymbol{\alpha})$.
Note that
 $ \pi_{\boldsymbol{\alpha}_i}\times   Lik(\boldsymbol{x}_i;\boldsymbol{\mu}_{\boldsymbol{\alpha}_i})$ in the RHS of \eqref{eq-CML2} corresponds to the {\it complete-data} likelihood of $(\boldsymbol{\alpha}_i, \boldsymbol{x}_i)$; therefore, the loss function \eqref{eq-CML2} is in fact the complete-data log-likelihood of $(\boldsymbol{A,X})$.

The loss function \eqref{eq-CML2} also corresponds to the extension of the classification maximum likelihood (CML) criterion \citep{celeux1992classification} applied to the CDM setting.
In Examples \ref{ex-DINA} and \ref{ex-GDINA}, using the loss function as in \eqref{eq-CML2} corresponds to the CML criterion for the DINA or GDINA model respectively.
It can be noted that the CML differs from the JMLE in that the former has an additional term $\log \pi_{\boldsymbol{\alpha}}$ in the loss function to make use of the information in the proportion parameters.
The CML is also closely related to the EM estimation for the marginal MLE in that the CML directly maximizes the complete-data log-likelihood whereas the EM algorithm maximizes the expected complete-data log-likelihood with respect to the posterior distribution of the latent variables.
Finally, it can also be underscored that, by incorporating a wide range of loss functions, the proposed criterion \eqref{eq-C2} is a generalization of the CML criterion \eqref{eq-CML2}.

\medskip
To implement the unified estimation framework, we develop an algorithm to minimize \eqref{eq-C2}.
The algorithm is a general iterative algorithm to classify each subject to the closet proficiency class.
Starting from initial values, the current loss for each subject's responses and the centroid of each latent class is first computed, after which the subject is assigned to the closest latent class that minimizes the loss.
Based on the assigned memberships, the estimates for the centroids and class proportions are updated.
The details of the steps are shown in Algorithm \ref{algo1}.
\begin{algorithm}[h!]
\caption{General Iterative Classification Algorithm}
\label{algo1}
\SetKwInOut{Input}{Input}
\SetKwInOut{Output}{Output}

\Input{\ Binary response matrix $\boldsymbol{X}\in\{0,1\}^{N\times J}$ and structural $Q$-matrix $\boldsymbol{Q}\in\{0,1\}^{J\times K}$}

Initialize $\hat{\boldsymbol{A}}^{(0)}$, $\hat{\boldsymbol{\mu}}^{(0)}$ and $\hat{\boldsymbol{\pi}}^{(0)}$.

\While{convergence not reached}{

At the $(t+1)^{th}$ iteration,

\textbf{Step 1}:
Compute the current loss between $\boldsymbol{x_i}$ and the centroid of each proficiency class,
\begin{equation*}
    l(\boldsymbol{x}_i,\hat{\boldsymbol{\mu}}_{\boldsymbol{\alpha}}^{(t)}) + h(\hat{\pi}_{\boldsymbol{\alpha}}^{(t)}) ,\ i = 1,\dots,N,\ \boldsymbol{\alpha}\in\{0,1\}^K.
    \label{eq-step1}
\end{equation*}{}

\textbf{Step 2}:
Assign each $\boldsymbol{x}_i$ to the closest proficiency class, as in,
    \begin{equation*}
        \hat{\boldsymbol{\alpha}}_i^{(t)} = \underset{\boldsymbol{\alpha}}{\arg\min} \ l(\boldsymbol{x}_i,\hat{\boldsymbol{\mu}}_{\boldsymbol{\alpha}}^{(t)}) + h(\hat{\pi}_{\boldsymbol{\alpha}}^{(t)}),\ i=1,\dots,N.
        \label{eq-step2}
    \end{equation*}
    \quad \quad \quad \ and obtain the resulting partition $\hat{\boldsymbol{C}}^{(t)}:= \boldsymbol{C} (\boldsymbol{\hat{A}}^{(t)} )$.
\vspace{0.1in}

\textbf{Step 3}:
Compute the centroid and proportion of each proficiency class,
    \begin{equation*}
        (\hat{\boldsymbol{\mu}}_{\boldsymbol{\alpha}}^{(t+1)}, \hat{\boldsymbol{\pi}}_{\boldsymbol{\alpha}}^{(t+1)}) = \underset{(\boldsymbol{\mu},\boldsymbol{\pi})}{\arg\min} \sum_{i \in \hat{C}_{\boldsymbol{\alpha}}^{(t)} } \Big(l(\boldsymbol{x}_i,\hat{\boldsymbol{\mu}}_{\boldsymbol{\alpha}}^{(t)}) + h(\hat{\pi}_{\boldsymbol{\alpha}}^{(t)})\Big), \ \boldsymbol{\alpha}\in\{0,1\}^K.
        \label{eq-step3}
    \end{equation*}{}}
\Output{$\hat{\boldsymbol{A}}$, $\hat{\boldsymbol{\mu}}$, and $\hat{\boldsymbol{\pi}}$.}
\end{algorithm}

In the CDM context, certain proficiency classes share the same item response parameters for each item given a particular $Q$-matrix.
For example, for all CDMs, any $\boldsymbol{\alpha}$ such that $\boldsymbol{\alpha}\succeq \boldsymbol{q}_j$, has the same item response function; for the DINA model, there are only two levels of probabilities for each item's  item response function, and the capable classes with $\boldsymbol{\alpha}\succeq \boldsymbol{q}_j$ share the same item response function $1-s_j$, 
and the incapable classes with $\boldsymbol{\alpha}\nsucceq \boldsymbol{q}_j$ share the same item response function $g_j$. 
Based on this observation, under certain model assumptions, the proficiency classes can be partitioned into some equivalence classes for each item according to the $Q$-matrix.
Specifically, for item $j$,
let $\Tilde{A}_j = \big\{\Tilde{A}_{j,\boldsymbol{\alpha}}=\{\boldsymbol{\alpha}': \mu_{j,\boldsymbol{\alpha}} = \mu_{j,\boldsymbol{\alpha}'}\} \big\}$.
Under this partitioning, the proficiency classes in the same equivalence class will have the same item response probability for this item.
For example, in a DINA model with two latent attributes, if $\boldsymbol{q}_j = (1,0)$, then the proficiency classes can be partitioned into $\Big\{ \big\{(0,0),(0,1)\big\}, \big\{(1,0),(1,1)\big\} \Big\}$, where $\boldsymbol{\alpha} \in \big\{(0,0),(0,1)\big\}$ will have the same item response function, $g_j$, and $\boldsymbol{\alpha} \in \big\{(1,0),(1,1)\big\}$  will also share the same  item response function, $1 -s_j$.
 Therefore, by incorporating information of the $Q$-matrix and certain model assumptions, we can develop an iterative classification algorithm tailored for CDMs that updates the centroids associated with equivalence classes together.

To illustrate, if we let the negative log-likelihood function be the loss function as specified in (\ref{eq-ce}), then Step 3 in Algorithm \ref{algo1} simplifies to
\begin{align*}
    \hat{\pi}_{\boldsymbol{\alpha}}^{(t+1)} = \frac{\sum_{i=1}^N I\{\hat{\boldsymbol{\alpha}}_i^{(t)} = \boldsymbol{\alpha} \}}{N},\quad
    \hat{\mu}^{(t+1)}_{j,\boldsymbol{\alpha}} = 
    \frac{\sum_{\boldsymbol{\alpha}'\in \Tilde{A}_{j,\boldsymbol{\alpha}}}\sum_{i \in \hat{C}^{(t)}_{\boldsymbol{\alpha}'}} \ x_{ij}}{\sum_{\boldsymbol{\alpha}'\in \Tilde{A}_{j,\boldsymbol{\alpha}}} |\hat{C}^{(t)}_{\boldsymbol{\alpha}'}|},
\end{align*}
where $|\cdot|$ is the cardinality of a set. Based on this simplification, the estimated proportion parameters are the sample proportions based on the estimated partition of the subjects, and the estimated centroids are the corresponding sample means of the equivalence classes also based on the estimated partition.
Moreover, if fixed and equal proportions, together with $L_2$ loss $l(x_{ij}, \mu_{j,\mathbf{\alpha}}) = (x_{ij} - \mu_{j,\mathbf{\alpha}})^2$ are used, the algorithm becomes the iterative algorithm for the GNPC method outlined in \cite{chiu2018cognitive}.

\section{Analysis of the Proposed Framework}
\label{sec-analysis}

In this section, we provide a theoretical analysis of the proposed framework.
We show that, under certain conditions, the proposed estimation framework will give consistent estimates. The consistency results can be regarded as extensions of those for the NPC and the GNPC methods developed in \cite{wang2015consistency} and \cite{chiu2019consistency}.
In addition to the asymptotic results, we also provide an analysis of the proposed algorithm in the finite sample situations.

As we introduced in Section \ref{sec-para}, all the parametric CDMs belong to the family of latent class models.
Hence, in our following analysis, we assume a general latent class model as the underlying model.
Our results below are also easily adapted to the $Q$-matrix restricted latent class models.
We use $\theta_{j,\boldsymbol{\alpha}}^0$ to denote the true probability of providing a positive response for the  $j$th item and latent pattern $\boldsymbol{\alpha}$, as in, $\theta_{j,\boldsymbol{\alpha}}^0 = \mathbb{P}(x_j=1\mid \boldsymbol{\alpha})$,
and we use $\boldsymbol{\theta}^0_{\boldsymbol{\alpha}} = (\theta^0_{1,\boldsymbol{\alpha}},\dots,
\theta^0_{J,\boldsymbol{\alpha}})$ to denote item   probability vector for latent pattern $\boldsymbol{\alpha}$.
We let $\boldsymbol{A}^0 = (\boldsymbol{\alpha}_i^0)_{i=1}^N$ denote the true latent pattern matrix of the $N$ subjects to be classified.
Before we establish the consistency results, we first make some mild assumptions.

\begin{assumption}
The loss function $l(x,\mu)$ is H\"older continuous in $\mu$ on $[\tau, 1-\tau]$ for any $\tau \in(0,0.5)$,
and the total loss \eqref{eq-C2} is minimized at class means given the subjects' membership, as in,
$\hat{\mu}_{j,\boldsymbol{\alpha}} =  \sum_{i \in C_{\boldsymbol{\alpha}}}  x_{ij} / |C_{\boldsymbol{\alpha}}|$.
\label{assump-minima}
\end{assumption}{}

\begin{assumption}
$h(\cdot)$ in \eqref{eq-C2} is a continuous nonincreasing function of the proportion parameters.
\label{assump-prop}
\end{assumption}{}

\begin{assumption}\label{assump-delta-theta}
There exist constants $ \delta_1,\delta_2>0$ such that
 $\underset{J\rightarrow \infty}{\lim} \{ \underset{\boldsymbol{\alpha}\neq \boldsymbol{\alpha}'}{\min}\  J^{-1}\|\boldsymbol{\theta}_{\boldsymbol{\alpha}}^0 -   \boldsymbol{\theta}_{ \boldsymbol{\alpha}'}^0 \|_1\} \geq \delta_1$,
 and
  $\delta_2 \leq \underset{j,\boldsymbol{\alpha}}{\min} \  \theta_{j,\boldsymbol{\alpha}}^0 < \underset{j,\boldsymbol{\alpha}}{\max} \  \theta_{j,\boldsymbol{\alpha}}^0 \leq 1-\delta_2$, where $\|\cdot\|_1$ denotes the $L_1$ norm.
\end{assumption}

\begin{assumption}\label{assump-delta-loss}
There exists $\delta \geq 1$ such that
\begin{equation}
    \Big| E[l(x_{ij},\theta_{j,\boldsymbol{\alpha}}^0)] - E[l(x_{ij},\theta_{j,\boldsymbol{\alpha}^0_i}^0)] \Big| \geq \big|\theta_{j,\boldsymbol{\alpha}}^0 - \theta_{j,\boldsymbol{\alpha}^0_i}^0\big|^{\delta},\ \forall \ \boldsymbol{\alpha} \neq \boldsymbol{\alpha}^0_i.
\label{eq-theta}
\end{equation}
\end{assumption}

One can easily check that the $L_2$ and cross-entropy (negative log-likelihood) loss functions given in Section \ref{sec-general}  satisfy Assumption \ref{assump-minima}.
Note that the second part of Assumption \ref{assump-minima} is a natural requirement for the consistent estimation of $\theta_{j,\boldsymbol{\alpha}}^0$,
as $\theta_{j,\boldsymbol{\alpha}}^0$ represents  the population average of the responses of subjects in latent class $\boldsymbol{\alpha}$, that is,  $\theta_{j,\boldsymbol{\alpha}}^0 = \mathbb{P}(x_j=1\mid \boldsymbol{\alpha})$.
Given the true memberships of the subjects, for an estimator $\hat\mu_{j,\boldsymbol{\alpha}}$ that is consistent for  $\theta_{j,\boldsymbol{\alpha}}^0$, it must satisfy $\big|\hat\mu_{j,\boldsymbol{\alpha}}-\sum_{i \in C_{\boldsymbol{\alpha}}}  x_{ij} / |C_{\boldsymbol{\alpha}}| \big |\to 0$ in probability by the law of large number.
An interesting counterexample is the $L_1$ loss function, which does not satisfy this assumption because given the memberships, $\hat\mu_{j,\boldsymbol{\alpha}}$ that minimizes the $L_1$ loss function is the sample median instead of the sample mean.
Since in the CDM setting the responses are binary, the sample median would be either 0 or 1, which makes $\hat\mu_{j,\boldsymbol{\alpha}}$ under the $L_1$ loss not a consistent estimator of  $\theta_{j,\boldsymbol{\alpha}}^0$ even when the true memberships are known.
In other words, the $L_1$ loss cannot provide consistent estimation of the centroid parameters while  the  $L_2$ and cross-entropy losses can, as to be shown in the following theorems.
More generally, following the M-estimation theory \citep{van2000asymptotic}, the second part of Assumption \ref{assump-minima} can be further relaxed to requiring $E_{\theta_{j,\boldsymbol{\alpha}}^0}[l(x_{ij},\mu_{j,\boldsymbol{\alpha}})]$ has a unique minima at $\theta_{j,\boldsymbol{\alpha}}^0$ and some additional technical conditions. For the presentation brevity, here we shall use the current  assumption, which is already broad enough for practical use.

Assumptions \ref{assump-prop} and \ref{assump-delta-theta} ensure the identifiability of the model, and also keep the true parameters away from the boundaries of the parameter space.
Particularly, the assumption $\underset{J\rightarrow \infty}{\lim} \big\{\underset{\boldsymbol{\alpha}\neq \boldsymbol{\alpha}'}{\min}\  J^{-1}\|\boldsymbol{\theta}_{\boldsymbol{\alpha}}^0 -   \boldsymbol{\theta}_{ \boldsymbol{\alpha}'}^0 \|_1\big\} \geq \delta_1$ implies that  there is sufficient information to distinguish any two different classes $\boldsymbol{\alpha}$ and $\boldsymbol{\alpha}'$, thus ensuring the completeness \citep{Chiu} and identifiability conditions \citep{gu2018partial}.
It is also similar to Condition (b) in \cite{wang2015consistency}:
\medskip

\noindent { \textbf{Condition(b).} \textit{Define the set $\mathcal{A}_{m, m'}=\{j \mid \eta_{mj}\neq \eta_{m'j}\}$, where $m$ and $m'$ index the attribute profiles of different proficiency classes among all the $M=2^K$ realizable proficiency classes; $\text{Card}(\mathcal{A}_{m, m'})\rightarrow \infty$ as $J\rightarrow \infty$.}}

\noindent  {Condition (b) in \cite{wang2015consistency} and Assumption \ref{assump-delta-theta} in our work are essentially stating that for any two different proficiency classes, there are infinitely many items such that the item response functions for these two proficiency classes are different.
}

The condition \eqref{eq-theta} in Assumption \ref{assump-delta-loss} also holds for the aforementioned loss functions in Section \ref{sec-general}.
For example, it is easy to check the condition (\ref{eq-theta}) for the $L_2$ loss and the cross-entropy loss.
For the $L_2$ loss, we have
$
     E\big[l(x_{ij},\theta_{j,\boldsymbol{\alpha}}^0)\big] - E\big[l(x_{ij},\theta_{j,\boldsymbol{\alpha}_i^0}^0)\big]
=  (\theta_{j,\boldsymbol{\alpha}}^0 -\theta_{j,\boldsymbol{\alpha}_i^0}^0)^2.
$
For the cross-entropy loss, we have
\allowdisplaybreaks
\begin{align*}
    & E\big[l(x_{ij},\theta_{j,\boldsymbol{\alpha}}^0)\big] - E\big[l(x_{ij},\theta_{j,\boldsymbol{\alpha}_i^0}^0)\big]  \\
    = &
    -\theta_{j,\boldsymbol{\alpha}_0}^0 \log (\theta_{j,\boldsymbol{\alpha}}^0) - (1-\theta_{j,\boldsymbol{\alpha}_i^0}^0) \log(1 -\theta_{j,\boldsymbol{\alpha}}^0) + \theta_{j,\boldsymbol{\alpha}_i^0}^0 \log (\theta_{j,\boldsymbol{\alpha}_i^0}^0) + (1-\theta_{j,\boldsymbol{\alpha}_i^0}^0) \log (1 - \theta_{j,\boldsymbol{\alpha}_i^0}^0) \\
    = &\  D_{KL}\Big(p(\theta_{j,\boldsymbol{\alpha}}^0)\ \big|\big|\ p(\theta_{j,\boldsymbol{\alpha}_i^0}^0)\Big)
    \geq   \ \frac{1}{2}\Big(\big|\theta_{j,\boldsymbol{\alpha}}^0 - \theta_{j,\boldsymbol{\alpha}_i^0}^0 \big| + \big|(1-\theta_{j,\boldsymbol{\alpha}}^0) - (1-\theta_{j,\boldsymbol{\alpha}_i^0}^0)\big|\Big)^2\\
    = &\ 2(\theta_{j,\boldsymbol{\alpha}}^0 - \theta_{j,\boldsymbol{\alpha}_i^0}^0)^2,
\end{align*}
where $D_{KL}(\cdot \mid\mid \cdot)$ is the Kullback-Leibler divergence, $p(\cdot)$  is the mass function of a Bernoulli distribution, and the last inequality follows from Theorem 1.3 in \cite{popescu2016bounds}.

Similar to the analysis of the joint maximum likelihood estimation in \cite{chiu2016joint}, we assume that there is a calibration dataset that would give a statistically consistent estimator of the calibration subjects' latent class membership $\hat{\boldsymbol{A}}_c$, in the sense that $\mathbb{P}(\hat{\boldsymbol{A}_c} \neq \boldsymbol{A}_c^0) \to 0$ as $J\to \infty$.
We use $N_c$ and $\boldsymbol{A}_c^0$ to denote the sample size   and    the true membership matrix of the calibration dataset, respectively.
Here the subscript $c$ denotes the calibration dataset.
Similar assumption is also made in the consistency theories of the GNPC method in \cite*{chiu2019consistency}.
In the next theorem, we show that the consistent membership estimator will give consistent estimators for the centroids of the latent classes.

\begin{theorem}
\label{thm-theta}
Suppose the data conforms to CDMs that can be expressed in terms of general latent class models, and Assumptions 1-3 hold.
Further assume that
$J \exp  \big( -N_c\epsilon  \big)\to 0$ as $J,N_c\rightarrow \infty$ for any $\epsilon>0$. 
If $\hat{\boldsymbol{A}}_c$ is a consistent estimator of $\boldsymbol{A}_c^0$, then $\hat{\boldsymbol{\mu}}$ is also   consistent for $\boldsymbol{\theta}^0$ as $J,N_c \xrightarrow{}\infty$, that is, $\|\hat{\boldsymbol{\mu}}-\boldsymbol{\theta}^0\|_\infty \xrightarrow{P} 0$ as $J,N_c \xrightarrow{}\infty$, where $\|\cdot\|_\infty$ is the supremum norm.
\end{theorem}

Theorem \ref{thm-theta} states that if we could get a consistent estimate of the calibration subjects' membership $\hat{\boldsymbol{A}}_c$, then the estimated centroids $\hat{\boldsymbol{\mu}}$ are also consistent for the true item  response probabilities $\boldsymbol{\theta}^0$ in a uniform sense that all item parameters can be uniformly consistently estimated.
The detailed proof is in online Appendix \ref{appendix-thm1}.
 {This result is similar to Lemmas 1 and 2 in \cite{chiu2019consistency} under the GNPC framework, and Theorem 2 in \cite{chiu2016joint} under the JMLE framework.}
Note that for the GNPC method, the centroids are weighted averages of the ideal responses from the DINA and DINO models.
As discussed in Example \ref{ex-GNPC}, if the DINA and DINO models have the same ideal responses (i.e., $\boldsymbol{\alpha} \succeq \boldsymbol{q}_j$ or $\boldsymbol{\alpha}  \odot \boldsymbol{q}_j={\mathbf 0}$), then the corresponding centroid will be fixed to be 0 or 1, which thus does not lead to a consistent estimation of the corresponding item response probability $\theta^0_{j,\boldsymbol{\alpha}}$; 
 however, note that for the nonparametric GNPC method, such a fixed centroid  does not necessarily lead to inconsistency of $\hat{\boldsymbol{\alpha}}$.
Here we allow all the centroid parameters to be free, and the consistency estimation is ensured as in Theorem~\ref{thm-theta}.

The next theorem shows that if we start with a consistent membership $\hat{\boldsymbol{A}}_c$ obtained from the calibration dataset,
and use the estimated centroids to classify the subjects, then the resulting classifications are also consistent for each subject.

\begin{theorem}
Suppose Assumptions 1--4 and the assumptions of Theorem \ref{thm-theta} hold.
If we start with a consistent $\hat{\boldsymbol{A}}_c$ obtained from a calibration dataset to estimate the centroid $\hat{\boldsymbol{\mu}}$,
then 
$\hat{\boldsymbol{\alpha}}_i$ obtained by Algorithm \ref{algo1} is also a consistent estimator of $\boldsymbol{\alpha}_i$ for each $i=1,\dots,N$.
\label{thm-alpha}
\end{theorem}

\noindent To establish the consistency in Theorem \ref{thm-alpha}, the following two lemmas are needed.

\begin{lemma}\label{lemma-1}
Suppose Assumptions in Theorem \ref{thm-alpha} hold.
For each subject $i$, the true attribute pattern  minimizes $E[l(\boldsymbol{x}_i,\hat{\boldsymbol{\mu}}_{\boldsymbol{\alpha}}) + h(\hat{\pi}_{\boldsymbol{\alpha}})]$ with probability approaching 1 as $J\xrightarrow{}\infty$, as in,
\begin{equation*}
    P \Big(
    \boldsymbol{\alpha}_i^0 = \underset{\boldsymbol{\alpha}}{\arg\min}
    \ E\big[ l(\boldsymbol{x}_i, \hat{\boldsymbol{\mu}}_{\boldsymbol{\alpha}}) +h(\hat{\pi}_{\boldsymbol{\alpha}})
    \big]
    \Big) \xrightarrow{} 1
    \quad \text{ as } J \xrightarrow{} \infty.
\end{equation*}
\end{lemma}

\begin{lemma}\label{lemma-2}
Suppose Assumptions in Theorem \ref{thm-alpha} hold,
then we have
\begin{equation*}
    P\Big(\underset{\boldsymbol{\alpha}}{\max} \Big|
    \frac{1}{J} \sum_{j=1}^J \big(
    l(x_{ij},\hat{\mu}_{j,\boldsymbol{\alpha}})
    - E[l(x_{ij},\theta_{j,\boldsymbol{\alpha}}^0)]
    \big) \Big| \geq \epsilon  \Big) \xrightarrow{} 0,
    \text{ as } J\xrightarrow{}\infty.
\end{equation*}
\end{lemma}

\noindent Lemma \ref{lemma-1} extends  Proposition 1 in \cite{wang2015consistency} and Lemma 3 in \cite{chiu2019consistency} to more general loss functions.
Lemma \ref{lemma-2} generalizes Proposition 3 in \cite{wang2015consistency} and Lemma 4 in \cite{chiu2019consistency}.
The detailed proofs of Lemma \ref{lemma-1}, Lemma \ref{lemma-2}, and Theorem \ref{thm-alpha} are given in online Appendices \ref{appendix-lemma1} -- \ref{appendix-thm2}.
Note that Theorem \ref{thm-alpha} only gives the consistency for each $\boldsymbol{\alpha}_i$; however, we can further establish uniform consistency for all $\boldsymbol{\alpha}_i$, $i = 1,\dots,N$, as shown in Theorem \ref{thm-uniform}.

\begin{theorem}
\label{thm-uniform}
Suppose all the assumptions of Theorem \ref{thm-alpha} hold.
Further assume that $N > J$,  {$N_c > J$} and for any $\epsilon > 0$,
 $N\exp({-J\epsilon})\xrightarrow{} 0$. 
 If we start with a consistent  $\hat{\boldsymbol{A}}_c$ obtained from a calibration dataset, then $\hat{\boldsymbol{\alpha}}_i$ obtained from Algorithm \ref{algo1} is uniformly consistent for $\boldsymbol{\alpha}_i$, for all $i=1,\dots,N$.
\end{theorem}

Uniform consistency has also been established for specific nonparametric methods,
such as Theorem 2 in \cite{wang2015consistency} and Theorem 2 in \cite{chiu2019consistency}.
Our uniformly consistent result in Theorem \ref{thm-uniform} can be regarded as their generalization.
Specifically, in \cite{wang2015consistency}, they showed the uniform consistency for the NPC method, where the loss function is taken to be $L_1$ loss and the centroids are fixed, to be the ideal responses of the DINA or DINO model.
In \cite{chiu2019consistency}, the authors generalize the uniform consistency for the NPC method to the GNPC method, where the loss function is $L_2$ loss and the centroids are weighted averages of ideal responses from the DINA and the DINO models.

The above analysis pertains the asymptotic properties of our framework.
For finite-sample situations, we have the following theoretical properties for the proposed iterative algorithms in Section \ref{sec-general},
which are established following the theory in \cite{celeux1992classification}.

\begin{proposition}
Any sequence $(\boldsymbol{A}^{(t)},\boldsymbol{\mu}^{(t)},\boldsymbol{\pi}^{(t)})$ obtained by Algorithm \ref{algo1} 
decreases the criterion (\ref{eq-C2}) and the sequence $L(\boldsymbol{A}^{(t)},\boldsymbol{\mu}^{(t)},\boldsymbol{\pi}^{(t)})$ converges to a stationary value.
Moreover, if for any fixed $\boldsymbol{A}$, the minima of the loss function $L(\boldsymbol{A} ,\boldsymbol{\mu},\boldsymbol{\pi})$ is well-defined, then the sequence $(\boldsymbol{A}^{(t)},\boldsymbol{\mu}^{(t)},\boldsymbol{\pi}^{(t)})$ also converges to a stationary point.
\label{prop-finite-1}
\end{proposition}{}

Proposition \ref{prop-finite-1} indicates that the update sequence $(\boldsymbol{A}^{(t)},\boldsymbol{\mu}^{(t)},\boldsymbol{\pi}^{(t)})$ from the proposed algorithm   converges to a stationary point of the proposed criterion \eqref{eq-C2} with finite samples.
Additionally, all the loss functions in the examples in Section \ref{sec-general} satisfy the condition that the minima is well-defined.
Now, consider a smoothed version of $L(\boldsymbol{A} ,\boldsymbol{\mu},\boldsymbol{\pi})$,
\begin{equation*}
    L(\boldsymbol{U},\boldsymbol{\mu},\boldsymbol{\pi}) = \sum_{\boldsymbol{\alpha}\in \{0,1\}^K} \sum_{i=1}^n u_{i\boldsymbol{\alpha}} \Big(l(\boldsymbol{x_i},\boldsymbol{\mu}_{\boldsymbol{\alpha}}) + h(\pi_{\boldsymbol{\alpha}})\Big),
\end{equation*}
where $\boldsymbol{U} = \{u_{im}\} \in [0,1]^{n\times 2^K}$ is a matrix with nonnegative entries and each column sums to one, which is called a standard classification matrix in \cite{celeux1992classification}.
Recall that $L(\boldsymbol{A} ,\boldsymbol{\mu},\boldsymbol{\pi}) = \sum_{\boldsymbol{\alpha}} \sum_{i \in  C_{\boldsymbol{\alpha}} (\boldsymbol{A} )} \Big(l(\boldsymbol{x}_i,\boldsymbol{\mu}_{\boldsymbol{\alpha}})+h(\pi_{\boldsymbol{\alpha}})\Big) =\sum_{\boldsymbol{\alpha}} \sum_{i=1}^n I(\boldsymbol{\alpha}_i = \boldsymbol{\alpha})\Big(l(\boldsymbol{x}_i,\boldsymbol{\mu}_{\boldsymbol{\alpha}})+h(\pi_{\boldsymbol{\alpha}})\Big)$.
Therefore, $L(\boldsymbol{U},\boldsymbol{\mu},\boldsymbol{\pi})$ can be regarded as a smoothed version, where the hard membership matrix $\boldsymbol{A}$ is replaced by $\boldsymbol{U}$.
Note that the minimum of $L(\boldsymbol{U},\boldsymbol{\mu},\boldsymbol{\pi})$ is attained when $\boldsymbol{U}$ is equal to some hard membership matrix $\boldsymbol{A}$.

\begin{proposition}
Assume that $L(\boldsymbol{U},\boldsymbol{\mu},\boldsymbol{\pi})$ has a local minimum at $(\boldsymbol{U}^*,\boldsymbol{\mu}^*,\boldsymbol{\pi}^*)$ and that the Hessian of $L(\boldsymbol{U},\boldsymbol{\mu},\boldsymbol{\pi})$ exists and is positive definite at $(\boldsymbol{U}^*,\boldsymbol{\mu}^*,\boldsymbol{\pi}^*)$.
Then there is a neighborhood of $(\boldsymbol{U}^*,\boldsymbol{\mu}^*,\boldsymbol{\pi}^*)$ such that starting with any $(\boldsymbol{U}^{(0)},\boldsymbol{\mu}^{(0)},\boldsymbol{\pi}^{(0)})$ in that neighborhood, the resulting sequence $(\boldsymbol{A}^{(t)},\boldsymbol{\mu}^{(t)},\boldsymbol{\pi}^{(t)})$ of the Algorithm \ref{algo1} 
converges to $(\boldsymbol{U}^*,\boldsymbol{\mu}^*,\boldsymbol{\pi}^*)$ at a linear rate.
\label{prop-finite-2}
\end{proposition}

Proposition \ref{prop-finite-2} states that if we start with a good initial value which is close enough to the optimal point, then the update sequence will also converge to the optimal point.
These two propositions give good finite-sample properties of our proposed estimation framework.
The detailed proofs of Proposition \ref{prop-finite-1} and Proposition \ref{prop-finite-2} are given in online Appendix \ref{appendix-prop3} and \ref{appendix-prop4}, respectively.

\section{Simulation Studies}
\label{sec-simulation}

We conducted comprehensive simulations under a variety of settings to compare the performance of different methods.
The methods compared were:
\begin{itemize}
    \item NPC: the baseline method, where the centers are the ideal responses from the DINA model, and the loss function is the $L_1$ loss;
    \item GNPC: the centers are weighted averages of the ideal responses from the DINA and DINO models, and the loss function is the $L_2$ loss;
    \item  GNPC + log penalty: add log penalties on the proportion parameters to the GNPC method, where the loss function is $L_2$ loss for the centroids plus the summation of the log functions of the proportion parameters;
    \item JMLE: the Joint Maximum Likelihood Estimate, where the centroid parameters are to be estimated, and the loss function is the negative log-likelihood;
    \item CMLE: the Classification Maximum Likelihood Estimate, where the centers and the loss function are the  same as JMLE but with an additional term of class proportions as specified in (\ref{eq-CML2});
    \item MMLE: the Marginal Maximum Likelihood Estimate obtained from the traditional EM algorithm under the DINA or GDINA model assumption.
\end{itemize}{}

\noindent MMLE, as one of, if not the the most commonly used estimation algorithm in the CDM literature, was included in the comparison to provide a more comprehensive understanding of how different CDM estimation methods perform.

For the underlying true models, we considered two different settings: all items conformed to the DINA, 
	or all items conformed to the GDINA model.
Following the simulation design in \cite{chiu2018cognitive}, the subjects' true latent attribute patterns were either drawn from a uniform distribution or derived from the multivariate normal threshold model.
More specifically, for the uniform setting, each latent pattern $\boldsymbol{\alpha}$ had the same probability $1/2^K$ of being drawn.
For the multivariate normal setting, each subject’s attribute profile was linked to a latent continuous ability vector $\boldsymbol{z} = (z_1,\dots,z_K)' \sim \mathcal{N}(\boldsymbol{0,\Sigma})$ with values along the main diagonal of $\boldsymbol{\Sigma}$ setting to 1 and the off-diagonal entries setting to either $r=0.4$ or $0.8$ for different levels of correlation.
The latent continuous ability vector $\boldsymbol{z}$ was randomly sampled, and the $k$th entry of the attribute pattern was determined by
    \[  \alpha_{ik} =
        \begin{cases}
            1, &  {z}_{ik} \geq \Phi^{-1}\big(\frac{k}{K+1}\big) \\
            0, & \text{otherwise.}
        \end{cases}
    \]
where $\Phi$ is the inverse cumulative distribution function of standard normal distribution.

We considered different numbers of latent attributes ($K=3$ or 5), different sample sizes ($N = 30$, $50$, $200$ or $500$) and different number of items ($J = 30$ or 50).
To ensure identifiability, we set the first two $K\times K$ submatrices of the $Q$-matrix to be identity matrices.
The remaining items were randomly generated from all the possible latent patterns.
When $K=5$, the $Q$-matrix contained items that measured up to three attributes and was constructed the same way as that for $K = 3$.
When the underlying model was the DINA or DINO model, different signal strengths were considered.
Specifically we set $s=g=0.1$ or 0.3.
{When the true model was the GDINA model, we also considered two different signal strength levels.
One had the same item   response functions as those specified in \cite{chiu2018cognitive}, which are listed in  Table \ref{theta-GDINA-1}; the other setting contained larger noise
as listed in  Table \ref{theta-GDINA-2}.}

\begin{table}[H]
\centering
\begin{tabular}{cccccccc}
\hline
$P(\boldsymbol{\alpha}_1)$ & $P(\boldsymbol{\alpha}_2)$ & $P(\boldsymbol{\alpha}_3)$ & $P(\boldsymbol{\alpha}_4)$ & $P(\boldsymbol{\alpha}_5)$ & $P(\boldsymbol{\alpha}_6)$ & $P(\boldsymbol{\alpha}_7)$ & $P(\boldsymbol{\alpha}_8)$ \\ \hline
0.2 & 0.9 &  &  &  &  &  &  \\
0.1 & 0.8 &  &  &  &  &  &  \\
0.1 & 0.9 &  &  &  &  &  &  \\
0.2 & 0.5 & 0.4 & 0.9 &  &  &  &  \\
0.1 & 0.3 & 0.5 & 0.9 &  &  &  &  \\
0.1 & 0.2 & 0.6 & 0.8 &  &  &  &  \\
0.1 & 0.2 & 0.3 & 0.4 & 0.4 & 0.5 & 0.7 & 0.9 \\
\hline
\end{tabular}
\caption{Item response parameters for GDINA with small noise}
\label{theta-GDINA-1}
\end{table}

\begin{table}[H]
\centering
\begin{tabular}{cccccccc}
\hline
$P(\boldsymbol{\alpha}_1)$ & $P(\boldsymbol{\alpha}_2)$ & $P(\boldsymbol{\alpha}_3)$ & $P(\boldsymbol{\alpha}_4)$ & $P(\boldsymbol{\alpha}_5)$ & $P(\boldsymbol{\alpha}_6)$ & $P(\boldsymbol{\alpha}_7)$ & $P(\boldsymbol{\alpha}_8)$ \\ \hline
0.3 & 0.7 &  &  &  &  &  &  \\
0.3 & 0.8 &  &  &  &  &  &  \\
0.3 & 0.4 & 0.7 & 0.8 &  &  &  &  \\
0.3 & 0.4 & 0.6 & 0.7 &  &  &  &  \\
0.2 & 0.3 & 0.6 & 0.7 &  &  &  &  \\
0.2 & 0.3 & 0.3 & 0.4 & 0.4 & 0.5 & 0.6 & 0.7 \\
\hline
\end{tabular}
\caption{Item response parameters for GDINA with large noise}
\label{theta-GDINA-2}
\end{table}

To evaluate the performance of different methods, two metrics were used: the pattern-wise agreement rate (PAR) and the attribute-wise agreement rate (AAR), as defined below,
    \[
    \text{PAR} = \frac{\sum_{i=1}^N I\{\hat{\boldsymbol{\alpha}}_i = \boldsymbol{\alpha}_i\}}{N},\quad
    \text{AAR} = \frac{\sum_{i=1}^N \sum_{k=1}^K I\{\hat{\alpha}_{ik} = \alpha_{ik}\}}{NK}.
    \]
 
 For parametric methods including JMLE, CMLE, MMLE of the DINA and the GDINA models, we also calculate the Mean Squared Errors (MSEs) for item   response probabilities ${\theta}$'s  of each latent class.
For each setting, we repeated 100 times and reported the obtained means of PAR, AAR and MSE.
 Note that the aforementioned methods are iterative in nature, hence, would be affected by how they are initialized. For comparability purposes, we treated the NPC method as the baseline in this work, and used its results to initialize all the other methods. Using the NPC to perform the initialization is a reasonable choice given its non-iterative nature.
In the following result plots, we use DINA or GDINA to stand for the results of MMLE obtained by the EM algorithm under the corresponding model assumptions.

\subsection*{Result I: DINA}

Figures \ref{fig:DINA-PAR-zoom} and \ref{fig:DINA-AAR-zoom} present the PARs and  AARs when the underlying process followed the DINA model.
Under the independent attribute (i.e., uniform) setting, the NPC performed the best, as expected, in almost all the cases.
 The JMLE performed similarly to the CMLE in most cases, and slightly better than CMLE when there were more latent attributes ($K=5$) -- this was so because the JMLE method correctly assumed that the true latent patterns were uniformly distributed.
The GNPC produced similar results to the JMLE and the CMLE in most cases, but much worse results with large noises ($s = g = 0.3$) and more items ($J=50$).
Adding the log penalty to the GNPC method degraded the results under the uniform setting especially when the sample size was large, which is also expected since the GNPC method implicitly assumes a uniform penalty on the latent classes.
In comparison, the MMLE of the DINA and GDINA models did not perform as well as the others.
This was particularly true when the noise was large and sample size was small.

Under the dependent attribute (i.e., multivariate normal) settings, although the NPC still performed the best in almost all the cases with moderate correlation ($r=0.4$), it performed poorly with larger correlation ($r=0.8$) and sample size ($N = 200$ or 500) as a consequence of more unequal latent patterns proportions.
The MMLE of the DINA provided the best results when the sample size was larger ($N=200/500$) and the correlation was large ($r$ = 0.8), but did  not perform well with smaller sample sizes.
The GNPC and JMLE performed similarly when the noise was small, but the GNPC was much worse than the JMLE when the noise was large.
Adding the log penalty on the proportions improved the performance of the GNPC method under the correlated settings, though still not as good as the CMLE method.
In contrast, the CMLE performed uniformly well in almost all cases, and its advantages became more apparent when there were more latent attributes, and the correlation and the noise were large.
Specifically, the CMLE performed similarly to the NPC when the sample sizes were small, and the MMLE of the DINA when the sample sizes were large.
In almost all of the conditions, the MMLE of the GDINA did not perform as well as the other methods, which was not unexpected as the DINA was the true model.

Mean Squared Errors (MSEs) for the item  response probabilities ${\theta}$'s using parametric methods including JMLE, CMLE and MMLE for the DINA and GDINA models are plotted in Figure \ref{fig:DINA-MSE}.
From the results, we can see that across different settings the MMLE for the DINA model gave the best item response probability estimates, while the MMLE for the GDINA model performed the worst.
The JMLE and CMLE methods provided similar results. 
It is actually not surprising that the MMLE for the DINA performed the best for item  response probability estimation since it correctly assumed a two-level DINA model and  directly estimated the corresponding guessing and slipping parameters, while other methods did not have such prior knowledge about the underlying model.

\subsection*{Result II: GDINA}

 Figures \ref{fig:GDINA-PAR-zoom} and \ref{fig:GDINA-AAR-zoom} show the PARs and the AARs when the data conformed to the GDINA model under different settings.
Based on the results, when the latent attributes were independent, the GNPC performed generally the best across the settings, whereas the JMLE ,the CMLE and the MMLE of the GDINA model improved with increasing sample size.
As in the DINA cases, the log penalty on the proportions degraded the performance of the GNPC method under the independent setting.
The JMLE provided comparable or slightly better results than the CMLE, particularly when $K$ was larger.
As mentioned earlier, this is because the JMLE correctly assumed a uniform prior distribution for the latent attributes, whereas the CMLE, although made no assumptions, needed to estimate additional parameters.

Under the correlated latent attributes settings, adding the log penalties on the proportions to the GNPC method greatly improved the performance especially when the noise was large or correlation was high.
GNPC+log penalty tended to provide the best results with small sample sizes, and the CMLE and the MMLE of the GDINA gave the best results with larger samples.
Moreover, with larger noises, the CMLE method provided better results than the MMLE of the GDINA model particularly when there were more latent attributes.
As the correlation became larger, with large sample sizes, the performance of the CMLE method became more similar to that of the MMLE of the GDINA, and better than the JMLE method, due to the proportions of latent attribute patterns no longer being equal.
Based on the above analysis, one can note that the CMLE method was more robust to large noise.

The MSEs for the parametric methods including JMLE, CMLE and MMLE for the GDINA model under the GDINA settings are given in Figure \ref{fig:GDINA-MSE}.
	These three methods gave similar results in most cases, while JMLE and CMLE performed better than the MMLE of the GDINA settings especially when the number of attributes was large or noises were large.

\subsection*{Summary and Recommendations}

\label{sec-simulation-summary}

{Based on the above simulation results, we can see that there is no dominating method that performed uniformly better than other methods across all the simulation settings.
Hence, the choice of the method should be based on the specific setting and other information we have at hand.
In the following, we provide recommendations in practice under different circumstances.}

{If we can safely assumed that the true underlying model is the DINA model, then the NPC method would give good results if the latent attributes are independent.
When the latent attributes are moderately correlated, either the NPC or the CMLE method is recommended.
When the correlations are high among the latent attributes, the NPC and the CMLE would perform well with small sample sizes, whereas the CMLE and the MMLE of the DINA model would give better results with sufficiently large data sizes.}

 In situations where the true model is the GDINA model, the GNPC method will perform generally well if the latent attributes are independent.
When the latent attributes are correlated and the sample size is small, the GNPC augmented by the log penalties on the proportion parameters is preferred.
However, when the sample is sufficiently large, the CMLE method is more robust.
The CMLE method also performs well with small sample sizes when the noise is large.

Finally, if   the true data-generating models are unknown, the CMLE method is recommended when the latent attributes are correlated.
If the latent attributes are independent, the GNPC method is preferred.
Moreover, when the sample size is large enough, the MMLE method for the GDINA model is also recommended.
If the noise is small, the GNPC method will also perform well when the sample size is small,  and augmenting the GNPC method with the log penalties on the proportion parameters will improve its performance under the correlated setting.

\begin{landscape}

\begin{figure}[H]
    \centering
    \subfigure{
    \includegraphics[width=2.56in]{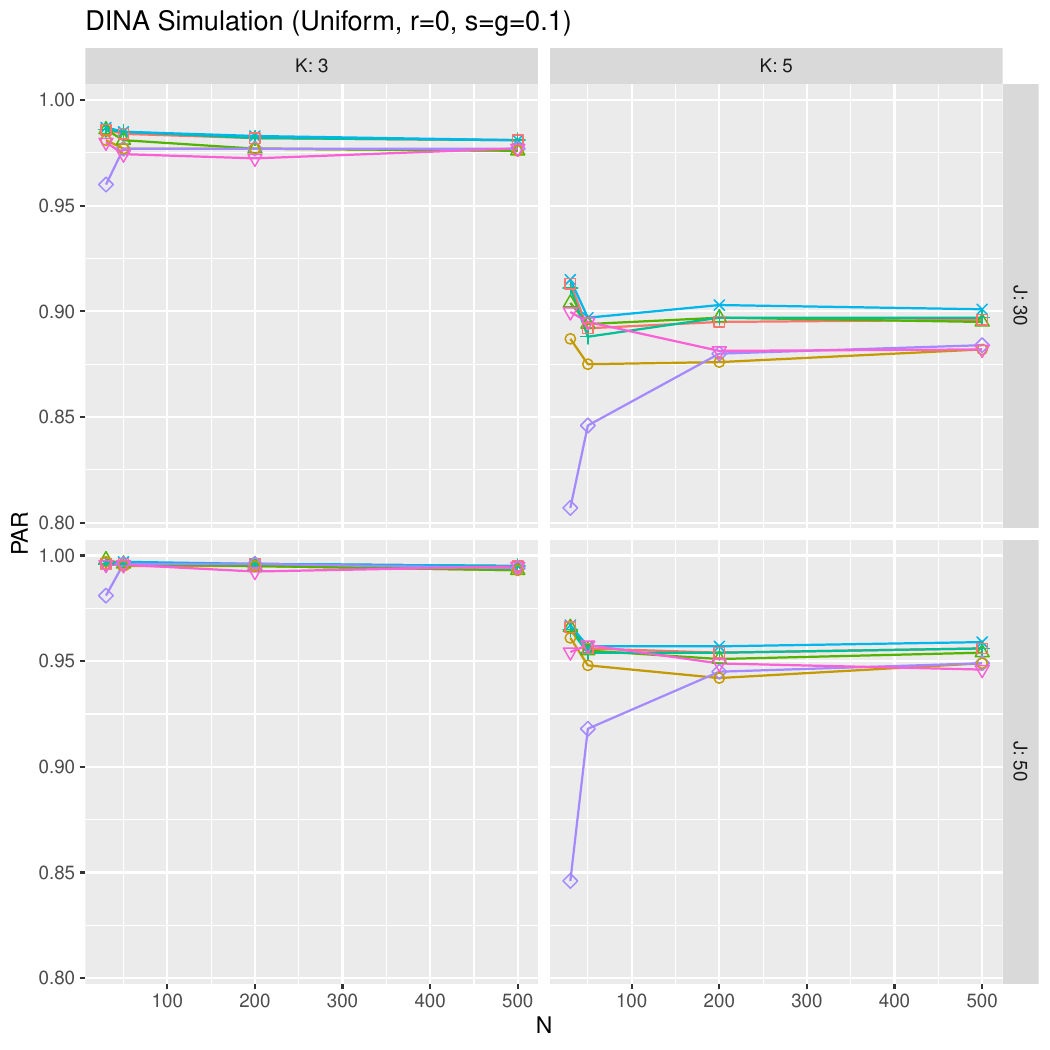}
    }
    \subfigure{
    \includegraphics[width=2.56in]{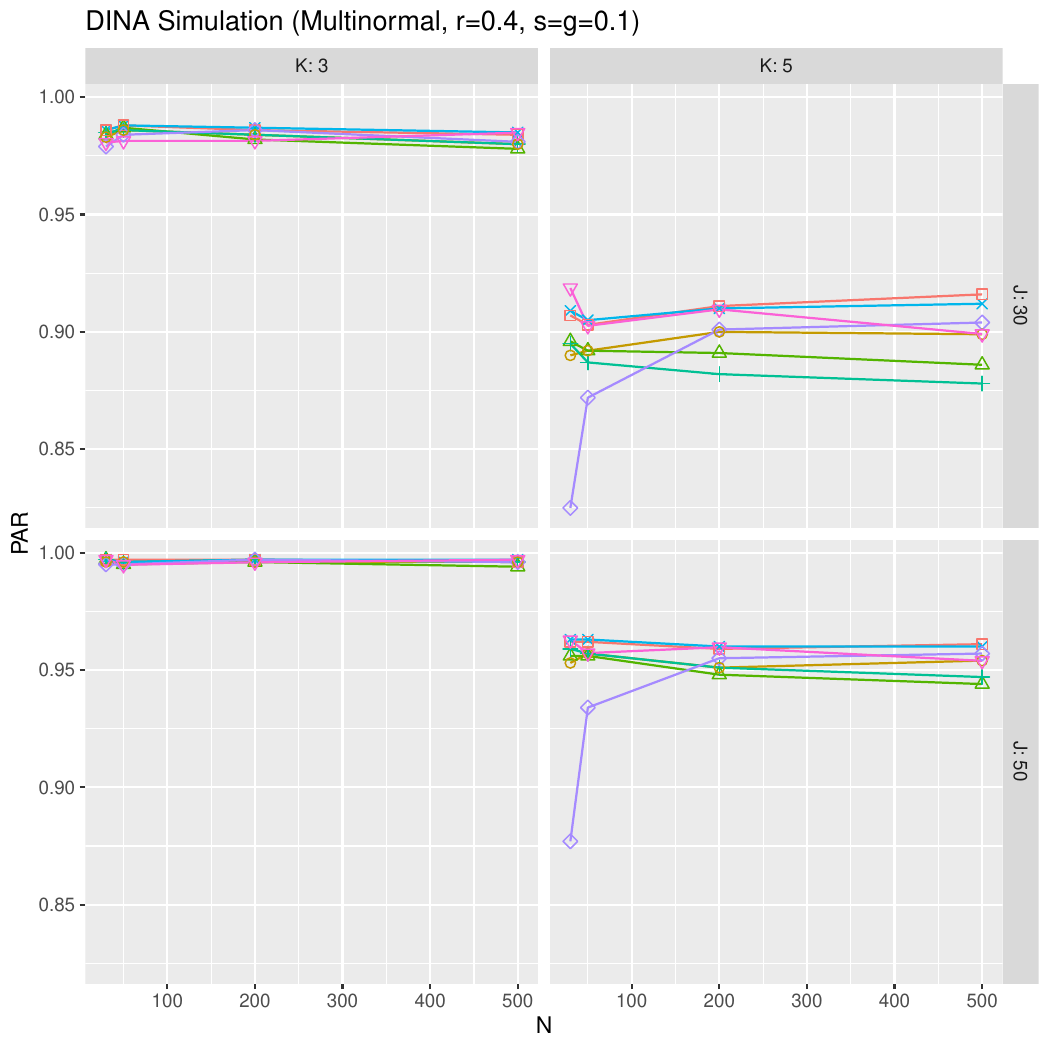}
    }\subfigure{
    \includegraphics[width=2.56in]{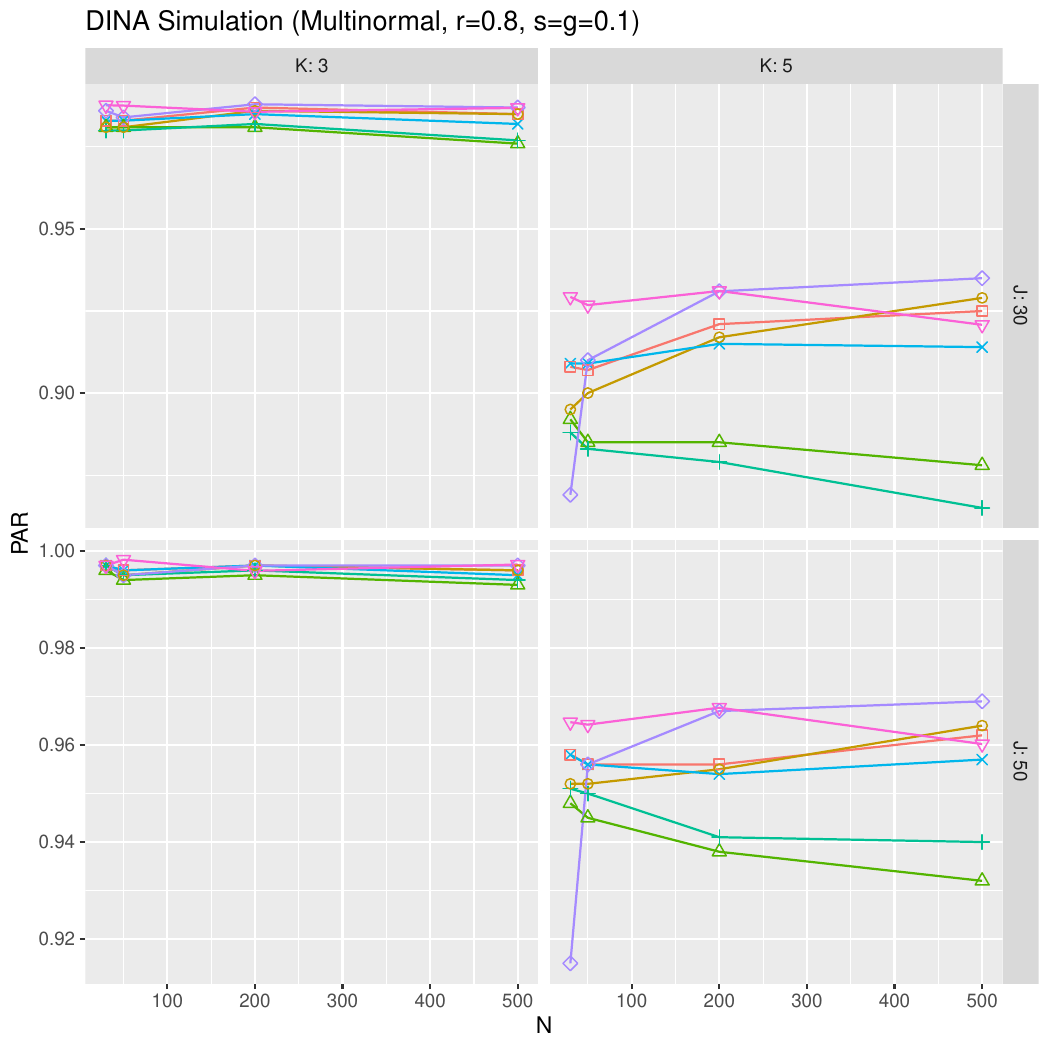}
    }
    \subfigure{
    \includegraphics[width=0.45in]{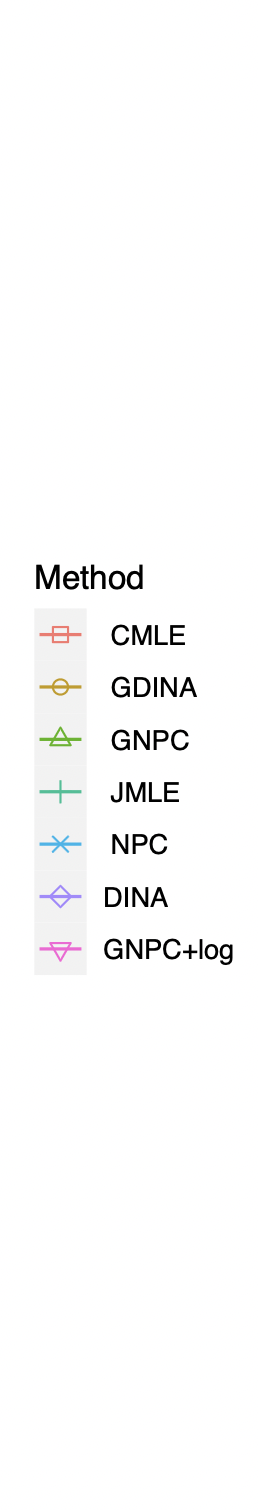}
    }\\
    \subfigure{
    \includegraphics[width=2.56in]{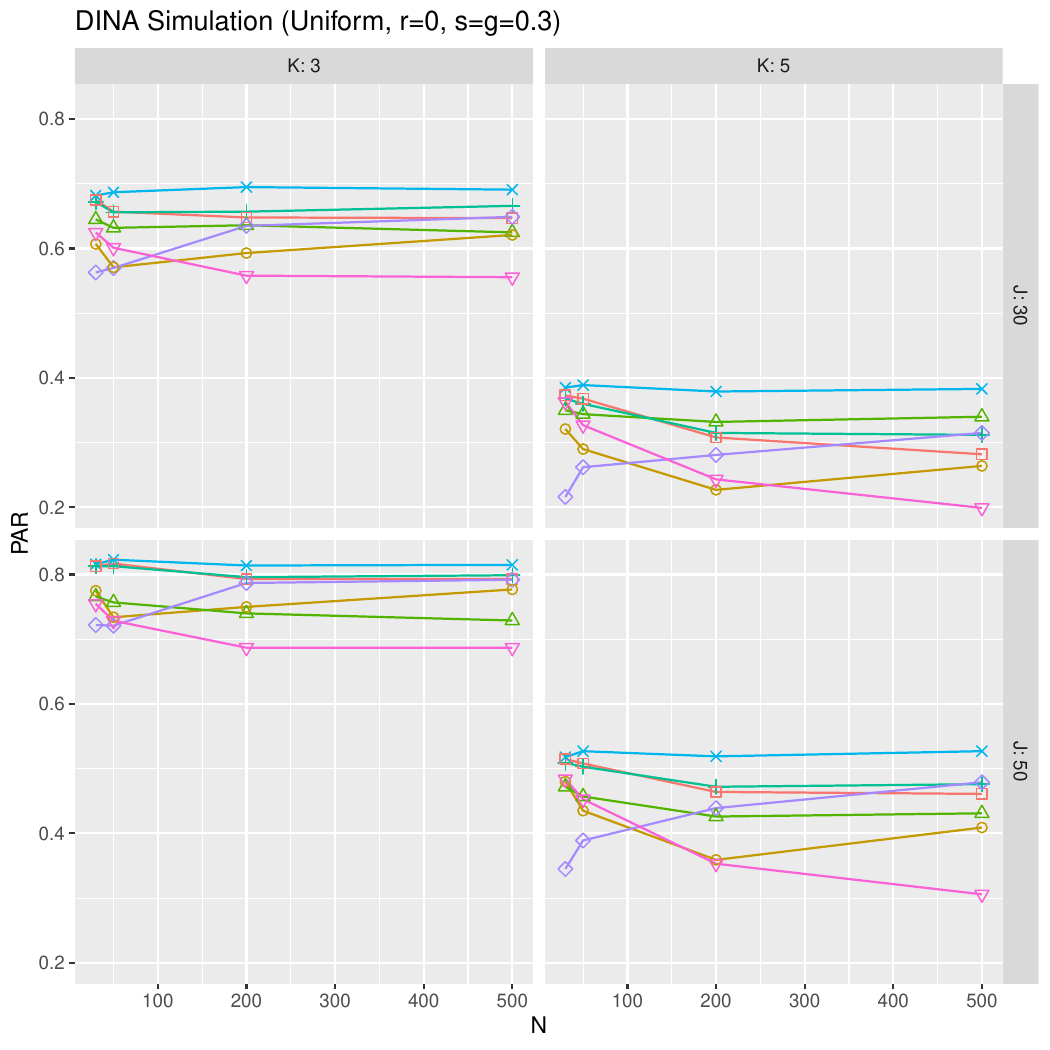}
    }
    \subfigure{
    \includegraphics[width=2.56in]{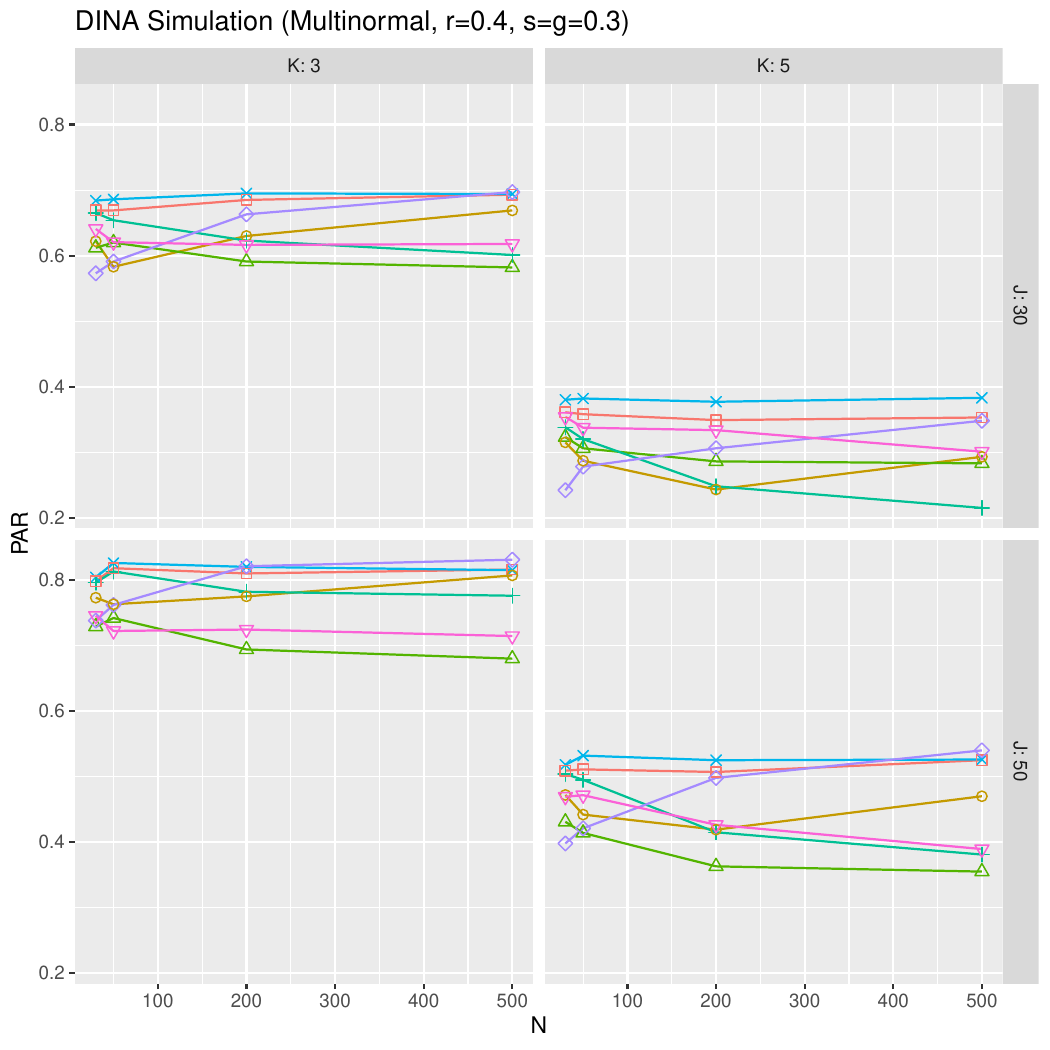}
    }
    \subfigure{
	\includegraphics[width=2.56in]{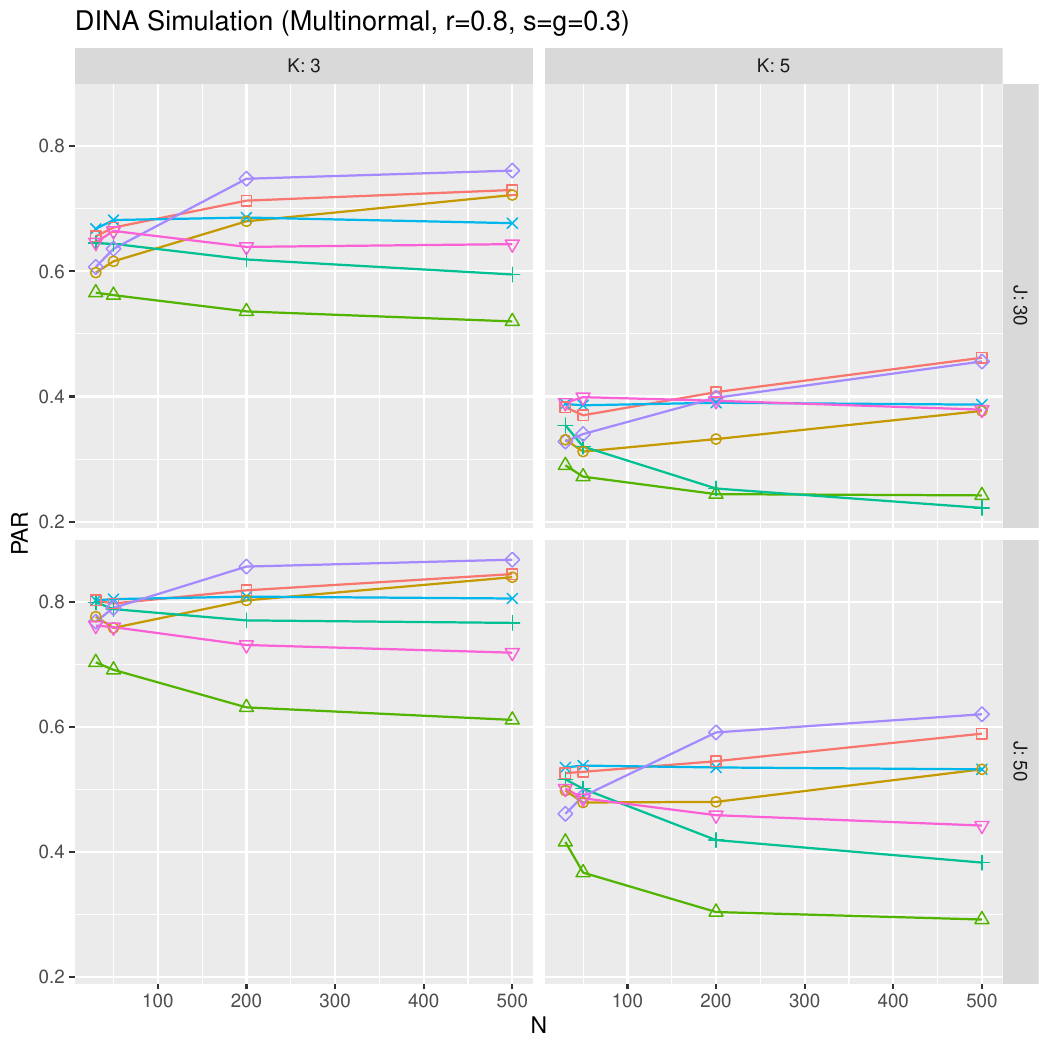}
    }
    \subfigure{
    \includegraphics[width=0.45in]{plots/PAR_legend.png}
    }
    \caption{PARs when the data conformed to the DINA model}
    \label{fig:DINA-PAR-zoom}
\end{figure}{}
	
\end{landscape}

\begin{landscape}
	\begin{figure}[H]
    \centering
    \subfigure{
    \includegraphics[width=2.56in]{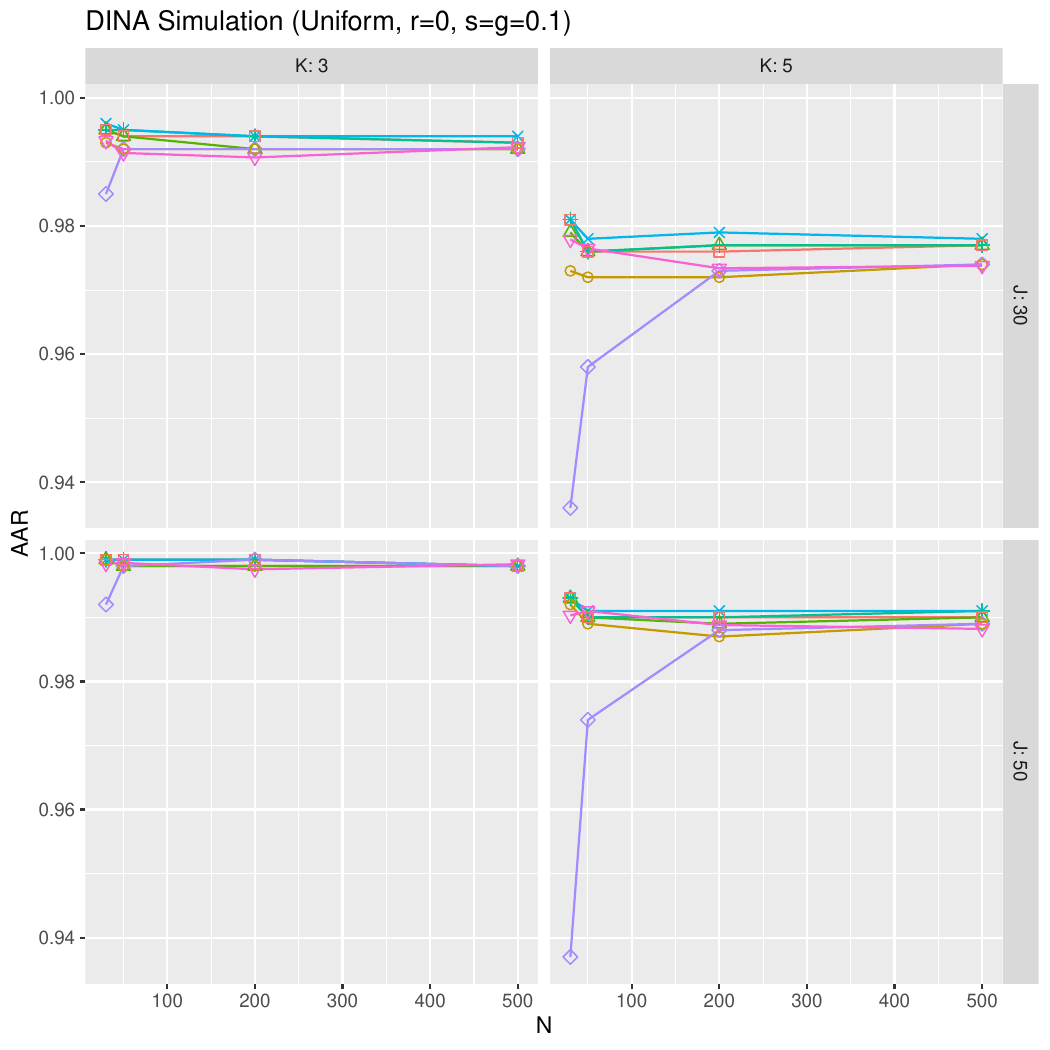}
    }
    \subfigure{
    \includegraphics[width=2.56in]{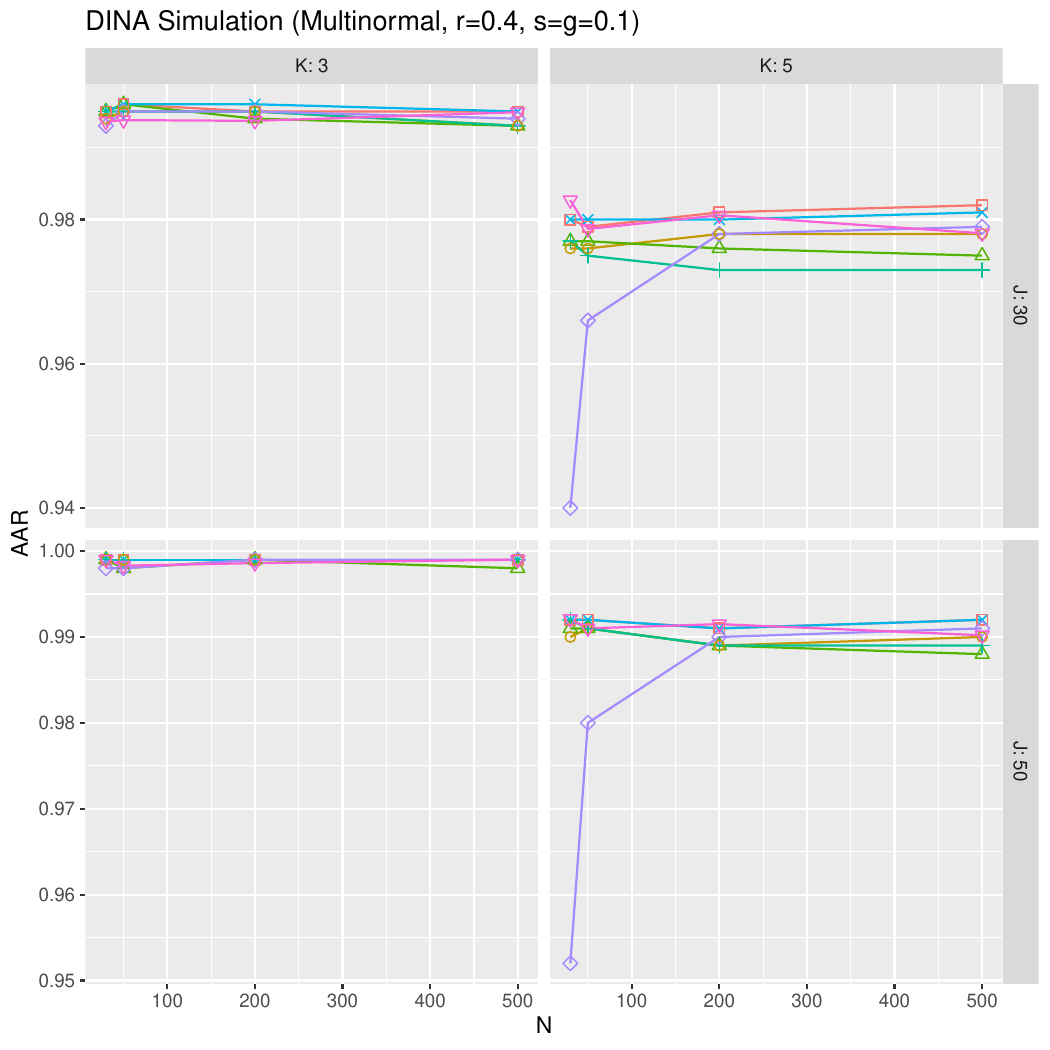}
    }
    \subfigure{
    \includegraphics[width=2.56in]{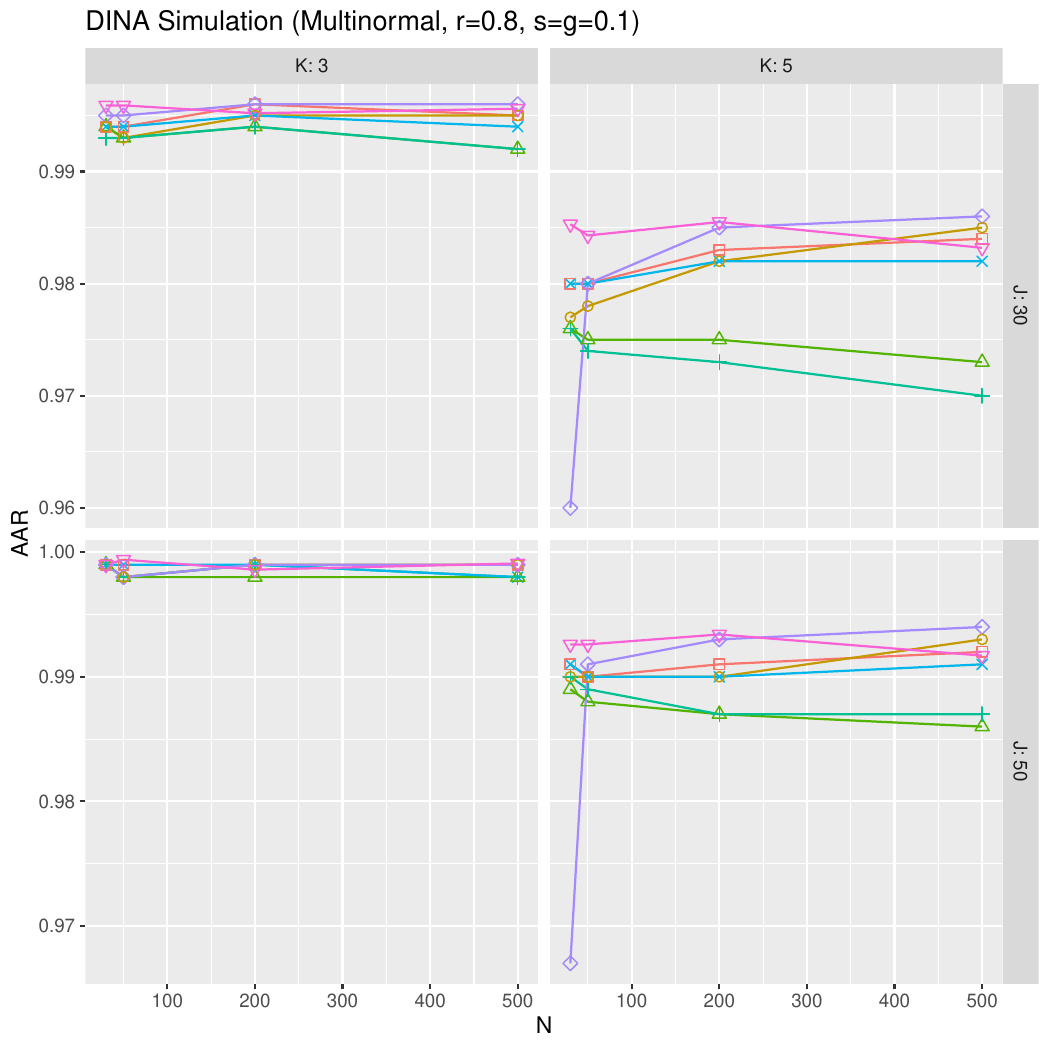}
    }
    \subfigure{
    \includegraphics[width=0.45in]{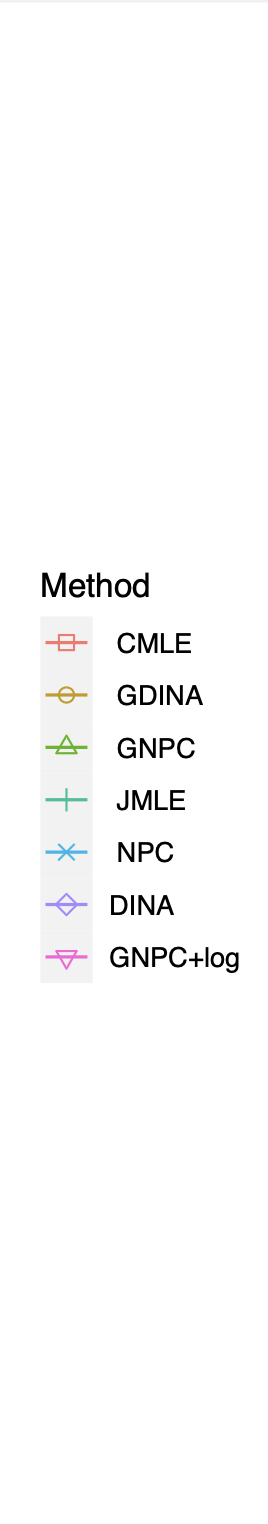}
    }
	\\
    \subfigure{
    \includegraphics[width=2.56in]{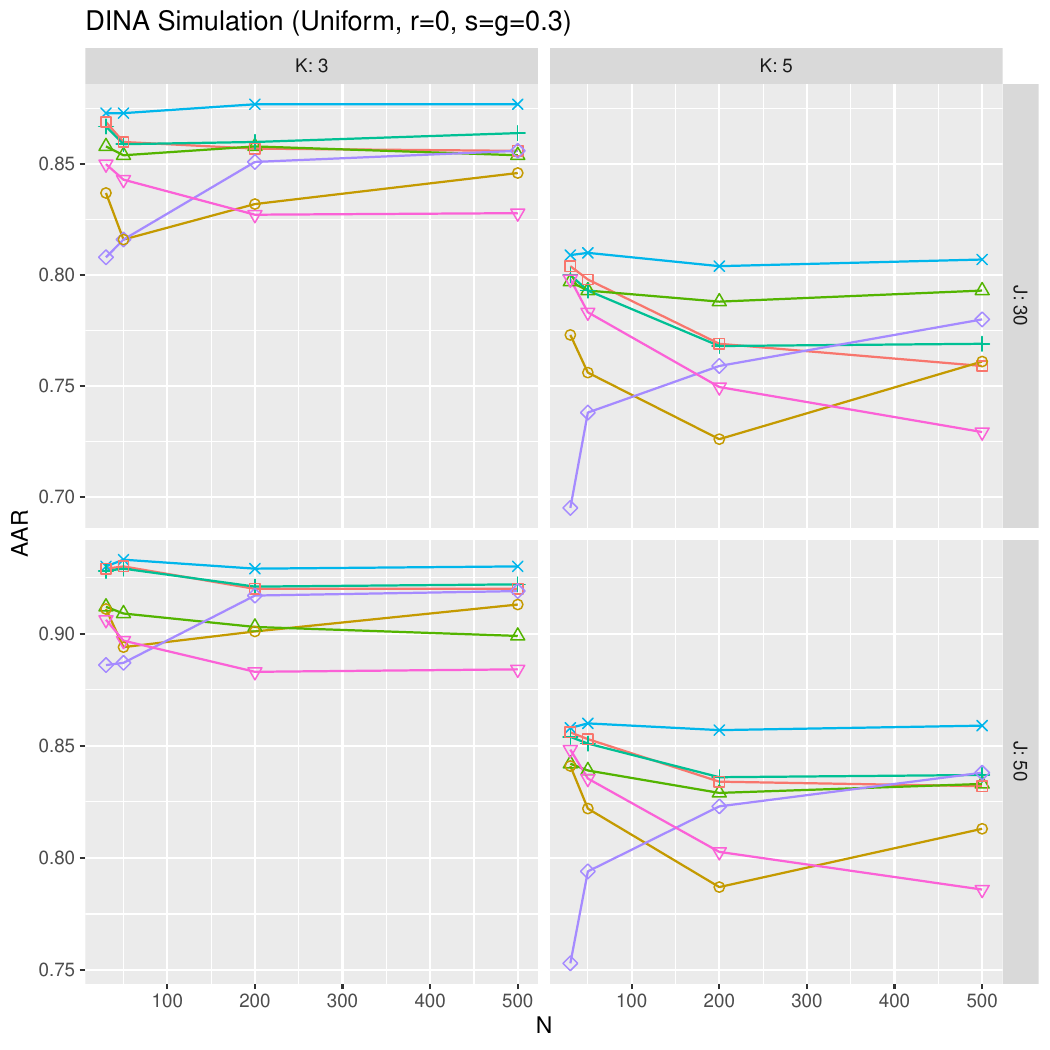}
    }
    \subfigure{
    \includegraphics[width=2.56in]{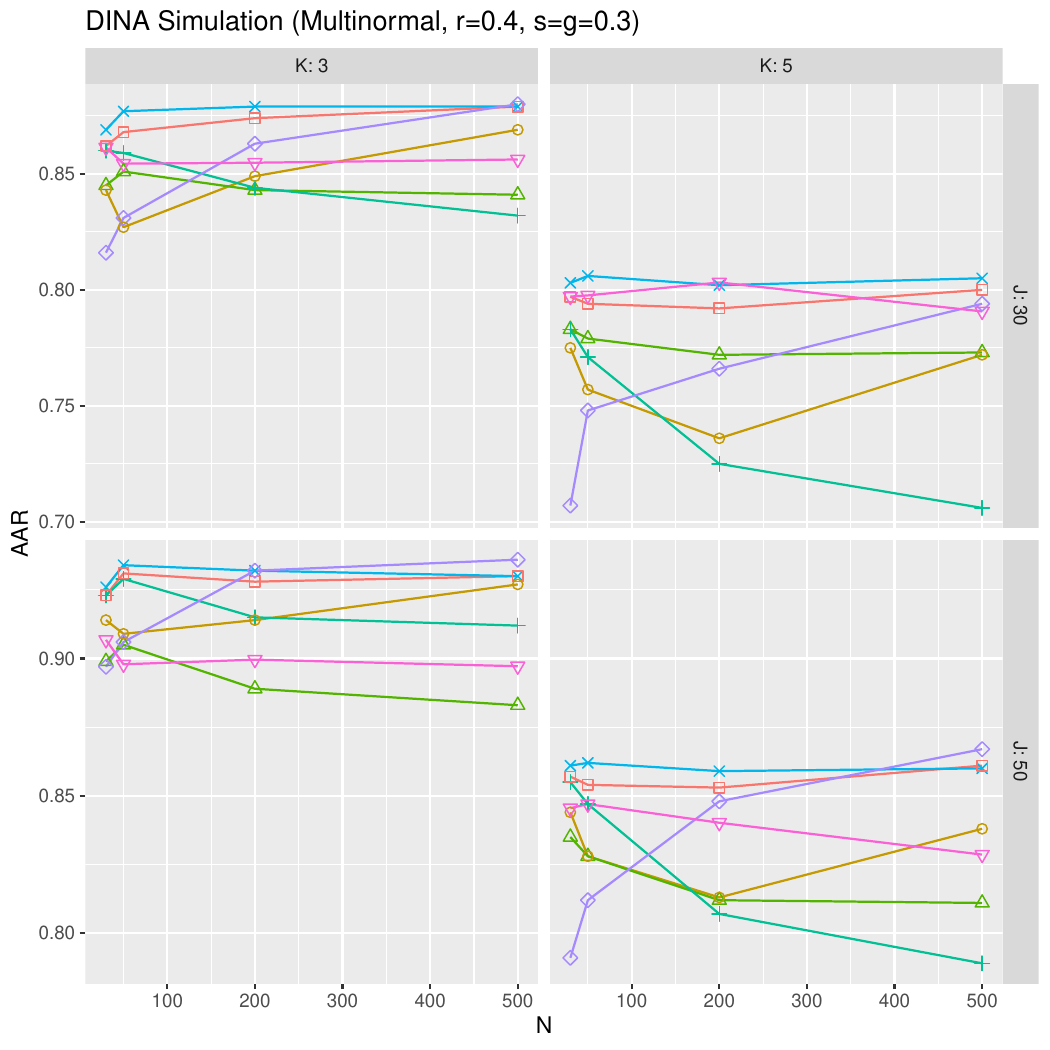}
    }
    \subfigure{
    \includegraphics[width=2.56in]{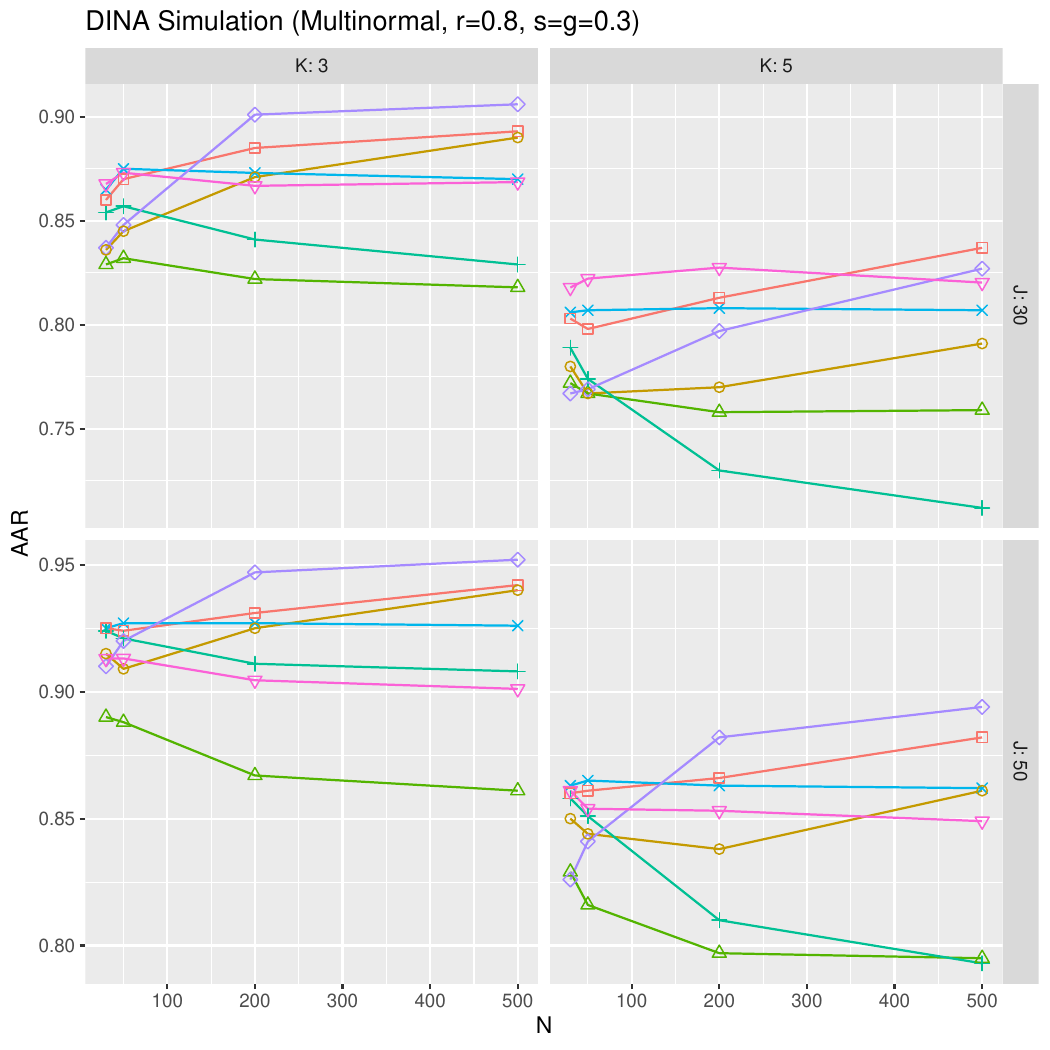}
    }
    \subfigure{
    \includegraphics[width=0.45in]{plots/AAR_legend.png}
    }
    \caption{AARs when the data conformed to the DINA model}
    \label{fig:DINA-AAR-zoom}
\end{figure}{}

\begin{figure}[H]
    \centering
    \subfigure{
    \includegraphics[width=2.56in]{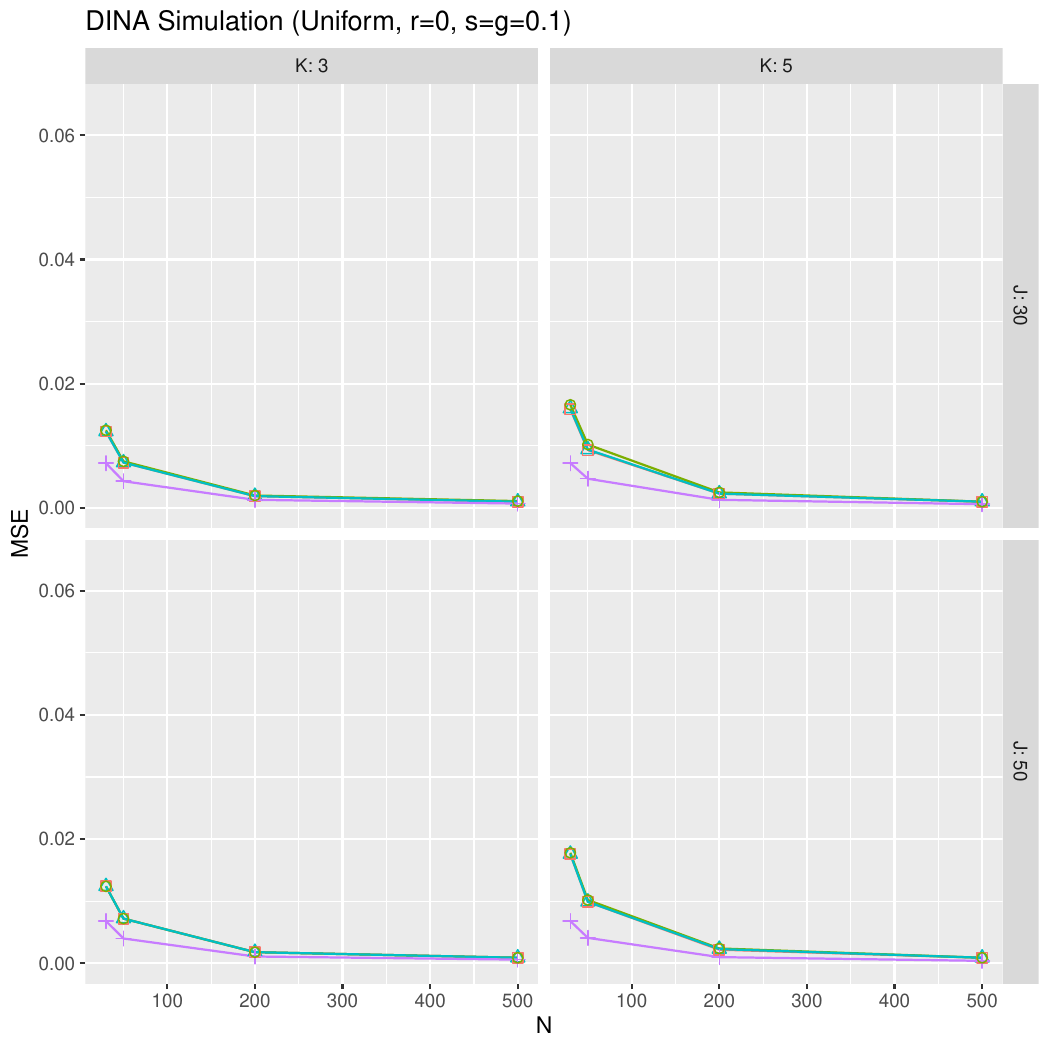}
    }
    \subfigure{
    \includegraphics[width=2.56in]{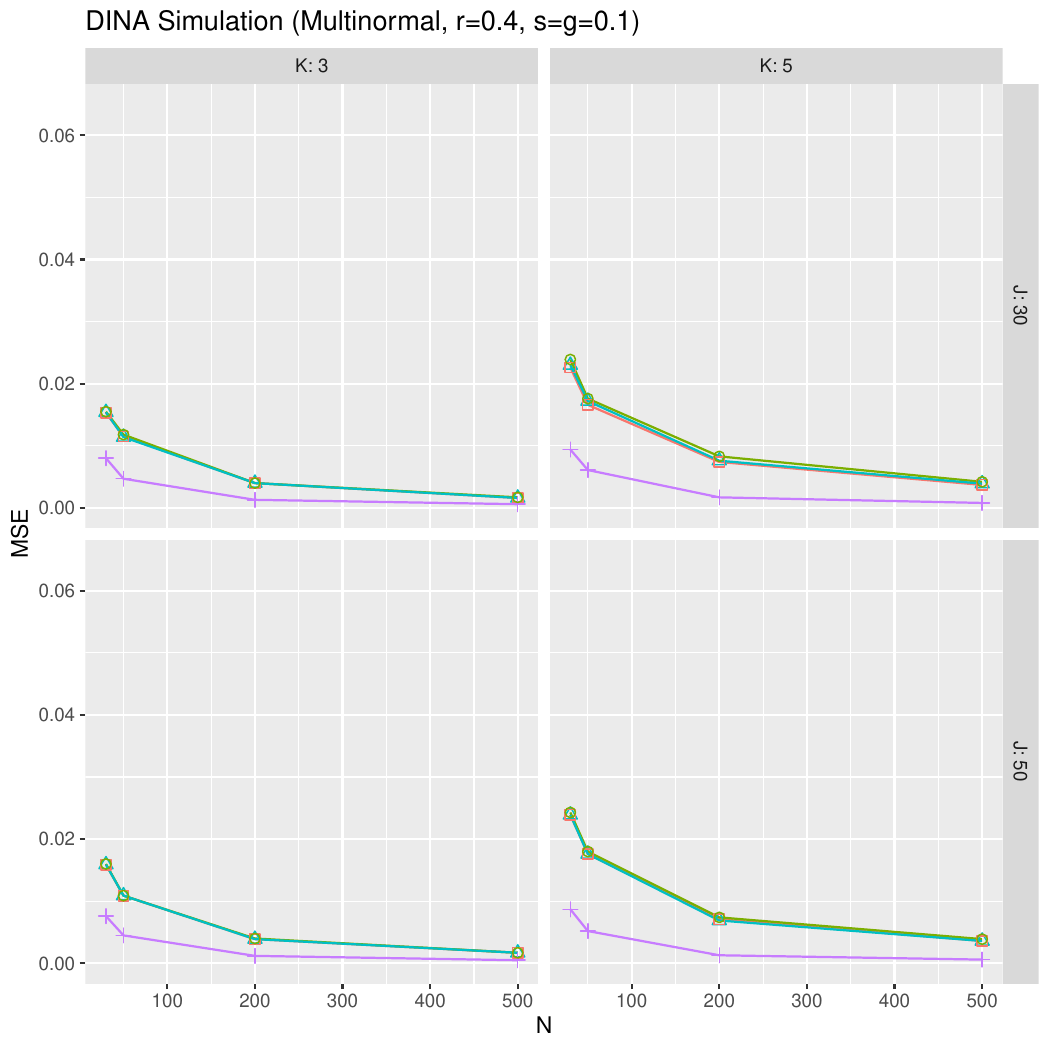}
    }
    \subfigure{
    \includegraphics[width=2.56in]{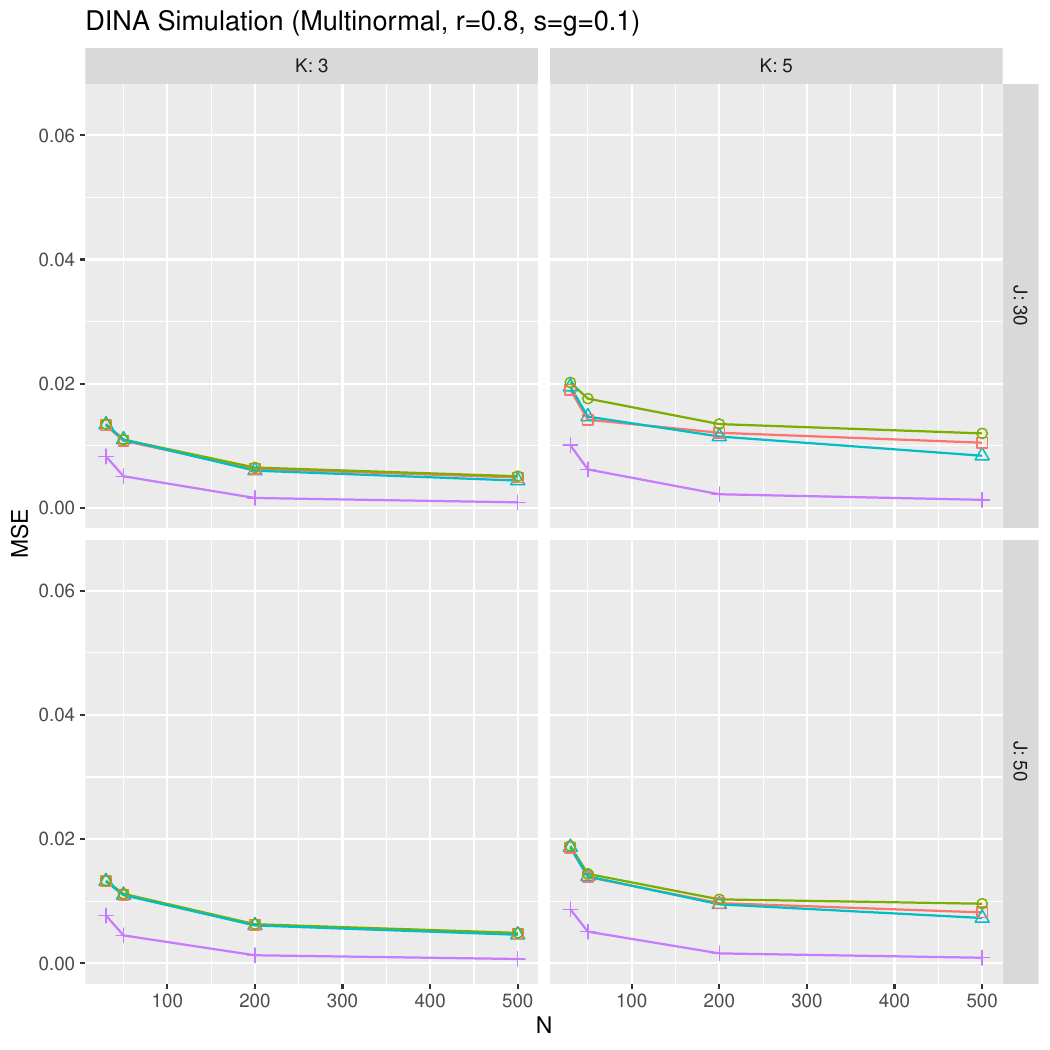}
    }
    \subfigure{
    \includegraphics[width=0.45in]{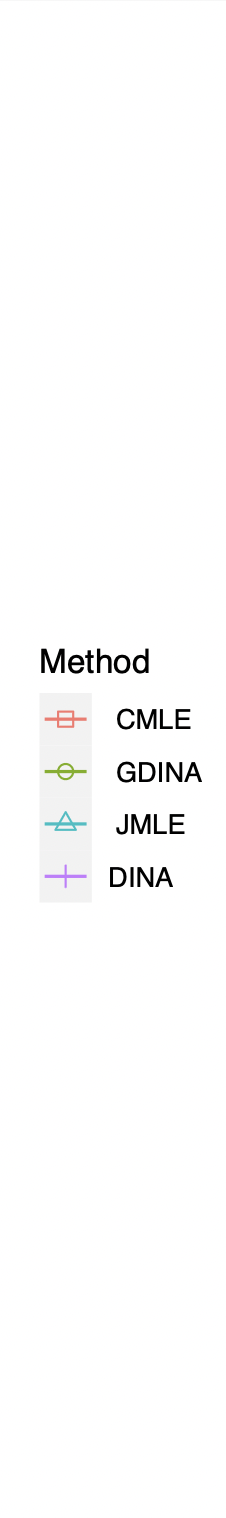}
    }
	\\
    \subfigure{
    \includegraphics[width=2.56in]{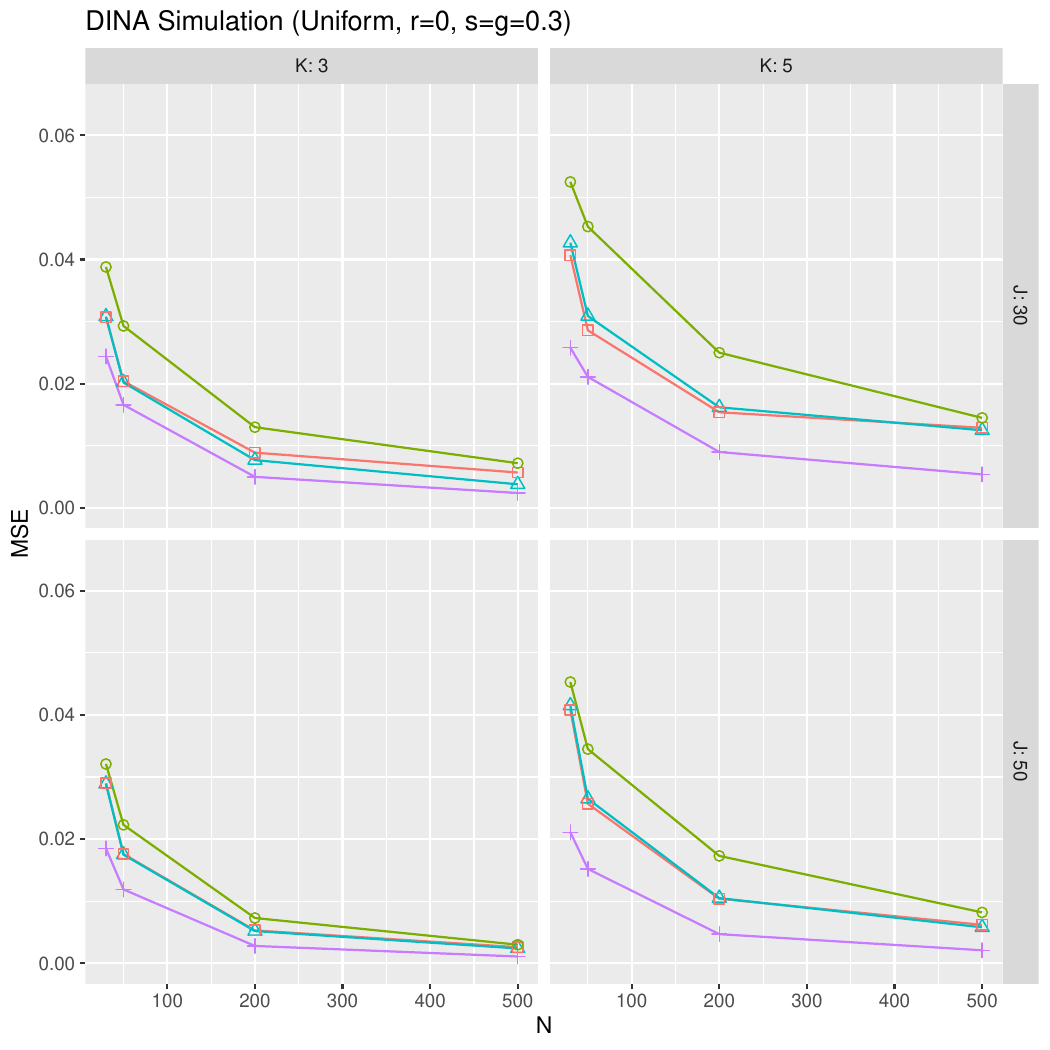}
    }
    \subfigure{
    \includegraphics[width=2.56in]{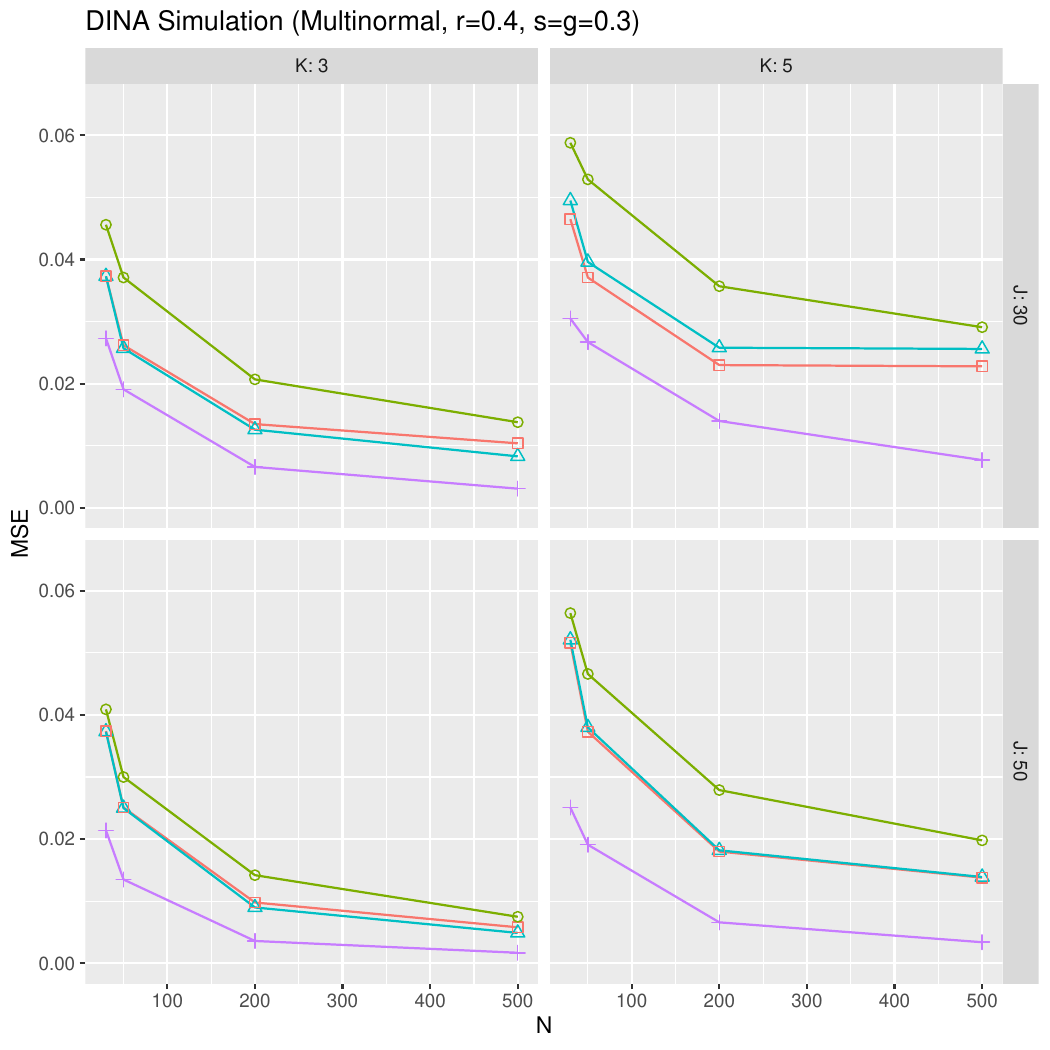}
    }
    \subfigure{
    \includegraphics[width=2.56in]{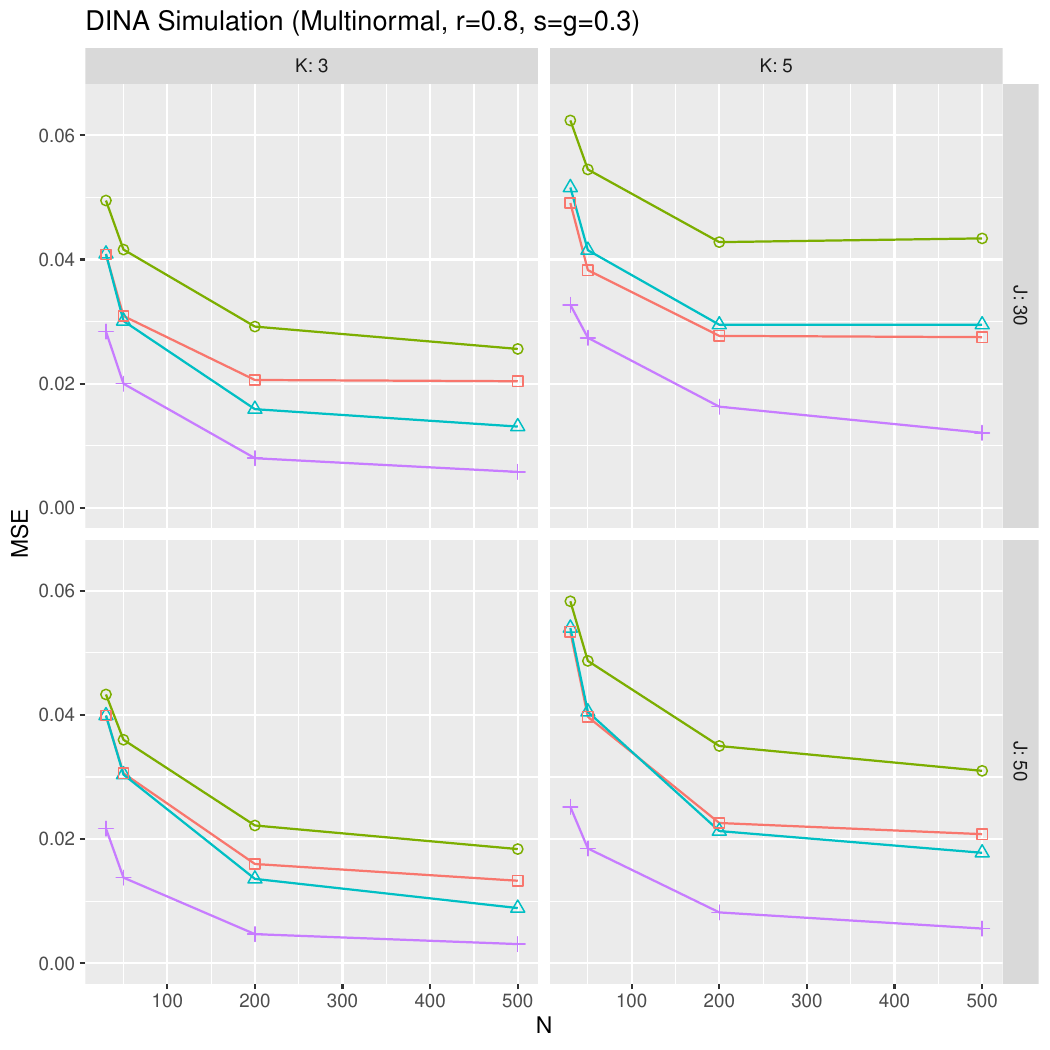}
    }
    \subfigure{
    \includegraphics[width=0.45in]{plots/MSE_legend.png}
    }
    \caption{MSE when the data conformed to the DINA model}
    \label{fig:DINA-MSE}
\end{figure}{}

\end{landscape}

\begin{landscape}
	\begin{figure}[H]
    \centering
    \subfigure{
    \includegraphics[width=2.56in]{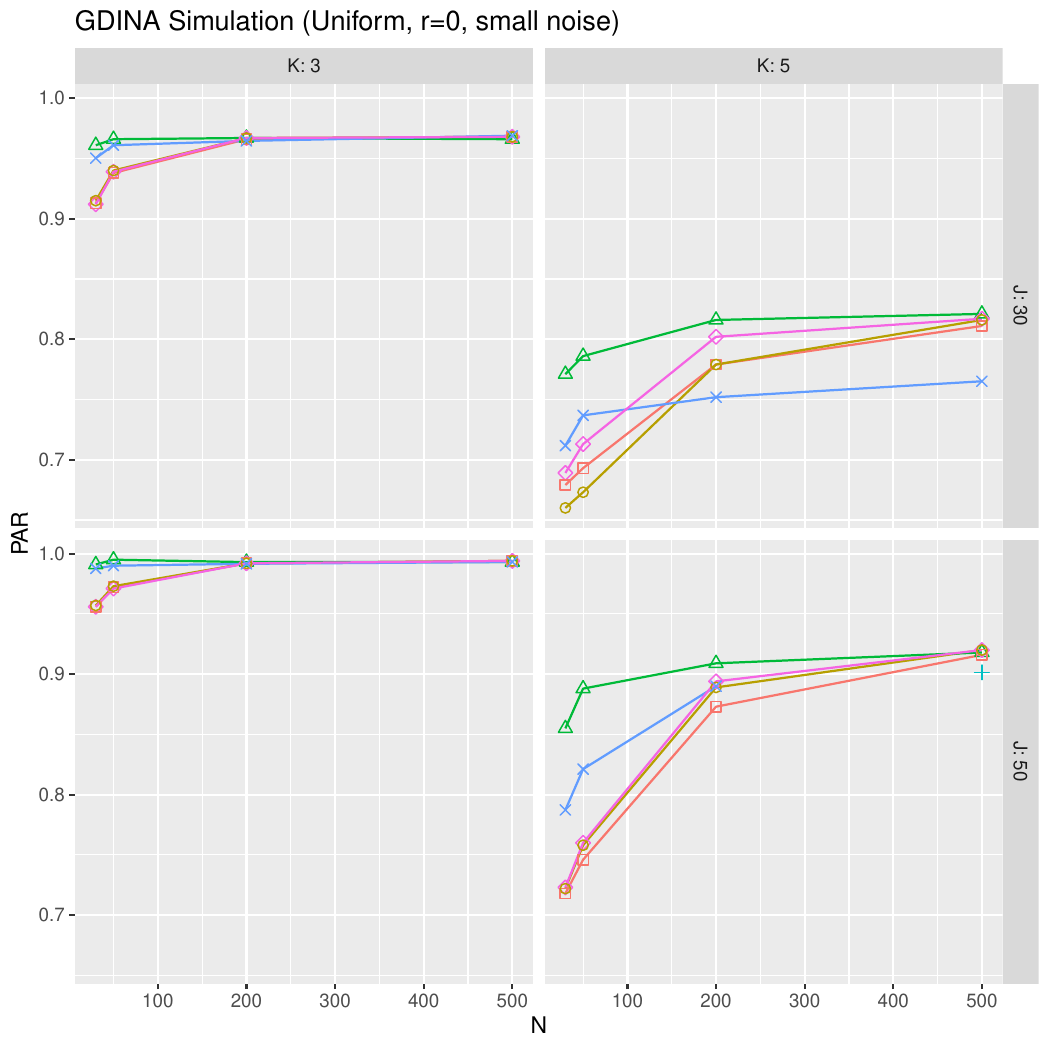}
    }
    \subfigure{
    \includegraphics[width=2.56in]{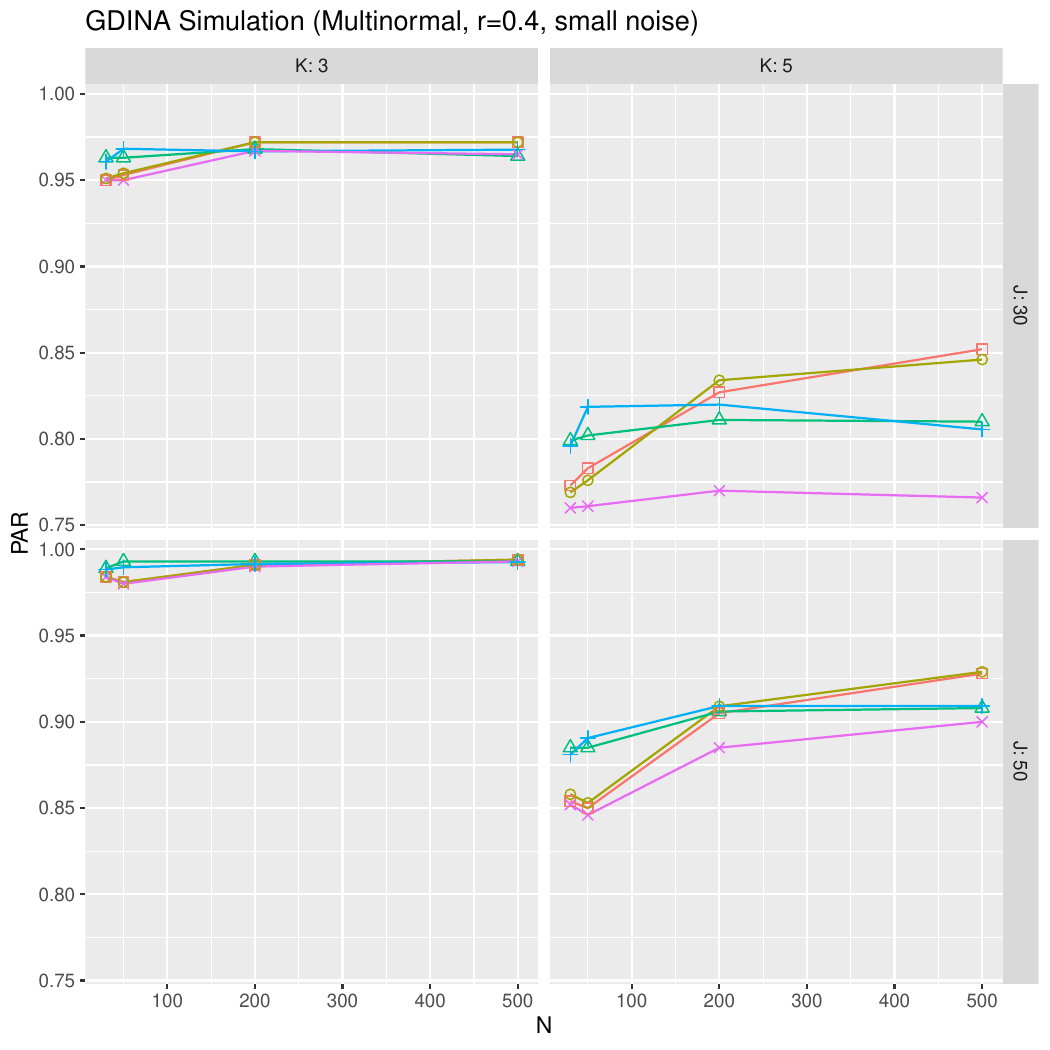}
    }
    \subfigure{
    \includegraphics[width=2.56in]{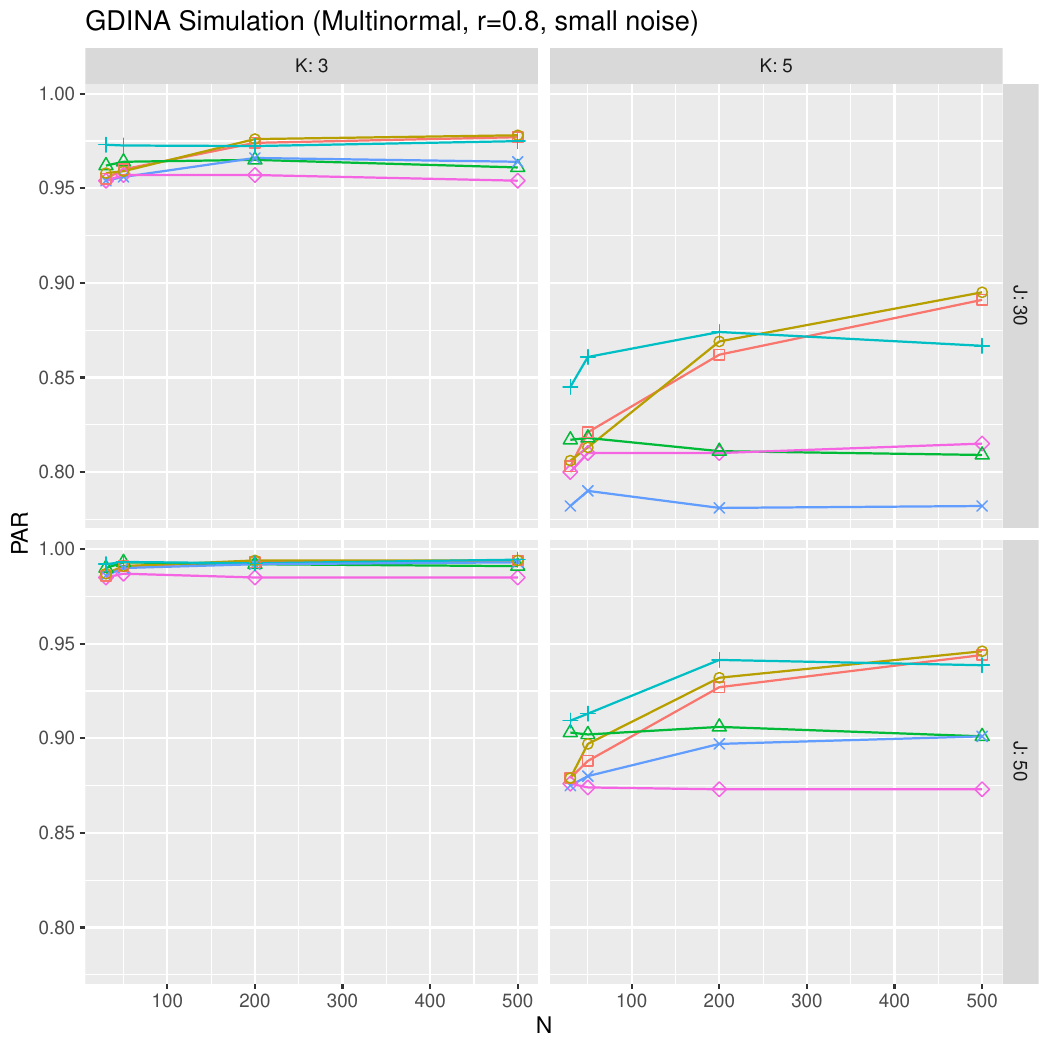}
    }
    \subfigure{
    \includegraphics[width=0.45in]{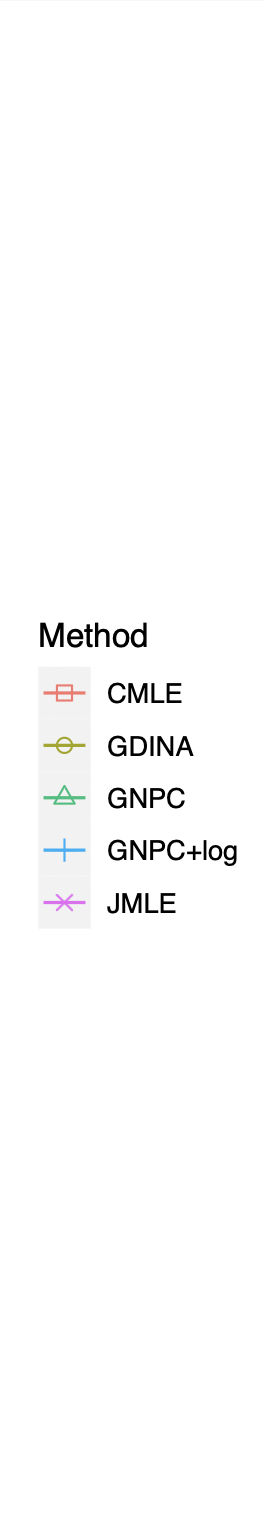}
    }
    \\
    \subfigure{
    \includegraphics[width=2.56in]{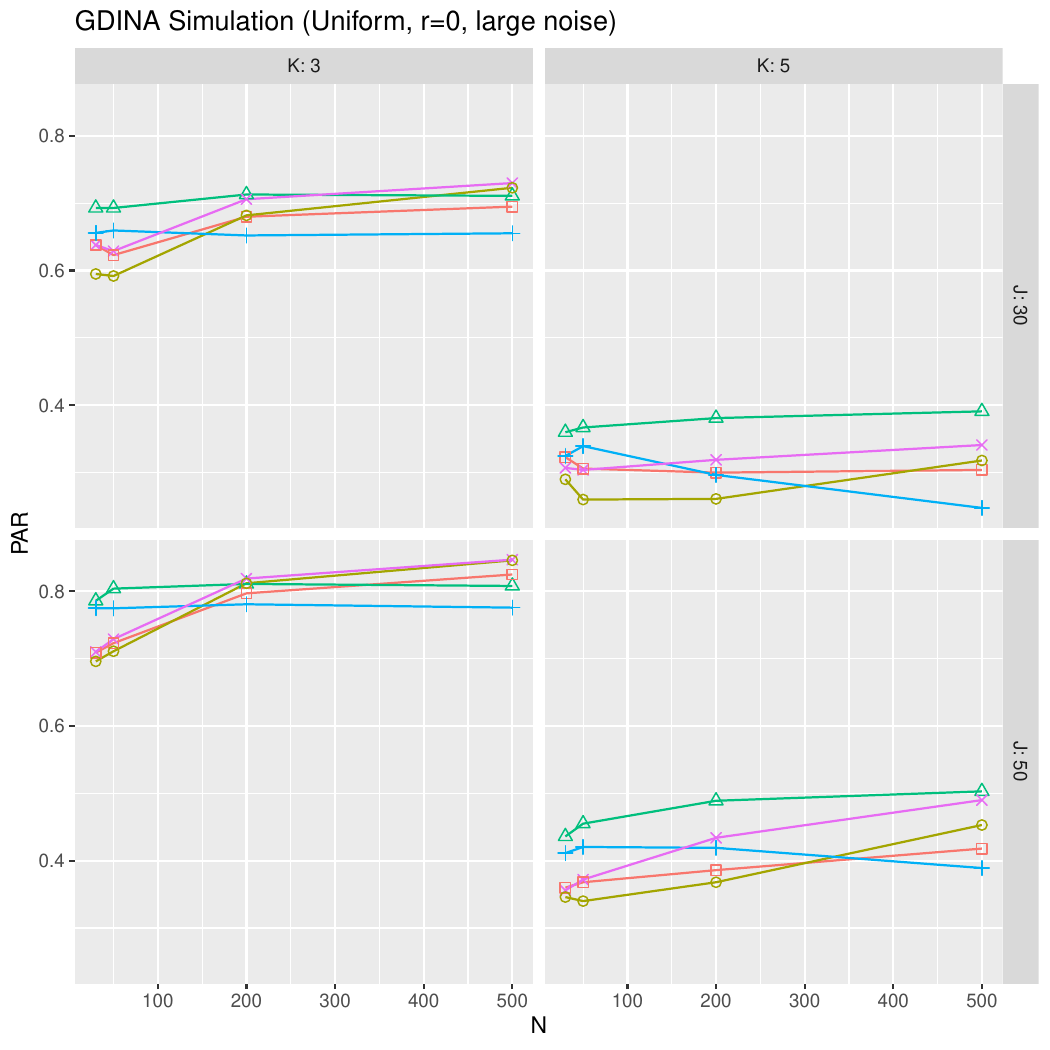}
    }
    \subfigure{
    \includegraphics[width=2.56in]{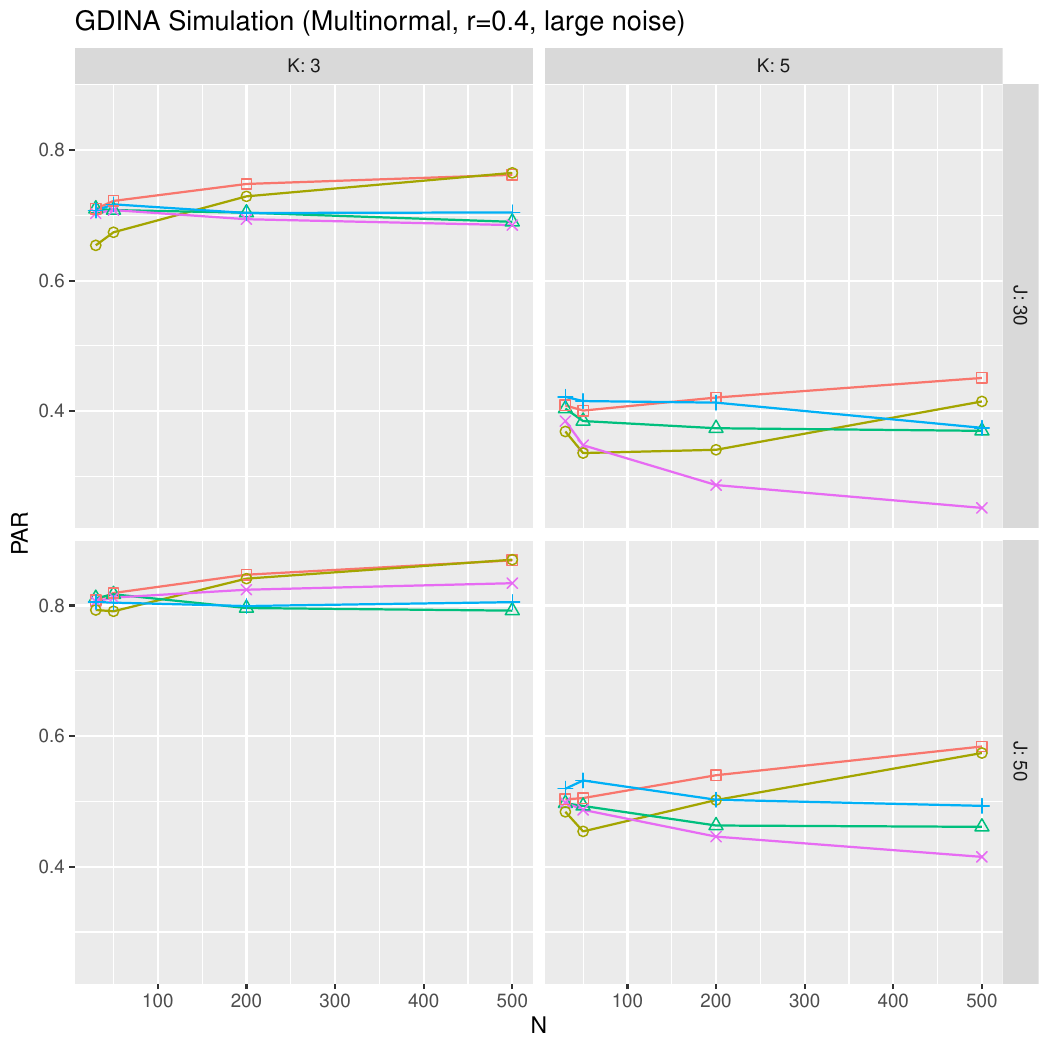}
    }
    \subfigure{
    \includegraphics[width=2.56in]{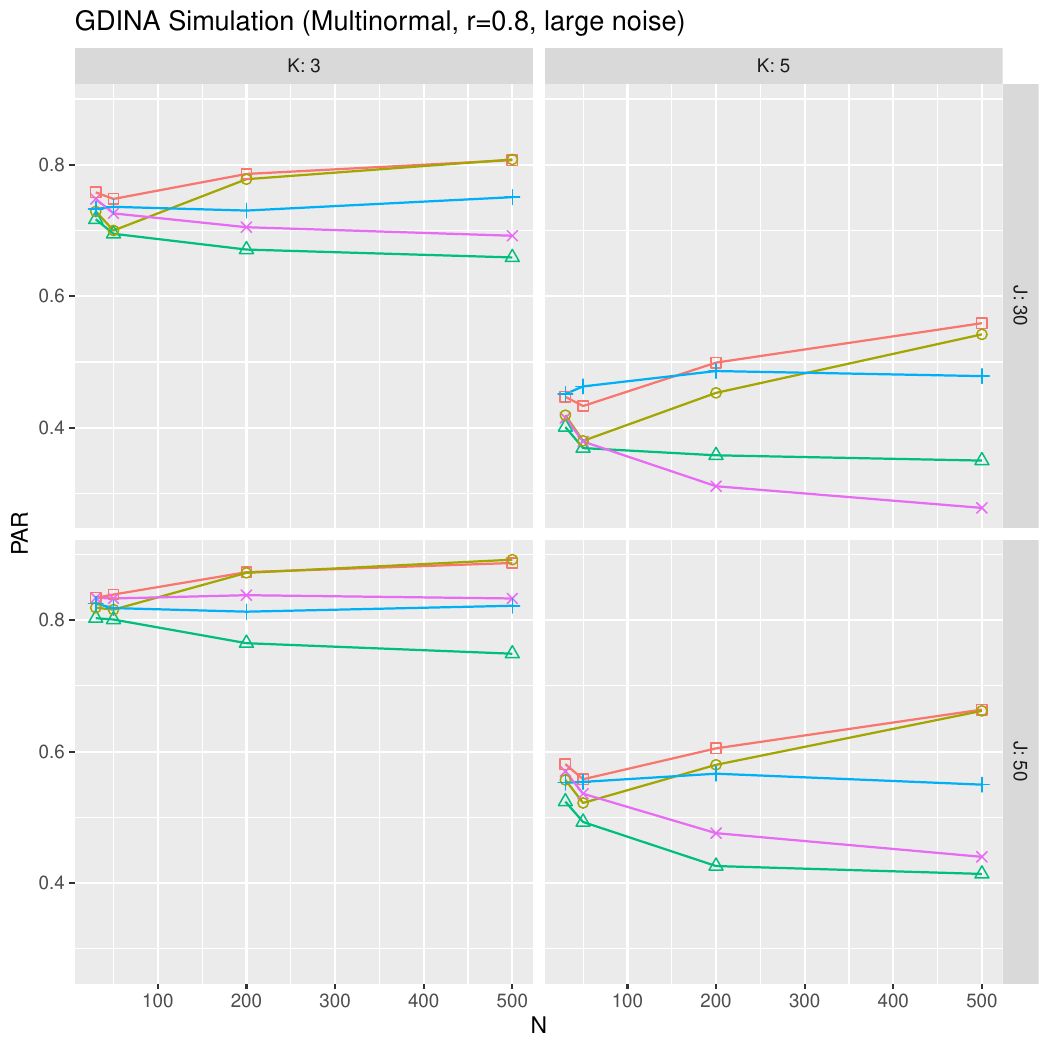}
    }
    \subfigure{
    \includegraphics[width=0.45in]{plots/GDINA_legend.png}
    }
    \caption{PARs when the data conformed to the GDINA model}
    \label{fig:GDINA-PAR-zoom}
\end{figure}{}
\end{landscape}

\begin{landscape}
	\begin{figure}[H]
    \centering
    \subfigure{
    \includegraphics[width=2.56in]{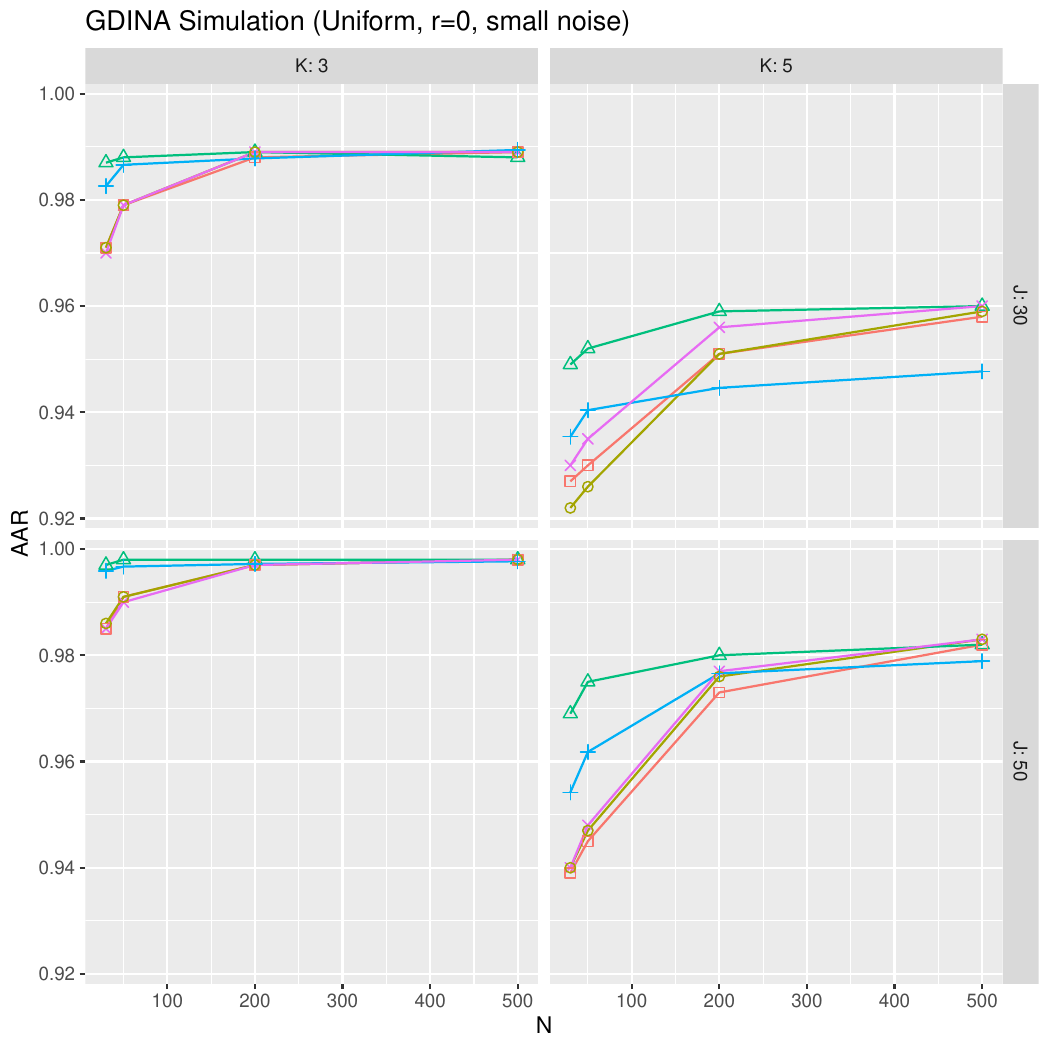}
    }
    \subfigure{
    \includegraphics[width=2.56in]{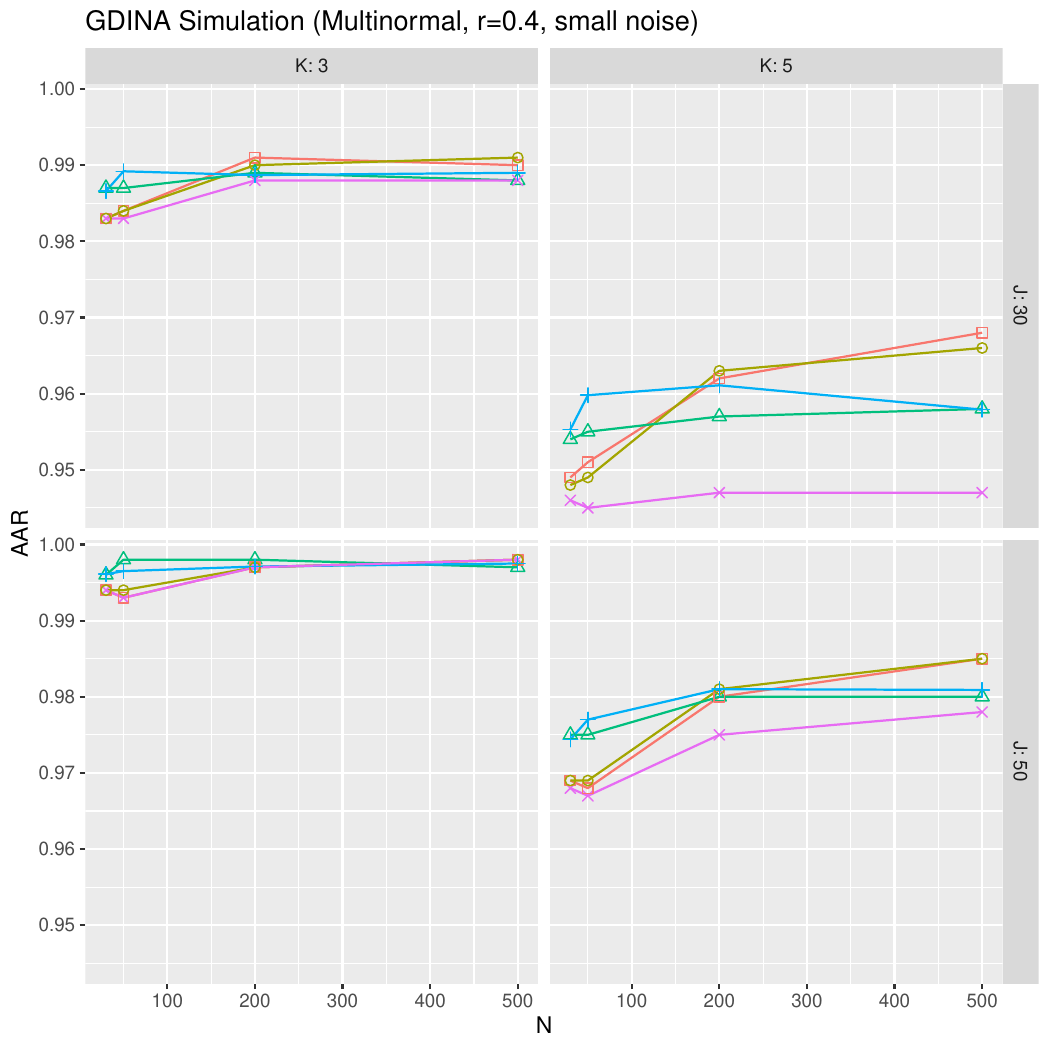}
    }
    \subfigure{
    \includegraphics[width=2.56in]{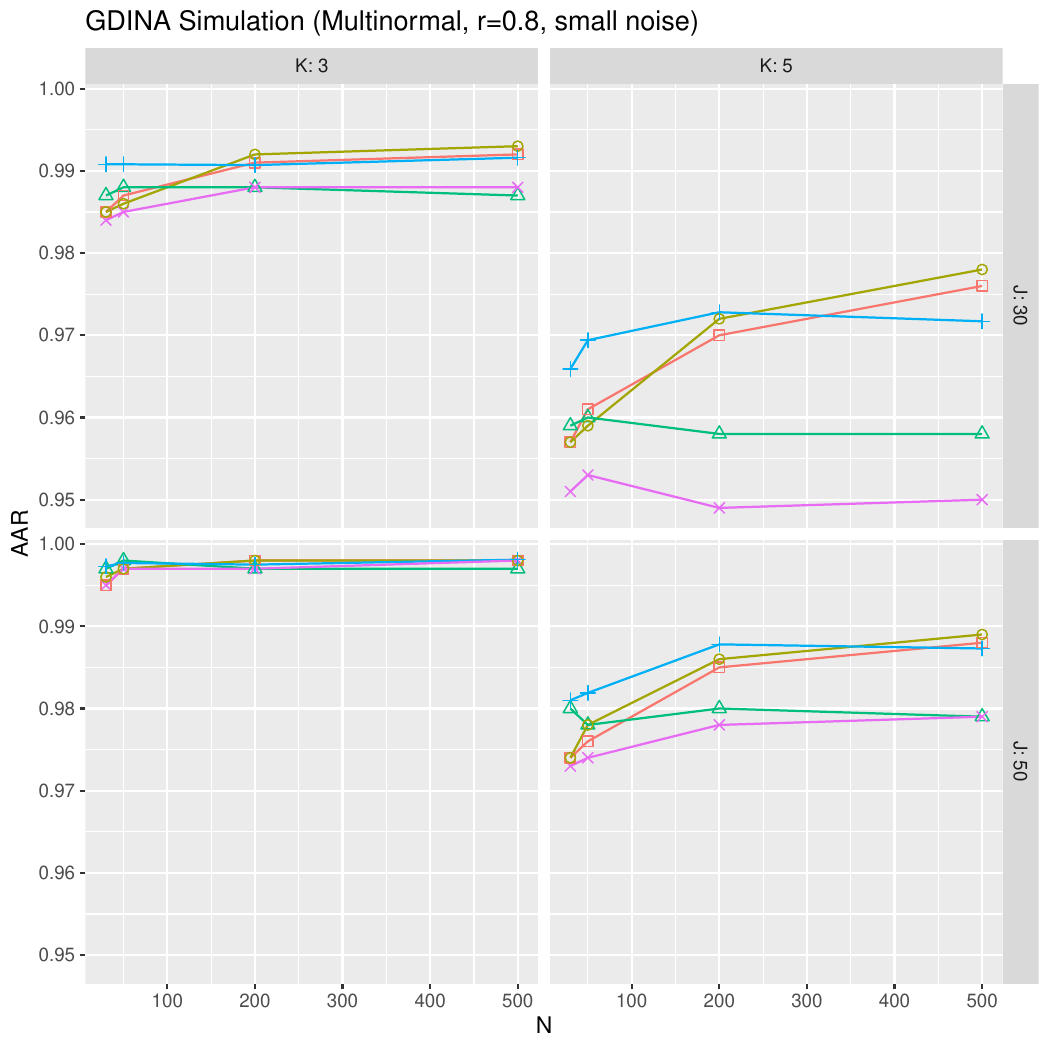}
    }
    \subfigure{
    \includegraphics[width=0.45in]{plots/GDINA_legend.png}
    }
    \\
    \subfigure{
    \includegraphics[width=2.56in]{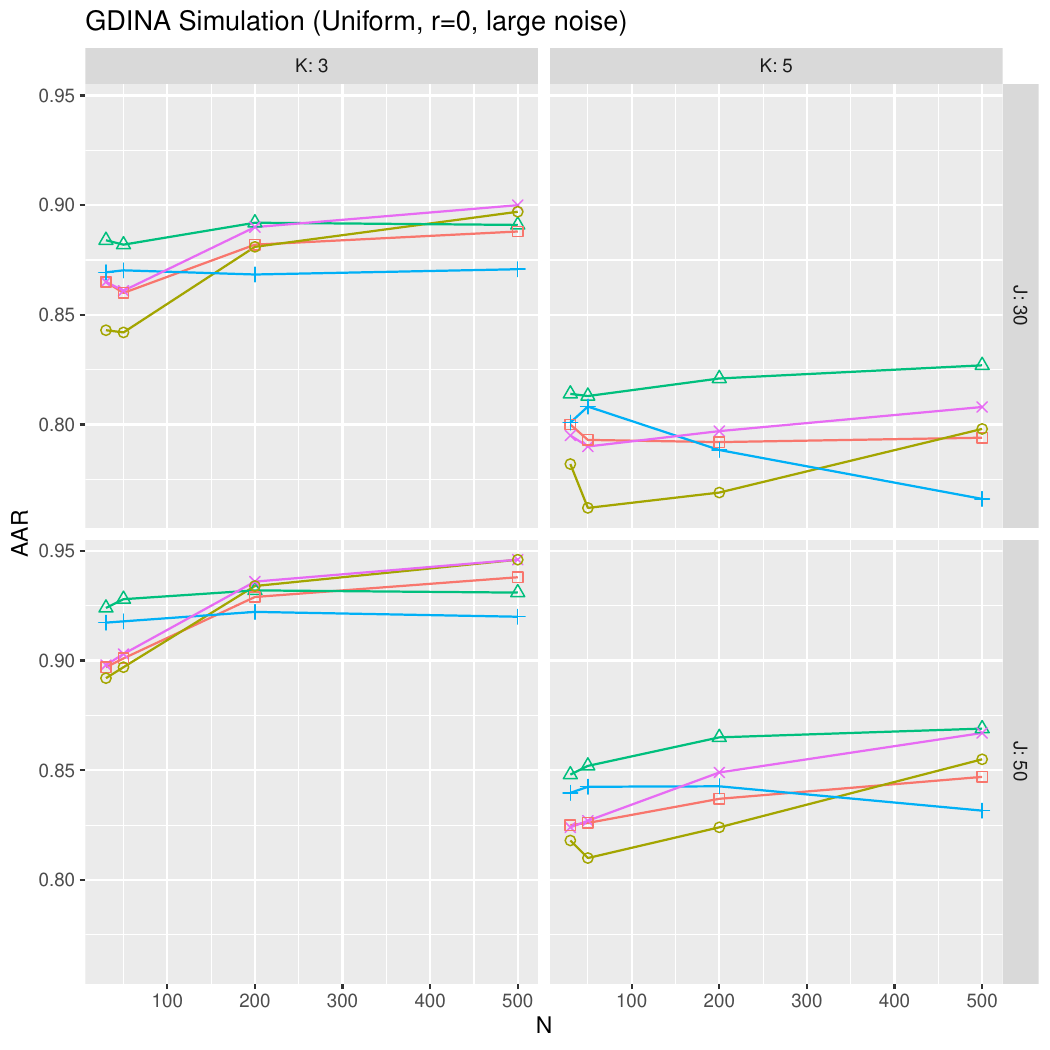}
    }
    \subfigure{
    \includegraphics[width=2.56in]{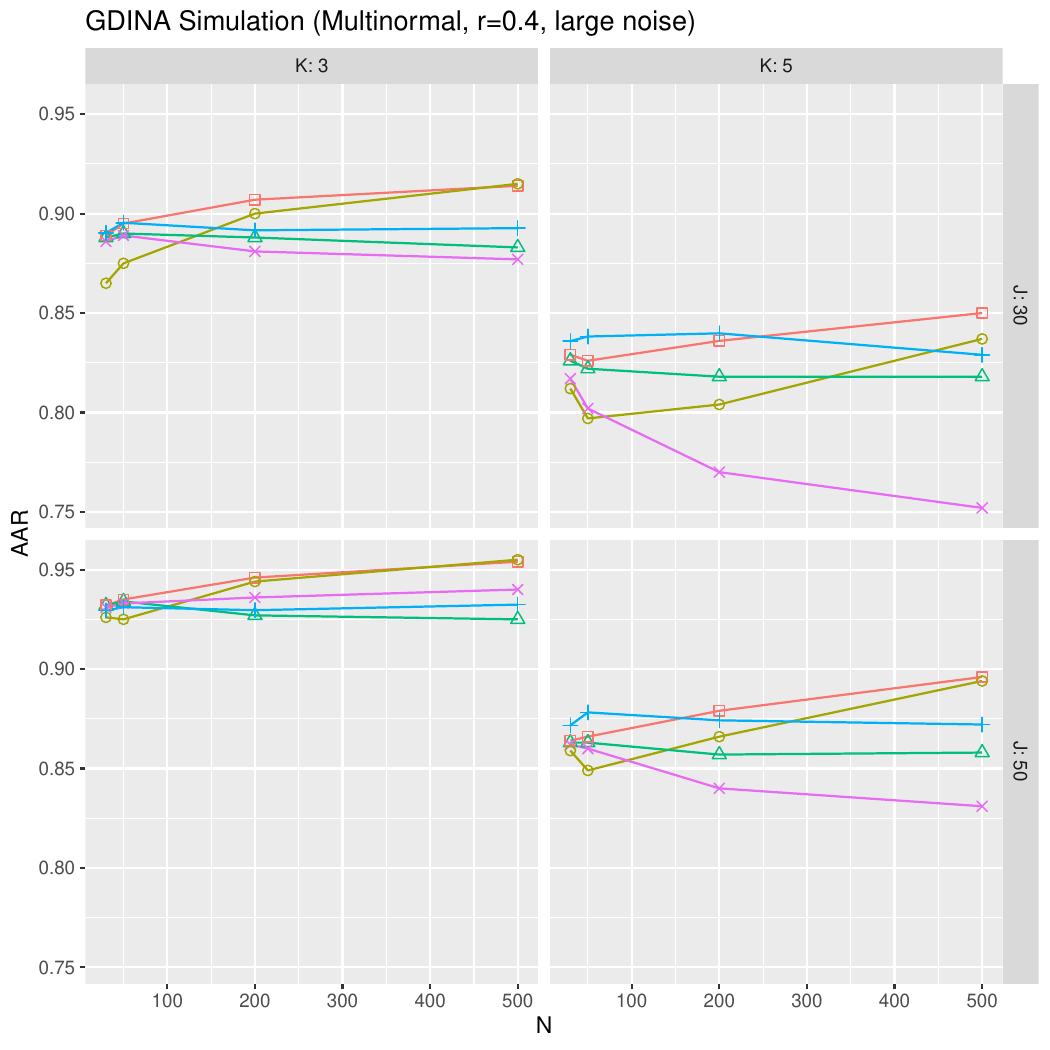}
    }
    \subfigure{
    \includegraphics[width=2.56in]{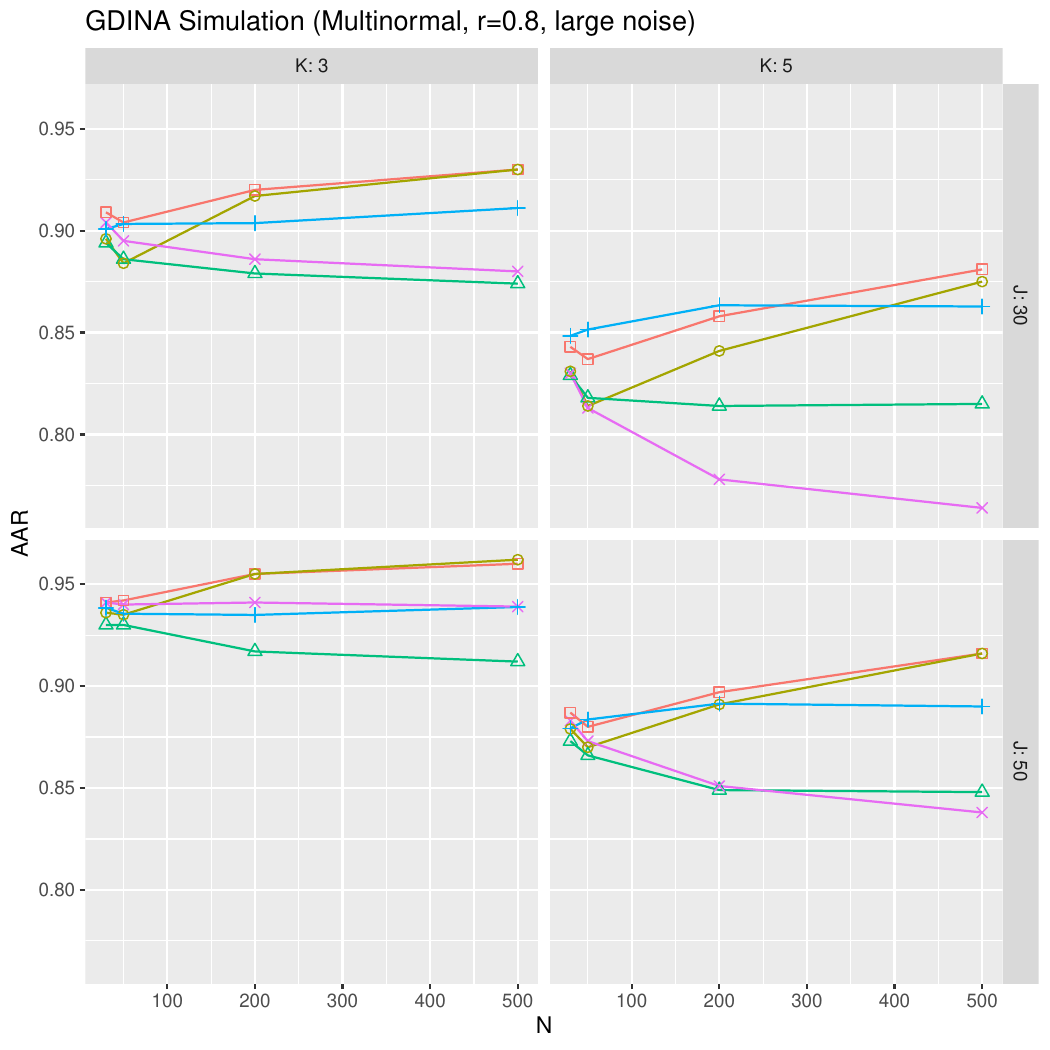}
    }
    \subfigure{
    \includegraphics[width=0.45in]{plots/GDINA_legend.png}
    }
    \caption{AARs when the data conformed to the GDINA model}
    \label{fig:GDINA-AAR-zoom}
\end{figure}{}
\end{landscape}

\begin{landscape}
\begin{figure}[H]
    \centering
    \subfigure{
    \includegraphics[width=2.56in]{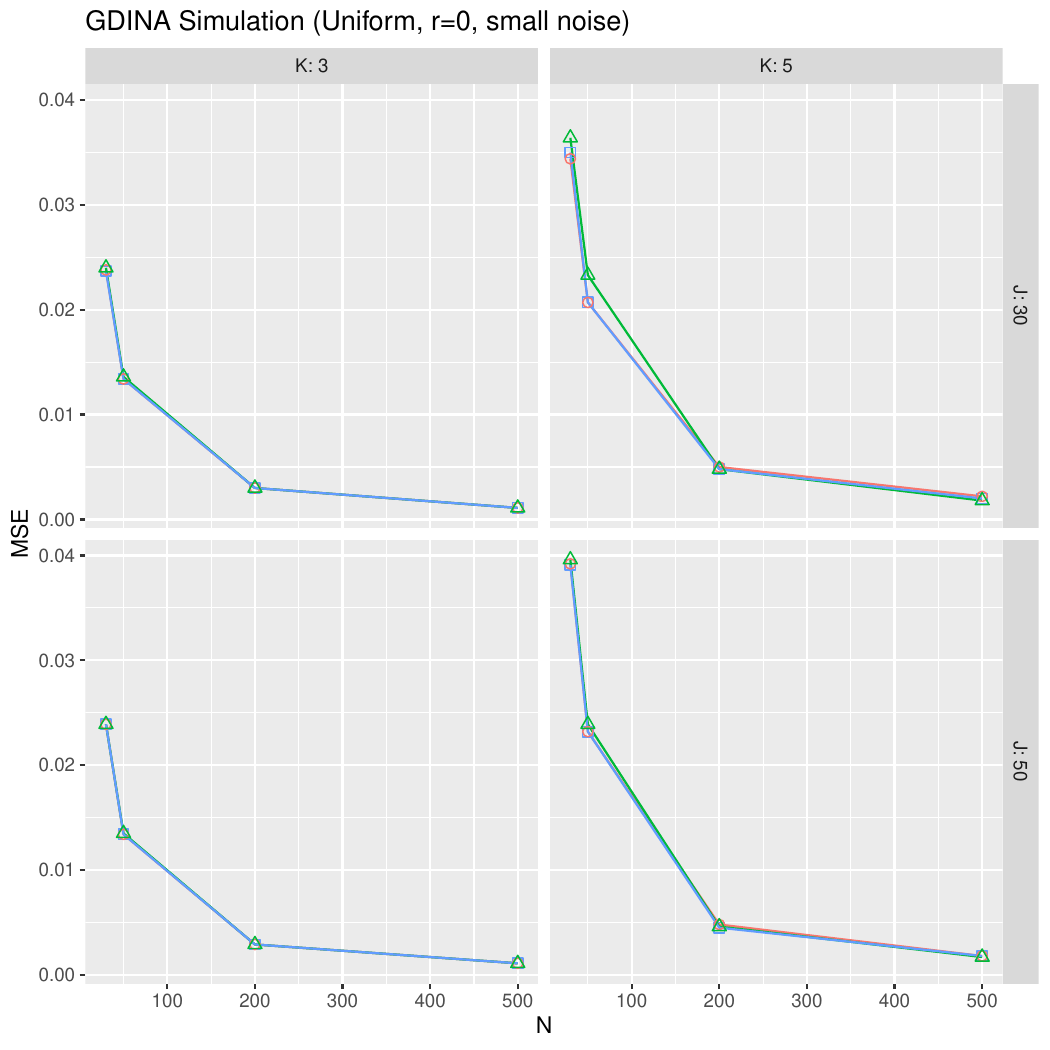}
    }
    \subfigure{
    \includegraphics[width=2.56in]{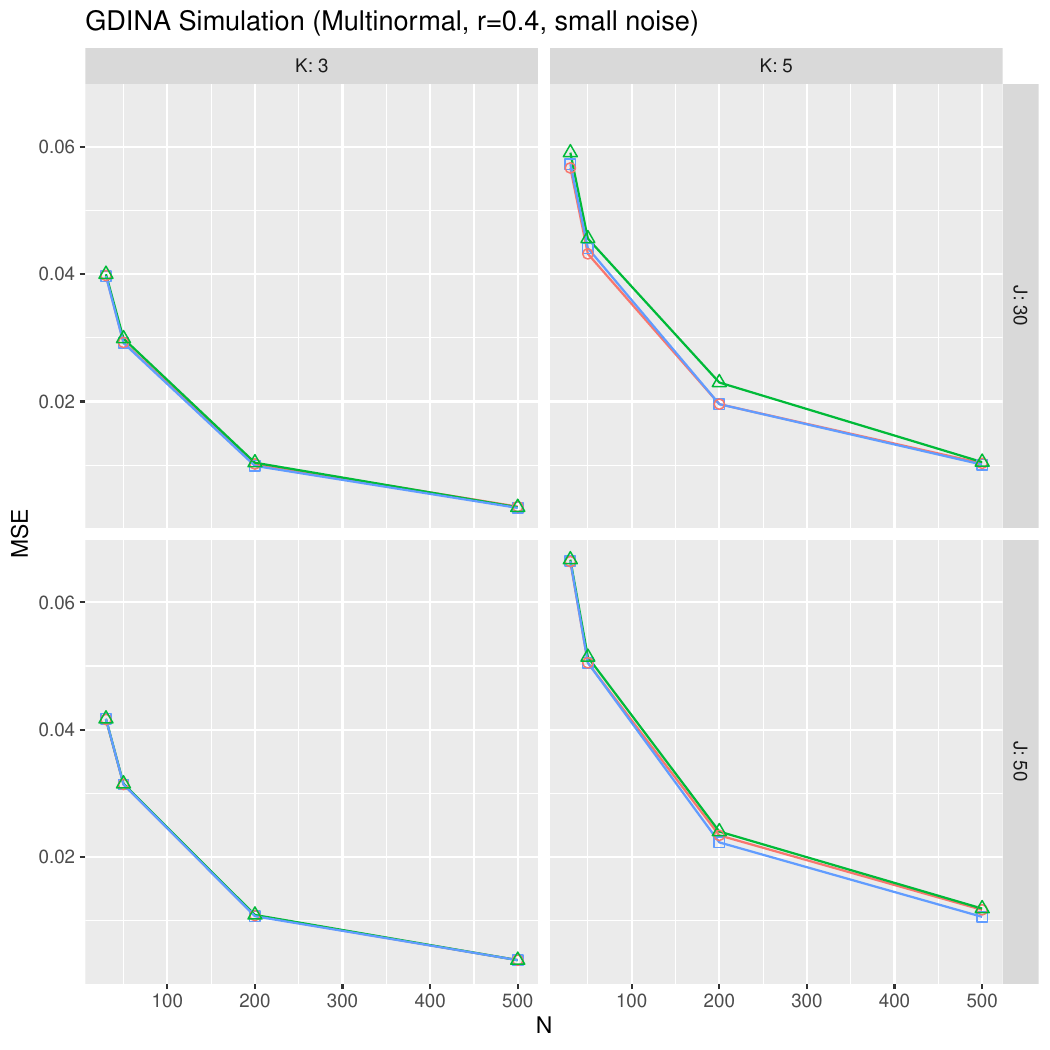}
    }
    \subfigure{
    \includegraphics[width=2.56in]{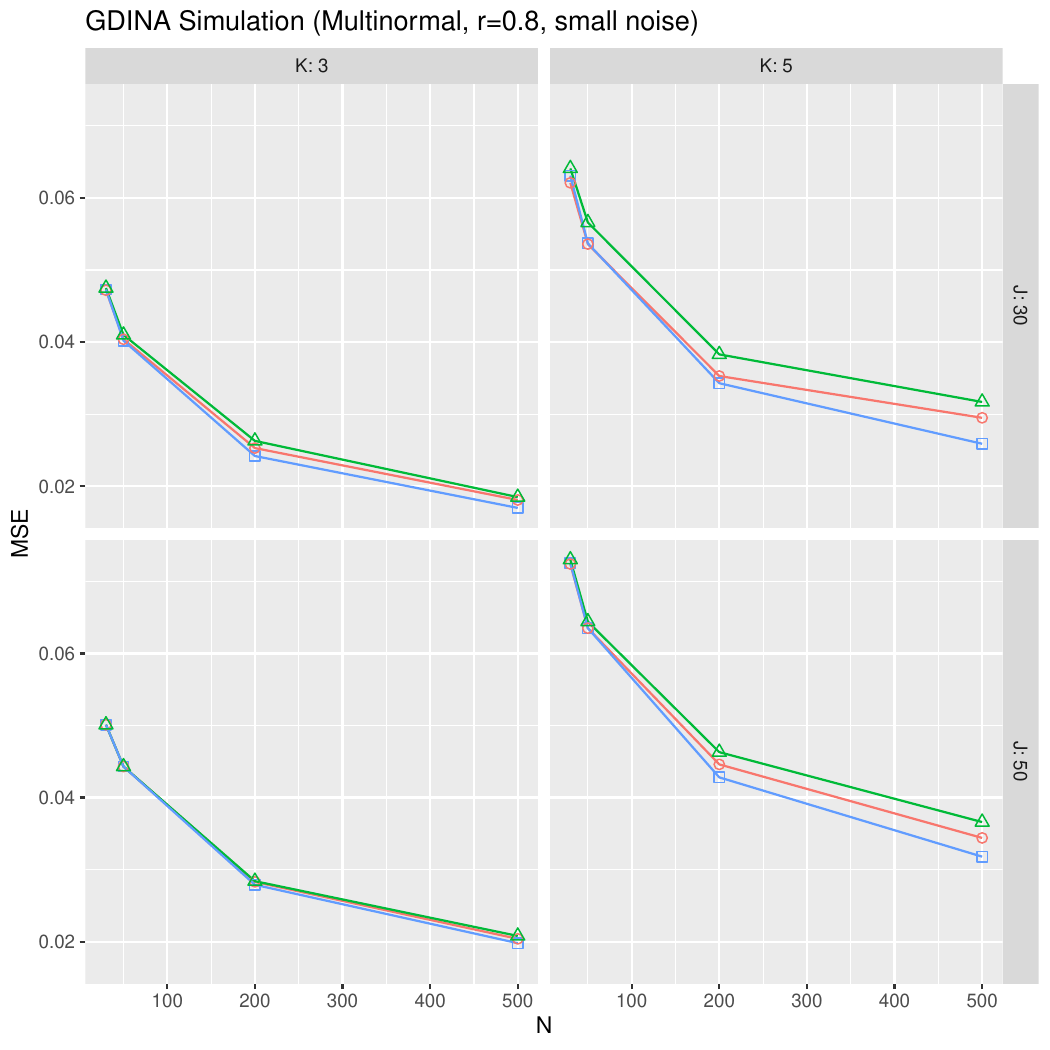}
    }
    \subfigure{
    \includegraphics[width=0.35in]{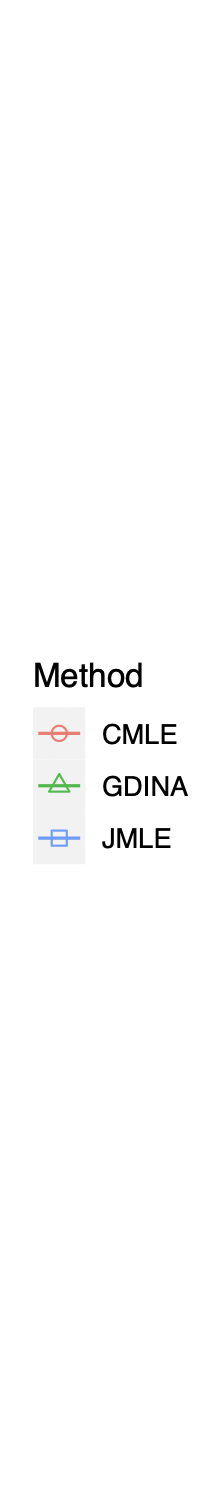}
    } 
    \\
    \subfigure{
    \includegraphics[width=2.56in]{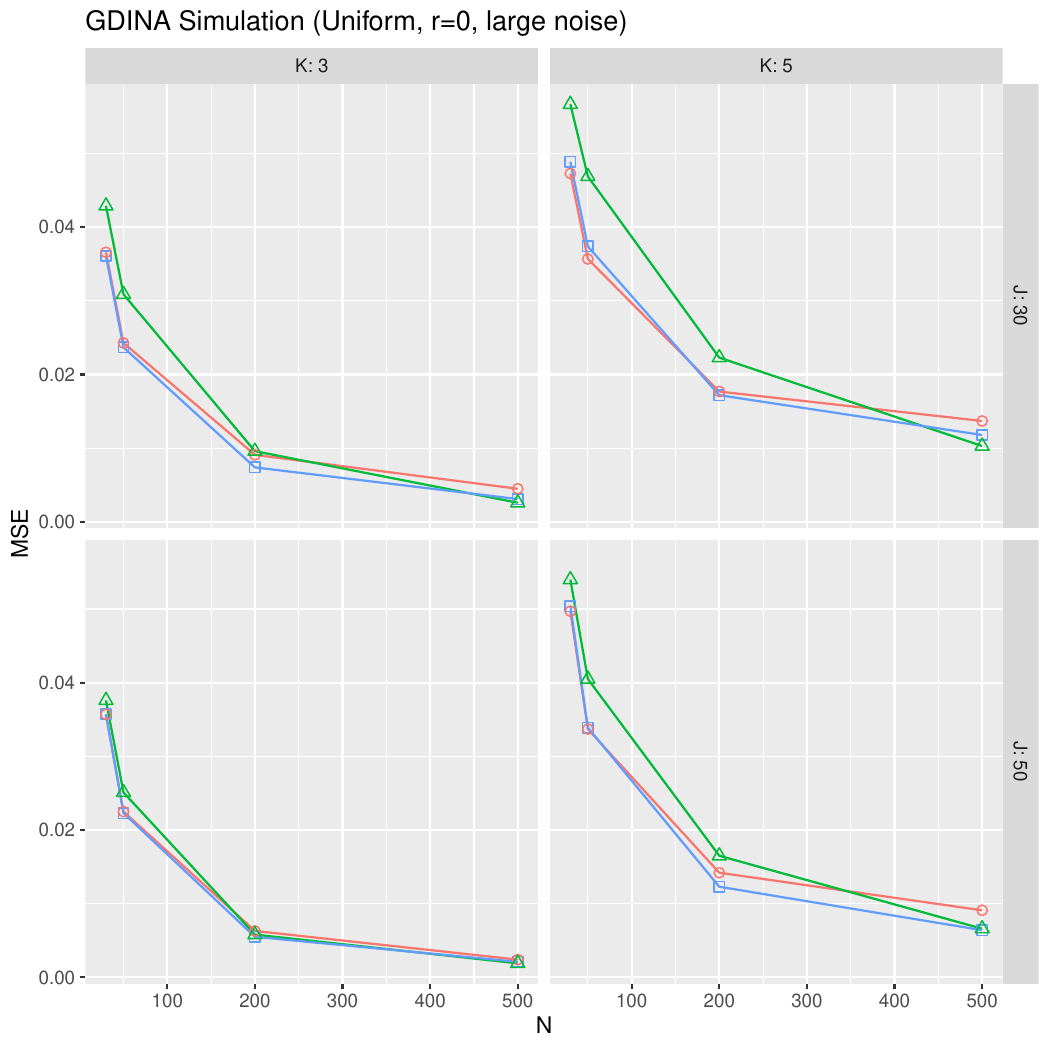}
    }
   	\subfigure{
    \includegraphics[width=2.56in]{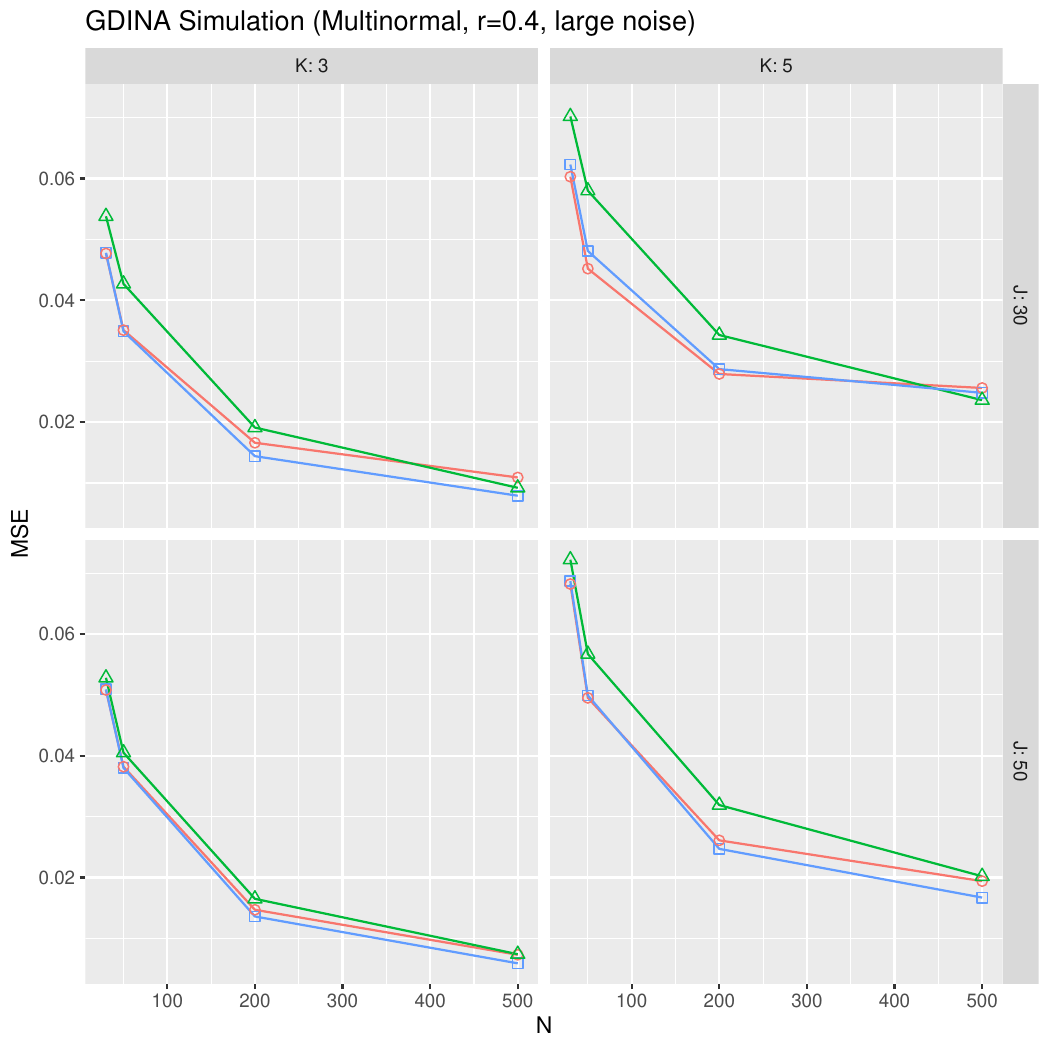}
   	}
    \subfigure{
    \includegraphics[width=2.56in]{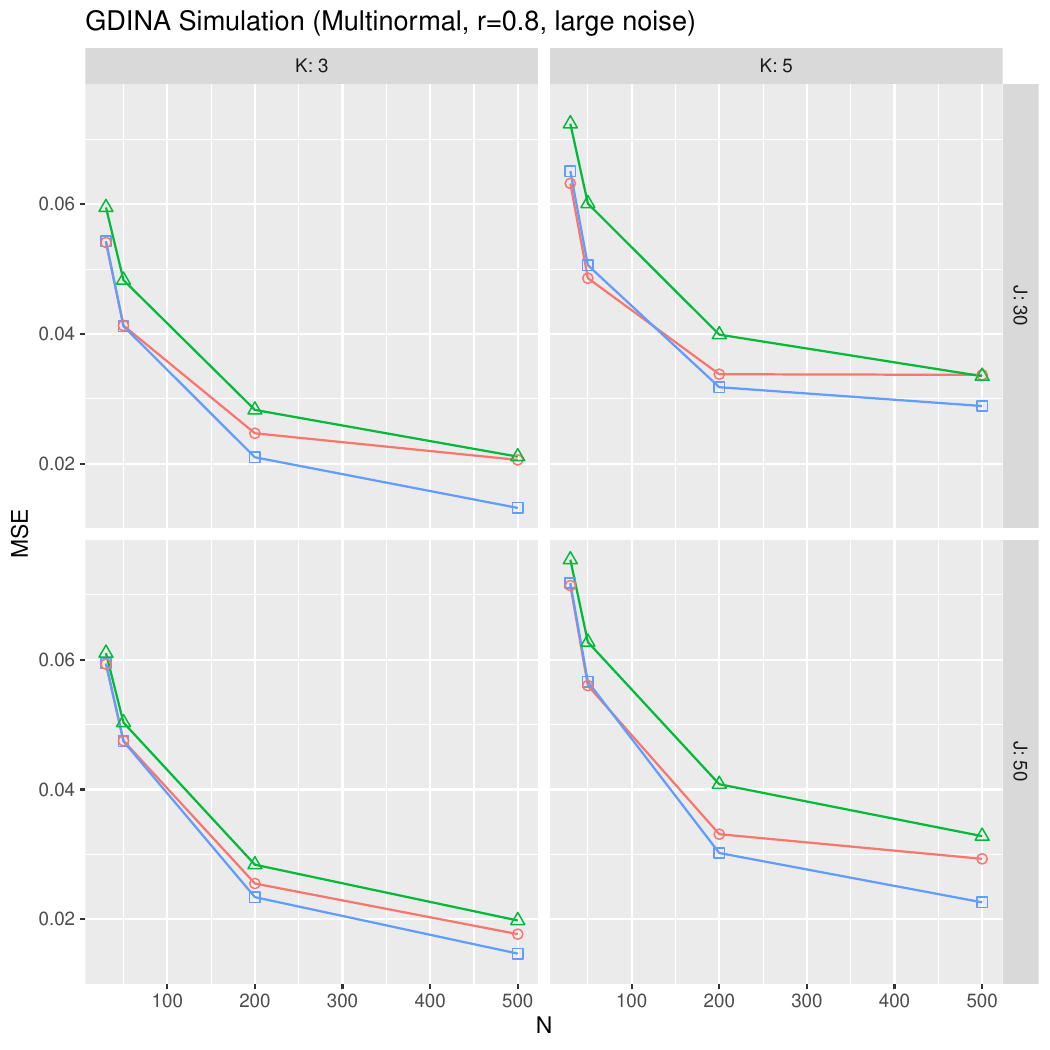}
    }
    \subfigure{
    \includegraphics[width=0.35in]{plots/GDINA_MSE_legend.png}
    }
    \caption{MSEs when the data conformed to the GDINA model}
    \label{fig:GDINA-MSE}
	\end{figure}{}
\end{landscape}

\section{Discussion}

In this paper, a unified estimation framework is proposed to bridge the parametric and nonparametric methods of cognitive diagnosis, and corresponding computational algorithms are developed.
Specifically, by choosing different loss functions and potentially imposing additional constraints on the centroids of the proficiency classes, the proposed framework essentially provides estimations for both parametric cognitive diagnosis models and nonparametric methods for classifying subjects to proficiency classes.
Moreover, we also provide theoretical analysis and establish consistency theories of the proposed framework.
The simulation studies under various settings demonstrate the advantages and disadvantages of different methods.

In our proposed framework \eqref{eq-C2}, we decompose the loss function into two additive parts.
In addition to the losses between the responses and class centroids, we also put a regularization term on the class proportions.
The regularization term can also play a role in selecting significant latent classes in the population.
For instance, similar to the CML in Examples \ref{ex-DINA} and \ref{ex-GDINA},  a log-type penalty  $h(\pi_{\boldsymbol{\alpha}}) = -\lambda \log (\pi_{\boldsymbol{\alpha}})$, where $\lambda > 0$ is a tuning parameter and $\pi_{\boldsymbol{\alpha}}$ is the proportion parameter for the latent pattern $\boldsymbol{\alpha}$, can be used.
Such a log-type penalty penalizes smaller proportions more heavily, and
as recently shown in \cite{gu2019learning}, can effectively select significant latent classes in the population.
Alternatively, to perform such latent class selection, the use of Lasso or elastic-net type penalty can be explored in the future.

Another interesting problem is the uncertainty quantification of the latent pattern classification.
Since in the proposed framework we directly assign the latent patterns by minimizing a loss function, the subjects' latent patterns are treated as fixed effects instead of random variables.
Based on the clustering literature, it is theoretically challenging to quantify the uncertainty of clustering accuracy.
One practical approach is to use bootstrap, where we resample the data multiple times and use the bootstrapped samples to quantify the estimation and classification uncertainty.
It is also possible to further model latent pattern probabilities and use large deviation theory to approximate the misclassification errors.
For instance, \cite{liu2015rate} studied the asymptotic misclassification error rate for CDMs under the assumption that the item parameters are pre-calibrated.
However, in the proposed framework, the item parameters and the latent patterns are unknown and  jointly estimated, and we focus on a more complicated double asymptotic regime, where the sample size $N$ and the number of items $J$ both go to infinity, making uncertainty quantification even more challenging.
This interesting problem will be explored further in the future.

One constraint of all the methods discussed in this paper pertain to the assumption that the $Q$-matrix is known and accurately specified.
In practice, the $Q$-matrix may not be given or subjectively specified by domain experts, with possible misspecifications.
There are some existing methods for estimating the $Q$-matrix in the literature \citep*{Chen2018, chen_liu_Xu_ying2015, Chung2018AnMA,7232852c3cdb4a2ebcac33aa88c1fedc, Liu_Xu_Z2012, XuEstimateQmatrix}.
Developing computational methods and theories for estimating CDMs with unknown $Q$-matrix under our proposed general framework is a natural next step that is left for future work.
Another possible extension is to consider hierarchical structures among the latent attributes \citep*{leighton2004attribute,templin2014hierarchical,ma_xu_2021},
which may exclude some latent patterns in the subjects' population.
Our proposed framework and computational algorithms should be easily adapted if the latent hierarchical structure is given.
Our theoretical analysis will also be readily carried over to the hierarchical setting.

\bibliographystyle{chicago}
\bibliography{bibref}

\pagebreak
\appendix

\begin{center}
{\large\bf SUPPLEMENTAL MATERIAL}
\end{center}

 \setcounter{equation}{0}
 \setcounter{figure}{0}
 \setcounter{table}{0}
 \setcounter{page}{1}
 \setcounter{section}{0}
 \setcounter{algocf}{0}
%  \makeatletter
 \renewcommand{\theequation}{S\arabic{equation}}
 \renewcommand{\thefigure}{S\arabic{figure}}
\renewcommand{\thealgocf}{S\arabic{algocf}}

\section{Appendix}
\label{sec-appendix}

In the appendix, we provide detailed proofs of the Lemmas and Theorems in Section \ref{sec-analysis}.

\subsection{Proof of Theorem 1}
\label{appendix-thm1}
\begin{proof}
{Our proof uses similar arguments as in  \cite{chiu2016joint}.}
First consider the case when the true membership $\boldsymbol{A}_c^0$ is known.
Since  $\hat{\mu}_{j,\boldsymbol{\alpha}} = \sum_{i\in C_{\boldsymbol{\alpha}}}x_{ij}/ |C_{\boldsymbol{\alpha}}| := \Bar{x}_{j,\boldsymbol{\alpha}}$,
by Hoeffding's inequality \citep{hoeffding1994probability}, for any $\epsilon>0$,
\begin{align*}
  \ P\big(\|\hat{\boldsymbol{\mu}}_{\boldsymbol{\alpha}}  - \boldsymbol{\theta}_{\boldsymbol{\alpha}}^0\|_\infty \geq \epsilon \ \big| \ \hat{\boldsymbol{A}}_c = \boldsymbol{A}_c^0 ) & =   P\Big(\max_j|\Bar{x}_{j,\boldsymbol{\alpha}}-\theta_{j,\boldsymbol{\alpha}}^0| \geq \epsilon \mid \hat{\boldsymbol{A}}_c = \boldsymbol{A}_c^0\Big ) \\
    & \leq \sum_{j=1}^J P\Big(|\Bar{x}_{j,\boldsymbol{\alpha}}-\theta_{j,\boldsymbol{\alpha}}^0| \geq \epsilon \mid \hat{\boldsymbol{A}}_c = \boldsymbol{A}_c^0\Big )  \\
    & \leq \ 2 J \exp \big( -2|C_{\boldsymbol{\alpha}}|\cdot \epsilon^2  \big).
 \end{align*}
Since $\lim_{n\to\infty}|C_{\boldsymbol{\alpha}}|/N_c\rightarrow \pi_{\boldsymbol{\alpha}}$ almost surely and $J \exp  \big( -N_c\epsilon  \big)\to 0$ for any $\epsilon>0$, we have $ J \exp \big( -2|C_{\boldsymbol{\alpha}}|\cdot \epsilon^2  \big) =J \exp \Big( -2\big(1 + o(1)\big)N_c \cdot \pi_{\boldsymbol{\alpha}}\cdot \epsilon^2  \Big) \to 0$ almost surely.

Now consider the case when $\hat{\boldsymbol{A}}_c$ is consistent for $\boldsymbol{A}_c^0$, that is, $P (\hat{\boldsymbol{A}}_c \neq \boldsymbol{A}_c^0) \rightarrow 0$.

Then for any $\epsilon >0$, we have
\begin{align*}
	& P\big(\|\hat{\boldsymbol{\mu}}_{\boldsymbol{\alpha}} - \boldsymbol{\theta}^0_{\boldsymbol{\alpha}}\|_\infty \geq \epsilon \big) \\
	\leq & \ P\big(\|\hat{\boldsymbol{\mu}}_{\boldsymbol{\alpha}}  - \boldsymbol{\theta}^0_{\boldsymbol{\alpha}}\|_\infty \geq \epsilon \ \big| \ \hat{\boldsymbol{A}}_c = \boldsymbol{A}_c^0 )\cdot P \big( \hat{\boldsymbol{A}}_c = \boldsymbol{A}_c^0 \big) + P\big(\|\hat{\boldsymbol{\mu}}_{\boldsymbol{\alpha}}  - \boldsymbol{\theta}^0_{\boldsymbol{\alpha}} \|_\infty \geq \epsilon \ \big| \ \hat{\boldsymbol{A}}_c \neq \boldsymbol{A}_c^0 )\cdot P \big( \hat{\boldsymbol{A}}_c \neq \boldsymbol{A}_c^0 \big)\\
	\leq &\ P\big(\|\hat{\boldsymbol{\mu}}_{\boldsymbol{\alpha}}  - \boldsymbol{\theta}_{\boldsymbol{\alpha}}^0\|_\infty \geq \epsilon  \ \big| \ \hat{\boldsymbol{A}}_c = \boldsymbol{A}_c^0 ) + P\big(\hat{\boldsymbol{A}}_c \neq \boldsymbol{A}_c^0 \big)\\
	\xrightarrow{P} & \ 0, \quad  \text{ as } J \rightarrow \infty.
\end{align*}
Therefore we have $\|\hat{\boldsymbol{\mu}}_{\boldsymbol{\alpha}} -\boldsymbol{\theta}^0_{\boldsymbol{\alpha}} \|_\infty \xrightarrow{P} 0$.
Since there are finitely many $\boldsymbol{\alpha}$'s, we have $\|\hat{\boldsymbol{\mu}} -\boldsymbol{\theta}^0 \|_\infty \xrightarrow{P} 0$
\end{proof}

\subsection{Proof of Lemma 1}
\label{appendix-lemma1}
\begin{proof}
\label{thm-theta-proof}
Let $\Tilde{\boldsymbol{\alpha}}_i$ denote the latent attribute pattern that minimizes $E[l(\boldsymbol{x}_i,\hat{\boldsymbol{\mu}}_{\boldsymbol{\alpha}})+h(\hat{\pi}_{\boldsymbol{\alpha}})]$, that is,
\begin{align*}
    \Tilde{\boldsymbol{\alpha}}_i := &
    \underset{\boldsymbol{\alpha}}{\arg\min}
    \big\{ E\big[l(\boldsymbol{x}_i,\hat{\boldsymbol{\mu}}_{\boldsymbol{\alpha}})+h(\hat{\pi}_{\boldsymbol{\alpha}})\big] \big\} \\
    = & \underset{\boldsymbol{\alpha}}{\arg\min}\ E\Big[\sum_{j=1}^J l(x_{ij},\hat{\mu}_{j,\boldsymbol{\alpha}}) + h(\hat{\pi}_{\boldsymbol{\alpha}})\Big] \\
    =& \underset{\boldsymbol{\alpha}}{\arg\min} \Big\{\frac{1}{J}\sum_{j=1}^J  E\big[l(x_{ij},\hat{\mu}_{j,\boldsymbol{\alpha}})\big] + \frac{1}{J} h(\hat{\pi}_{\boldsymbol{\alpha}})\Big\}.
\end{align*}
For the second term, under the Assumption 2, since $\hat{\pi}_{\boldsymbol{\alpha}}$ is asymptotically bounded and $h(\cdot)$ is continuous, hence $h(\hat{\pi}_{\boldsymbol{\alpha}})$ is also bounded,
and we have $ h(\hat{\pi}_{\boldsymbol{\alpha}})/J \xrightarrow{} 0$ as $J\xrightarrow{}\infty$, which is asymptotically negligible.
For the first term, we need to compare $\frac{1}{J}\sum_{j=1}^J  E\big[l(x_{ij},\hat{\mu}_{j,\boldsymbol{\alpha}})\big]$ and $\frac{1}{J}\sum_{j=1}^J  E\big[l(x_{ij},\hat{\mu}_{j,\boldsymbol{\alpha}_i^0})\big]$ for any $\boldsymbol{\alpha}\neq \boldsymbol{\alpha}_i^0$.
\begin{align} \nonumber
    &\ \frac{1}{J}\sum_{j=1}^J  E\big[l(x_{ij},\hat{\mu}_{j,\boldsymbol{\alpha}})\big] - \frac{1}{J}\sum_{j=1}^J  E\big[l(x_{ij},\hat{\mu}_{j,\boldsymbol{\alpha}^0_i})\big] \\ \nonumber
    =& \ \Big(
    \frac{1}{J}\sum_{j=1}^J  E\big[l(x_{ij},\hat{\mu}_{j,\boldsymbol{\alpha}})\big] - \frac{1}{J}\sum_{j=1}^J  E\big[l(x_{ij},\theta_{j,\boldsymbol{\alpha}}^0)\big]
    \Big) +
    \Big(
    \frac{1}{J}\sum_{j=1}^J  E\big[l(x_{ij},\theta_{j,\boldsymbol{\alpha}}^0)\big] - \frac{1}{J}\sum_{j=1}^J  E\big[l(x_{ij},\theta_{j,\boldsymbol{\alpha}^0_i}^0)\big]
    \Big)\\ \nonumber
    & \ +\Big(
    \frac{1}{J}\sum_{j=1}^J  E\big[l(x_{ij},\theta_{j,\boldsymbol{\alpha}^0_i}^0)\big] - \frac{1}{J}\sum_{j=1}^J  E\big[l(x_{ij},\hat{\mu}_{j,\boldsymbol{\alpha}^0_i})\big]
    \Big)\\
    :=&\ E_1 + E_2 + E_3. \label{eq-E123}
\end{align}
Since $\hat{\boldsymbol{\mu}}$ is consistent for $\boldsymbol{\theta}^0$, by Assumption 1, we have $E_1\xrightarrow{P} 0$ and $E_3\xrightarrow{P} 0$.
Specifically, first consider the case when $\hat{\boldsymbol{A}}_c = \boldsymbol{A}_c^0$.
By Assumption \ref{assump-delta-theta}, we know that the true item  response probabilities  are bounded.
There exists $\delta_2 \in (0,0.5)$ such that $\delta_2 \leq \underset{j,\boldsymbol{\alpha}}{\min} \  \theta_{j,\boldsymbol{\alpha}}^0 < \underset{j,\boldsymbol{\alpha}}{\max} \  \theta_{j,\boldsymbol{\alpha}}^0 \leq 1-\delta_2, \forall 1\leq j \leq J, \boldsymbol{\alpha}\in \{0,1\}^K.$
Let's now look at the probability that $\hat{\mu}_{j,\boldsymbol{\alpha}}$ is also bounded.
Specifically, we consider $P(\hat{\mu}_{j,\boldsymbol{\alpha}} \geq 1-\delta_2/2 \mid \hat{\boldsymbol{A}}_c = \boldsymbol{A}_c^0)$ and $P(\hat{\mu}_{j,\boldsymbol{\alpha}} \leq \delta_2/2 \mid \hat{\boldsymbol{A}}_c = \boldsymbol{A}_c^0)$ respectively.
Since  $\hat{\mu}_{j,\boldsymbol{\alpha}} = \sum_{i\in C_{\boldsymbol{\alpha}}}x_{ij}/ |C_{\boldsymbol{\alpha}}| := \Bar{x}_{j,\boldsymbol{\alpha}}$,
we have
\begin{align*}
    P(\hat{\mu}_{j,\boldsymbol{\alpha}} \geq 1-\delta_2/2 \mid \hat{\boldsymbol{A}}_c = \boldsymbol{A}_c^0 ) &= \
    P(\Bar{x}_{j,\boldsymbol{\alpha}}-\theta_{j,\boldsymbol{\alpha}}^0 \geq 1-\delta_2/2 -\theta_{j,\boldsymbol{\alpha}}^0 \mid \hat{\boldsymbol{A}}_c = \boldsymbol{A}_c^0 )\\[10pt]
    &\leq \ \exp \big( -2|C_{\boldsymbol{\alpha}}|(1-\delta_2/2-\theta_{j,\boldsymbol{\alpha}}^0)^2  \big)\\[10pt]
    &\leq \ \exp \big( -|C_{\boldsymbol{\alpha}}|\delta_2^2/2 \big).
\end{align*}
Similarly, we also have $P(\hat{\mu}_{j,\boldsymbol{\alpha}} \leq \delta_2/2 \mid \hat{\boldsymbol{A}}_c = \boldsymbol{A}^0_c ) \leq \exp \big( -|C_{\boldsymbol{\alpha}}|\delta_2^2/2 \big)$.
Therefore,
\begin{equation*}
	P\big(\min_j \hat{\mu}_{j,\boldsymbol{\alpha}} \leq
	\delta_2/2 \text{ or } \max_j \hat{\mu}_{j,\boldsymbol{\alpha}} \geq
	1 - \delta_2/2\mid \hat{\boldsymbol{A}}_c = \boldsymbol{A}^0_c \big) \leq 2 J \exp \big( -|C_{\boldsymbol{\alpha}}|\delta_2^2/2 \big).
\end{equation*}
Moreover, since under the Assumption \ref{assump-minima},
the loss function is assumed to be H\"older continuous, that is, there exist $c>0$ and $\beta>0$, such that for any $\mu_1,\mu_2 \in (\delta_2/2, 1-\delta_2/2)$,
we have $|l(x,\mu_1) - l(x,\mu_2)| \leq c |\mu_1 - \mu_2|^\beta$ for $x = 0$ or $1$.
Then
\begin{align*}
	|E_1| & = \Big| \frac{1}{J}\sum_{j=1}^J  E\big[l(x_{ij},\hat{\mu}_{j,\boldsymbol{\alpha}})\big] - \frac{1}{J}\sum_{j=1}^J  E\big[l(x_{ij},\theta_{j,\boldsymbol{\alpha}}^0)\big] \Big|\\
	& \leq \frac{1}{J}\sum_{j=1}^J E\Big[\big| l(x_{ij},\hat{\mu}_{j,\boldsymbol{\alpha}}) - l(x_{ij},\theta_{j,\boldsymbol{\alpha}}^0)\big|\Big] \\
	& \leq \frac{1}{J}\sum_{j=1}^J   E\big[c|\hat{\mu}_{j,\boldsymbol{\alpha}} - \theta_{j,\boldsymbol{\alpha}}^0|^{\beta}\big] \\
	& \leq c \max_{j} \{E \big[|\hat{\mu}_{j,\boldsymbol{\alpha}} - \theta_{j,\boldsymbol{\alpha}}^0|^{\beta}\big]\}
\end{align*}
Therefore for any $\epsilon > 0$,
\begin{align}\nonumber
	& \ P(|E_1| > \epsilon)\\ \nonumber
	\leq & \  P(|E_1| > \epsilon \mid  \hat{\boldsymbol{A}}_c = \boldsymbol{A}^0_c) + P(\hat{\boldsymbol{A}}_c \neq\boldsymbol{A}^0_c) \\ \nonumber
	\leq &\ P(E_1| > \epsilon \mid  \hat{\boldsymbol{A}}_c = \boldsymbol{A}^0_c, \delta_2 /2 < \hat{\mu}_{j,\boldsymbol{\alpha}} < 1 - \delta_2 / 2, j = 1, \dots, J)\\ \nonumber
	& + P\big(\min_j \hat{\mu}_{j,\boldsymbol{\alpha}} \leq
	\delta_2/2 \text{ or } \max_j \hat{\mu}_{j,\boldsymbol{\alpha}} \geq
	1 - \delta_2/2\mid \hat{\boldsymbol{A}}_c = \boldsymbol{A}^0_c \big) + P(\hat{\boldsymbol{A}}_c \neq\boldsymbol{A}^0_c)\\ \nonumber
	\leq &\ P(||\hat{\boldsymbol{\mu}}_{\boldsymbol{\alpha}} - \boldsymbol{\theta}_{\boldsymbol{\alpha}}^0||_{\infty} > (\epsilon / c)^{1/\beta}) + 2 J\exp(-|C_{\boldsymbol{\alpha}}|\delta_2^2/2) + P(\hat{\boldsymbol{A}}_c \neq\boldsymbol{A}^0_c)\\ \label{eq-E1}
	\leq & \ 2 J \exp(-2 |C_{\boldsymbol{\alpha}}|(\epsilon/c)^{2/\beta}) + 2 J\exp(-|C_{\boldsymbol{\alpha}}|\delta_2^2/2) + P(\hat{\boldsymbol{A}}_c \neq\boldsymbol{A}^0_c)\\ \nonumber
	= & \ 2 J \exp\Big(-2 \big(1 + o(1)\big) N_c\cdot \pi_{\boldsymbol{\alpha}}\cdot (\epsilon/c)^{2/\beta}\Big) + 2 J\exp\Big(- \big(1 + o(1)\big) N_c\cdot \pi_{\boldsymbol{\alpha}}\cdot \delta_2^2/2\Big) + P(\hat{\boldsymbol{A}}_c \neq\boldsymbol{A}^0_c)\\ \nonumber
	& \rightarrow 0,
\end{align}
where \eqref{eq-E1} follows from Theorem 1.
Similarly we can show that $E_3 \xrightarrow{P} 0$ as well.

For the second term, by Assumption \ref{assump-delta-loss}, we have
\begin{equation}
E_2 =
\frac{1}{J}\sum_{j=1}^J  E\big[l(x_{ij},\theta_{j,\boldsymbol{\alpha}}^0)\big] - \frac{1}{J}\sum_{j=1}^J  E\big[l(x_{ij},\theta_{j,\boldsymbol{\alpha}^0_i}^0)\big]
\geq \frac{1}{J}\sum_{j=1}^J |\theta_{j,\boldsymbol{\alpha}^0_i}^0 - \theta_{j,\boldsymbol{\alpha}}^0|^\delta,
\label{eq-E2}	
\end{equation}
for any $\boldsymbol{\alpha}\neq \boldsymbol{\alpha}_i^0$.
Since in Assumption \ref{assump-delta-theta}, there exists $\delta_1 > 0$ such that $\underset{J\rightarrow \infty}{\lim} \underset{\boldsymbol{\alpha}\neq \boldsymbol{\alpha}'}{\min} ||\boldsymbol{\theta}^0_{\boldsymbol{\alpha}} - \boldsymbol{\theta}^0_{\boldsymbol{\alpha}'} ||_1 / J> \delta_1$,
then for a small enough $c_0>0$, there exists $ c_1 > 0$ such that $\big| \{ j: |\theta_{j,\boldsymbol{\alpha}}^0 - \theta_{j,\boldsymbol{\alpha'}}^0| \geq c_0 \} \big| \geq c_1 J$ for any $\boldsymbol{\alpha} \neq \boldsymbol{\alpha}'$ and large enough $J$.
That is, there should be as many items as of order $J$ that can differentiate two different classes.
Otherwise,   $\big| \{ j: |\theta_{j,\boldsymbol{\alpha}}^0 - \theta_{j,\boldsymbol{\alpha'}}^0| \geq c_0 \} \big| / J \rightarrow 0$, which contradicts with the Assumption \ref{assump-delta-theta} for a small enough $c_0$.
Then in \eqref{eq-E2}, we have $E_2 \geq c_1 c_0^{\delta}$ as $J\rightarrow \infty$.
Therefore, the true attribute pattern   minimizes $E[l(\boldsymbol{x}_i,\hat{\boldsymbol{\mu}}_{\boldsymbol{\alpha}};\hat{\pi}_{\boldsymbol{\alpha}})]$ with probability approaching 1.
\end{proof}{}

\subsection{Proof of Lemma 2}
\label{appendix-lemma2}
\begin{proof}
We first decompose the probability in Lemma \ref{lemma-2} into two parts:
\begin{align}
   &\ P\Big(\underset{\boldsymbol{\alpha}}{\max} \Big|
    \frac{1}{J} \sum_{j=1}^J \big(
    l(x_{ij},\hat{\mu}_{j,\boldsymbol{\alpha}})
    - E[l(x_{ij},\theta_{j,\boldsymbol{\alpha}}^0)]
    \big) \Big| \geq \epsilon  \Big)\\
    \leq & \ P\Big(\underset{\boldsymbol{\alpha}}{\max} \Big|
    \frac{1}{J} \sum_{j=1}^J \big(
    l(x_{ij},\hat{\mu}_{j,\boldsymbol{\alpha}})
    - l(x_{ij},\theta_{j,\boldsymbol{\alpha}}^0)
    \big) \Big| \geq \epsilon/2  \Big)
    + P\Big(\underset{\boldsymbol{\alpha}}{\max} \Big|
    \frac{1}{J} \sum_{j=1}^J \big(
    l(x_{ij},\theta_{j,\boldsymbol{\alpha}}^0)
    - E[l(x_{ij},\theta_{j,\boldsymbol{\alpha}}^0)]
    \big) \Big| \geq \epsilon/2  \Big).
    \label{ineq-hoeff}
\end{align}
The first term in (\ref{ineq-hoeff}) goes to zero since $\hat{\boldsymbol{\theta}}$ is uniform consistent for $\boldsymbol{\theta}^0$.
Specifically, from Lemma 1, we have $P(\hat{\mu}_{j,\boldsymbol{\alpha}} \leq \delta_2/2 \text{ or } \hat{\mu}_{j,\boldsymbol{\alpha}} \geq 1 - \delta_2/2\mid \hat{\boldsymbol{A}}_c = \boldsymbol{A}^0_c ) \leq 2\exp \big( -|C_{\boldsymbol{\alpha}}|\delta_2^2/2 \big)$.
Moreover, 
due to the H\"older continuity of the loss function,
we have $|l(x,\mu_1) - l(x,\mu_2)| \leq c |\mu_1 - \mu_2|^\beta$ for $x = 0$ or $1$.
Then
\begingroup
\allowdisplaybreaks
\begin{align*}
    & P\Big(\underset{\boldsymbol{\alpha}}{\max}\Big|
    \frac{1}{J} \sum_{j=1}^J \big(
    l(x_{ij},\hat{\mu}_{j,\boldsymbol{\alpha}})
    - l(x_{ij},\theta_{j,\boldsymbol{\alpha}}^0)
    \big) \Big| \geq \epsilon/2  \ \Big| \ \delta_2/2 < \hat{\mu}_{j,\boldsymbol{\alpha}} \leq 1-\delta_2/2, \hat{\boldsymbol{A}}_c = \boldsymbol{A}^0_c \Big)\\
    \leq & \ \sum_{\boldsymbol{\alpha}}P\Big(\Big|
    \frac{1}{J} \sum_{j=1}^J \big(
    l(x_{ij},\hat{\mu}_{j,\boldsymbol{\alpha}})
    - l(x_{ij},\theta_{j,\boldsymbol{\alpha}}^0)
    \big) \Big| \geq \epsilon/2 \ \Big| \ \delta_2/2 < \hat{\mu}_{j,\boldsymbol{\alpha}} \leq 1-\delta_2/2,\ \hat{\boldsymbol{A}}_c = \boldsymbol{A}^0_c \Big)\\
    \leq & \ 2^K \sum_{j=1}^J P\Big(\Big|
    l(x_{ij},\hat{\mu}_{j,\boldsymbol{\alpha}})
    - l(x_{ij},\theta_{j,\boldsymbol{\alpha}}^0)
    \Big| \geq \epsilon/2 \ \Big| \ \delta_2/2 < \hat{\mu}_{j,\boldsymbol{\alpha}} \leq 1-\delta_2/2, \hat{\boldsymbol{A}}_c  = \boldsymbol{A}^0_c  \Big) \\
    \leq &\ 2^K \sum_{j=1}^J P\Big(\Big|
   \hat{\mu}_{j,\boldsymbol{\alpha}}
    - \theta_{j,\boldsymbol{\alpha}}^0
    \Big|^{\beta} \geq \epsilon/2c \ \Big| \ \delta_2/2 < \hat{\mu}_{j,\boldsymbol{\alpha}} \leq 1-\delta_2/2, \hat{\boldsymbol{A}}_c = \boldsymbol{A}^0_c \Big) \\
    = & \ 2^K \sum_{j=1}^J P\Big(\Big|
   \Bar{x}_{j,\boldsymbol{\alpha}}
    - \theta_{j,\boldsymbol{\alpha}}^0
     \Big| \geq (\epsilon/2c)^{1/\beta} \ \Big| \ \delta_2/2 < \hat{\mu}_{j,\boldsymbol{\alpha}} \leq 1-\delta_2/2, \hat{\boldsymbol{A}}_c = \boldsymbol{A}^0_c \Big)\\[10pt]
    \leq &\ 2^{K+1} J \exp \big(-2|C_{\boldsymbol{\alpha}}|(\epsilon/2c)^{2/\beta}
    \big).
\end{align*}
\endgroup
Then we have
\begin{align*}
    & \quad  P\Big(\underset{\boldsymbol{\alpha}}{\max} \Big|
    \frac{1}{J} \sum_{j=1}^J \big(
    l(x_{ij},\hat{\mu}_{j,\boldsymbol{\alpha}})
    - l(x_{ij},\theta_{j,\boldsymbol{\alpha}}^0)
    \big) \Big| \geq \epsilon/2 \ \Big| \ \hat{\boldsymbol{A}}_c = \boldsymbol{A}^0_c \Big)\\
    &\leq   \sum_{\boldsymbol{\alpha}} \sum_{j=1}^J \Big[ P(\hat{\mu}_{j,\boldsymbol{\alpha}} < \delta_2/2 \text{ or } \hat{\mu}_{j,\boldsymbol{\alpha}} > 1 - \delta_2/2 \mid \hat{\boldsymbol{A}}_c = \boldsymbol{A}^0_c) \\
    &\qquad \qquad \  + P\big(\big|
   \hat{\mu}_{j,\boldsymbol{\alpha}}
    - \theta_{j,\boldsymbol{\alpha}}^0
    \big| \geq (\epsilon/2c)^{1/\beta} \ \big| \ \delta_2/2 < \hat{\mu}_{j,\boldsymbol{\alpha}} \leq 1-\delta_2/2,\ \hat{\boldsymbol{A}}_c = \boldsymbol{A}^0_c \big)\Big]\\
     & \leq 2^{K+1} J \exp(-|C_{\boldsymbol{\alpha}}| \delta_2^2/2)
    + 2^{K+1} J \exp (-2|C_{\boldsymbol{\alpha}}|(\epsilon/2c)^{2/\beta})\\
    & = 2^{K+1} J \exp\Big(-\big(1 + o(1)\big)N_c\cdot \pi_{\boldsymbol{\alpha}}\cdot  \delta_2^2/2\Big)
    + 2^{K+1} J \exp \Big(-2\big(1 + o(1)\big)N_c\cdot \pi_{\boldsymbol{\alpha}}\cdot (\epsilon/2c)^{2/\beta}\Big)\\
    & \rightarrow 0, \text{ as } J \rightarrow \infty.
\end{align*}
Therefore, we have
\begin{align*}
	& \ P\Big(\underset{\boldsymbol{\alpha}}{\max} \Big|
    \frac{1}{J} \sum_{j=1}^J \big(
    l(x_{ij},\hat{\mu}_{j,\boldsymbol{\alpha}})
    - l(x_{ij},\theta_{j,\boldsymbol{\alpha}}^0)
    \big) \Big| \geq \epsilon/2 \Big)\\
	& \leq P\Big(\underset{\boldsymbol{\alpha}}{\max} \Big|
    \frac{1}{J} \sum_{j=1}^J \big(
    l(x_{ij},\hat{\mu}_{j,\boldsymbol{\alpha}})
    - l(x_{ij},\theta_{j,\boldsymbol{\alpha}}^0)
    \big) \Big| \geq \epsilon/2 \ \Big| \ \hat{\boldsymbol{A}}_c = \boldsymbol{A}^0_c \Big) \cdot P(\hat{\boldsymbol{A}}_c = \boldsymbol{A}^0_c) +  P(\hat{\boldsymbol{A}}_c \neq \boldsymbol{A}^0_c)\\
    &\leq P\Big(\underset{\boldsymbol{\alpha}}{\max} \Big|
    \frac{1}{J} \sum_{j=1}^J \big(
    l(x_{ij},\hat{\mu}_{j,\boldsymbol{\alpha}})
    - l(x_{ij},\theta_{j,\boldsymbol{\alpha}}^0)
    \big) \Big| \geq \epsilon/2 \ \Big| \ \hat{\boldsymbol{A}}_c = \boldsymbol{A}^0_c \Big)  +  P(\hat{\boldsymbol{A}}_c \neq \boldsymbol{A}^0_c)\\
    & \rightarrow 0, \text{ as } J \rightarrow \infty.
\end{align*}

Next we need to bound the second term.
By Assumption \ref{assump-delta-theta}, $\theta_{j,\boldsymbol{\alpha}}^0$'s are uniformly bounded and thus $l(x_{ij},\theta_{j,\boldsymbol{\alpha}}^0)$'s are also uniformly bounded.
There exists $M > 0$ such that $\big|l(x_{ij},\theta_{j,\boldsymbol{\alpha}}^0)\big| \leq M$ for any $j$ and $\boldsymbol{\alpha}$. 
Then by Hoeffding's inequality \citep{hoeffding1994probability}, we have
\begin{equation*}
    P\Big(\Big|
    \frac{1}{J} \sum_{j=1}^J\big(
    l(x_{ij},\theta_{j,\boldsymbol{\alpha}}^0)
    - E[l(x_{ij},\theta_{j,\boldsymbol{\alpha}}^0)]
    \big) \Big| \geq \epsilon/2  \Big)
    \leq
    2 \exp \big(
    - J\epsilon^2 / 2M^2
    \big),
\end{equation*}
and therefore
\begin{align*}
    &\quad P\Big(\underset{\boldsymbol{\alpha}}{\max} \Big|
    \frac{1}{J} \sum_{j=1}^J\big(
    l(x_{ij},\theta_{j,\boldsymbol{\alpha}}^0)
    - E[l(x_{ij},\theta_{j,\boldsymbol{\alpha}}^0)]
    \big) \Big| \geq \epsilon/2  \Big) \\
    & \leq  \ \sum_{\boldsymbol{\alpha}} P\Big(\Big|
    \frac{1}{J}\sum_{j=1}^J \big(
    l(x_{ij},\theta_{j,\boldsymbol{\alpha}}^0)
    - E[l(x_{ij},\theta_{j,\boldsymbol{\alpha}}^0)]
    \big) \big| \geq \epsilon/2  \Big)\\
    & \leq \ 2^{K+1}\exp \big(
    - J\epsilon^2 / 2M^2
    \big) \longrightarrow 0, \text{ as } J\xrightarrow{} \infty.
\end{align*}
\end{proof}{}

\subsection{Proof of Theorem 2}
\label{appendix-thm2}
\begin{proof}
Since $\hat{\boldsymbol{A}}_c$ is consistent for $\boldsymbol{A}_c^0$, by Theorem \ref{thm-theta}, $\hat{\boldsymbol{\mu}}$ is consistent for $\boldsymbol{\theta}^0$.
Note that $\hat{\boldsymbol{\alpha}}_i \neq \boldsymbol{
\alpha}_i^0$ is equivalent to that
\begin{equation}
    \frac{1}{J} \sum_{j=1}^J l(x_{ij},\hat{\mu}_{j,\boldsymbol{\alpha}_i^0}) + \frac{1}{J} h(\hat{\pi}_{\boldsymbol{\alpha}_i^0})
    >
    \frac{1}{J} \sum_{j=1}^J l(x_{ij},\hat{\mu}_{j,\hat{\boldsymbol{\alpha}}_i}) + \frac{1}{J} h(\hat{\pi}_{\hat{\boldsymbol{\alpha}}_i}).
    \label{eq-delta}
\end{equation}
From Assumptions \ref{assump-minima} and \ref{assump-delta-loss} and the proof of Lemma \ref{lemma-1}, we know
\begin{equation}
    \frac{1}{J} \sum_{j=1}^J E[l(x_{ij},\theta_{j,\boldsymbol{\alpha}_i^0}^0)]
    <
    \frac{1}{J} \sum_{j=1}^J E[l(x_{ij},\theta_{j,\hat{\boldsymbol{\alpha}}_i}^0)] - c_1 c_0^{\delta}
\end{equation}
Let $c_2 = c_1 c_0^{\delta}$ and take $\epsilon = c_2 / 4$ in Lemma \ref{lemma-2}, and consider the event
\[
B_{\epsilon}(J):=\big\{ \underset{\boldsymbol{\alpha}}{\max} \Big|
    \frac{1}{J} \sum_{j=1}^J \big(
    l(x_{ij},\hat{\mu}_{j,\boldsymbol{\alpha}})
    - E[l(x_{ij},\theta_{j,\boldsymbol{\alpha}}^0)]
    \big) \Big| < \epsilon \big\}.
\]
Since $h(\hat{\pi}_{\alpha})$ is bounded, there exists some $J_0$ such that for any $J\geq J_0$, we have
$
    \big|\frac{1}{J} h(\hat{\pi}_{\boldsymbol{\alpha}_i^0}) - \frac{1}{J} h(\hat{\pi}_{\hat{\boldsymbol{\alpha}}_i})\big| < c_2 / 4.
$
When $B_{c_2/4}(J)$ occurs, it implies that
\[
    \Big|
    \frac{1}{J} \sum_{j=1}^J \big(
    l(x_{ij},\hat{\mu}_{j,\boldsymbol{\alpha}_i^0})
    - E[l(x_{ij},\theta_{j,\boldsymbol{\alpha}_i^0}^0)]
    \big) \Big| < c_2/4,
\]
and
\[
    \Big|
    \frac{1}{J} \sum_{j=1}^J \big(
    l(x_{ij},\hat{\mu}_{j,\hat{\boldsymbol{\alpha}}_i})
    - E[l(x_{ij},\theta_{j,\hat{\boldsymbol{\alpha}}_i}^0)]
    \big) \Big| < c_2/4.
\]
Then in equation \eqref{eq-delta},
\[
    \text{LHS} < \frac{1}{J} \sum_{j=1}^J E[l(x_{ij},\theta_{j,\boldsymbol{\alpha}_i^0}^0)] + c_2/4 + \frac{1}{J} h(\hat{\pi}_{\boldsymbol{\alpha}_i^0}),
\]
and
\[
    \text{RHS} > \frac{1}{J} \sum_{j=1}^J E[l(x_{ij},\theta_{j,\hat{\boldsymbol{\alpha}}_i}^0)] - c_2/4 + \frac{1}{J} h(\hat{\pi}_{\hat{\boldsymbol{\alpha}}_i}),
\]
which implies that
\begin{align*}
    \frac{1}{J} \sum_{j=1}^J E[l(x_{ij},\theta_{j,\hat{\boldsymbol{\alpha}}_i}^0)] &< \frac{1}{J} \sum_{j=1}^J E[l(x_{ij},\theta_{j,\boldsymbol{\alpha}_i^0}^0)] + c_2 / 2 + \frac{1}{J} h(\pi_{\boldsymbol{\alpha}_i^0}) - \frac{1}{J} h(\pi_{\hat{\boldsymbol{\alpha}}_i})\\
    &< \frac{1}{J} \sum_{j=1}^J E[l(x_{ij},\theta_{j,\boldsymbol{\alpha}_i^0}^0)] + 3c_2/4\\
    &< \frac{1}{J} \sum_{j=1}^J E[l(x_{ij},\theta_{j,\hat{\boldsymbol{\alpha}}_i}^0)],
\end{align*}
where the last inequality is from equation (\ref{eq-delta}) and results in a contradiction.
It indicates that $\{\hat{\boldsymbol{\alpha}}_i \neq \boldsymbol{\alpha}_i^0\} \subset B_{c_2/4}(J)^c$ for $J$ large enough.
And therefore we have
\begin{align*}
\vspace{-5in}
    P\Big(\hat{\boldsymbol{\alpha}}_i \neq \boldsymbol{\alpha}_i^0\Big)
    \leq & P\big(B_{c_2/4}(J)^c\big)\\
    \leq &  P\Big(\underset{\boldsymbol{\alpha}}{\max} \Big|
    \frac{1}{J} \sum_{j=1}^J\big(
    l(x_{ij},\hat{\mu}_{j,\boldsymbol{\alpha}})
    - E[l(x_{ij},\theta_{j,\boldsymbol{\alpha}}^0)]
    \big) \Big| \geq c_2/4 \Big)\\[5pt]
     & \longrightarrow 0, \text{ as } J\xrightarrow{}\infty. \quad \text{ (by Lamma \ref{lemma-2})}
\end{align*}

\end{proof}{}

\subsection{Proof of Theorem 3}
\label{appendix-thm3}
\begin{proof}

Following the proof of Theorem \ref{thm-alpha}, we have
\begin{align*}
    &P\Big(\bigcup_i \big\{\hat{\boldsymbol{\alpha}}_i \neq \boldsymbol{\alpha}_i^0\big\} \ \Big| \ \hat{\boldsymbol{A}}_c = \boldsymbol{A}^0_c \Big)\\
    \leq &\ \sum_{i} P\Big(\big\{\hat{\boldsymbol{\alpha}}_i \neq \boldsymbol{\alpha}_i^0\big\} \ \Big| \ \hat{\boldsymbol{A}}_c = \boldsymbol{A}^0_c \Big)\\
    \leq & \ N \cdot P\Big(B_{c_2/4}(J)^c \ \Big| \ \hat{\boldsymbol{A}}_c = \boldsymbol{A}^0_c \Big)\\
    \leq &\  N \cdot P\Big(\underset{\boldsymbol{\alpha}}{\max} \Big|
    \frac{1}{J} \sum_{j=1}^J\Big(
    l(x_{ij},\hat{\mu}_{j,\boldsymbol{\alpha}})
    - E\big[l(x_{ij},\theta_{j,\boldsymbol{\alpha}}^0)\big]
    \Big) \Big| \geq c_2/4 \ \Big|\ \hat{\boldsymbol{A}}_c = \boldsymbol{A}^0_c \Big)\\[3pt]
    \leq & \ 2^{K+1} NJ \exp(-|C_{\boldsymbol{\alpha}}| \delta_2^2/2)
    + 2^{K+1} NJ \exp (-2|C_{\boldsymbol{\alpha}}|(c_2/8c)^{2/\beta}) + 2^{K+1} N \exp \big(
    - Jc_2^2 / 32M^2
    \big)\\[10pt]
    \leq & \ 2^{K+1} N^2 \exp(-|C_{\boldsymbol{\alpha}}| \delta_2^2/2)
    + 2^{K+1} N^2 \exp (-2|C_{\boldsymbol{\alpha}}|(c_2/8c)^{2/\beta}) + 2^{K+1} N \exp \big(
    - Jc_2 ^2 / 32M^2
    \big).
\end{align*}
Under the Assumption 2, we have $\lim_{n\to\infty}|C_{\boldsymbol{\alpha}}|/N_c \rightarrow \pi_{\boldsymbol{\alpha}}$ almost surely; therefore $N^2 \exp(-|C_{\boldsymbol{\alpha}}| \delta_2^2/2) = N^2 \exp\Big(-\big(1 + o(1)\big)N_c \cdot \pi_{\boldsymbol{\alpha}} \cdot \delta_2^2/2\Big)$ and $N^2 \exp (-2|C_{\boldsymbol{\alpha}}|(c_2/8c)^{2/\beta}) = N^2 \exp \Big(-2 \big(1 + o(1)\big)N_c \cdot \pi_{\boldsymbol{\alpha}}\cdot (c_2/8c)^{2/\beta}\Big)$.
Then we have
\begin{align*}
    &P\Big(\bigcup_i \{\hat{\boldsymbol{\alpha}}_i \neq \boldsymbol{\alpha}_i^0\}\Big)\\
    \leq & \ P\Big(\bigcup_i \{\hat{\boldsymbol{\alpha}}_i \neq \boldsymbol{\alpha}_i^0\}\ \Big| \ \hat{\boldsymbol{A}}_c = \boldsymbol{A}^0_c \Big)P\Big(\hat{\boldsymbol{A}}_c = \boldsymbol{A}^0_c \Big)
	 + P\Big(\bigcup_i \{\hat{\boldsymbol{\alpha}}_i \neq \boldsymbol{\alpha}_i^0\}\ \Big| \ \hat{\boldsymbol{A}}_c \neq \boldsymbol{A}^0_c \Big) P\Big(\hat{\boldsymbol{A}}_c \neq \boldsymbol{A}^0_c \Big)\\
   	\leq & \ P\Big(\bigcup_i \{\hat{\boldsymbol{\alpha}}_i \neq \boldsymbol{\alpha}_i^0\}\ \Big| \ \hat{\boldsymbol{A}}_c = \boldsymbol{A}^0_c \Big)	 + P\Big(\hat{\boldsymbol{A}}_c \neq \boldsymbol{A}^0_c \Big)\\
    \leq & \ 2^{K+1}N^2 \exp\Big(-\big(1 + o(1)\big)N_c \cdot \pi_{\boldsymbol{\alpha}} \cdot \delta_2^2/2\Big)
   + 2^{K+1} N^2 \exp \Big(-\big(1 + o(1)\big)2N_c \cdot \pi_{\boldsymbol{\alpha}}\cdot (c_2/8c)^{2/\beta}\Big)\\
   & \ + 2^{K+1} N \exp \big(
    - Jc_2^2 / 32M^2
    \big) +P\Big(\hat{\boldsymbol{A}}_c \neq \boldsymbol{A}^0_c \Big) \\
    \leq & \ 2^{K+1}N^2 \exp\Big(-\big(1 + o(1)\big)J \cdot \pi_{\boldsymbol{\alpha}} \cdot \delta_2^2/2\Big)
   + 2^{K+1} N^2 \exp \Big(-2\big(1 + o(1)\big)J \cdot \pi_{\boldsymbol{\alpha}}\cdot (c_2/8c)^{2/\beta}\Big) \\
   &\ + 2^{K+1} N \exp \big(
    - Jc_2^2 / 32M^2
    \big) +P\Big(\hat{\boldsymbol{A}}_c \neq \boldsymbol{A}^0_c \Big)  \\
    = & \ 2^{K+1}\Big[N \exp\Big(-\big(1 + o(1)\big)J \cdot \pi_{\boldsymbol{\alpha}} \cdot \delta_2^2/4\Big)\Big]^2
   + 2^{K+1} \Big[ N \exp \Big(-\big(1 + o(1)\big)J \cdot \pi_{\boldsymbol{\alpha}}\cdot (c_2/8c)^{2/\beta}\Big) \Big]^2\\
   & \ + 2^{K+1} N \exp \big(
    - Jc_2^2 / 32M^2
    \big) +P\Big(\hat{\boldsymbol{A}}_c \neq \boldsymbol{A}^0_c \Big)  \\
   \rightarrow &\  0, \text{ as } J\xrightarrow{}\infty.
\end{align*}
Therefore, $\hat{\boldsymbol{\alpha}}_i$'s are uniformly consistent for $\boldsymbol{\alpha}_i$'s for all $i=1,\dots,N$.
\end{proof}

\subsection{Proof of Proposition 1}
\label{appendix-prop3}

\begin{proof}
Our proof uses similar arguments as in \cite{celeux1992classification}.
In Step 3 of Algorithm \ref{algo1}, we have
\begin{equation*}
	L(\boldsymbol{A}^{(t)},\boldsymbol{\mu}^{(t+1)}, \boldsymbol{\pi}^{(t+1)}) \leq L(\boldsymbol{A}^{(t)},\boldsymbol{\mu}^{(t)}, \boldsymbol{\pi}^{(t)}).
\end{equation*}
\noindent Moreover, since
$\hat{\boldsymbol{\alpha}}_i^{(t+1)} = \underset{\boldsymbol{\alpha}}{\arg\min} \ l(\boldsymbol{x}_i,\hat{\boldsymbol{\mu}}_{\boldsymbol{\alpha}}^{(t+1)}) + h(\hat{\pi}_{\boldsymbol{\alpha}}^{(t+1)})$, which is equivalent to that $l\big(\boldsymbol{x}_i,\hat{\boldsymbol{\mu}}_{\hat{\boldsymbol{\alpha}}_i^{(t+1)}}^{(t+1)}\big) + h\big(\hat{\pi}_{\hat{\boldsymbol{\alpha}}_i^{(t+1)}}^{(t+1)}\big) \leq l(\boldsymbol{x}_i,\hat{\boldsymbol{\mu}}_{\boldsymbol{\alpha}}^{(t+1)}) + h(\hat{\pi}_{\boldsymbol{\alpha}}^{(t+1)})$ for any $\boldsymbol{\alpha}\neq \hat{\boldsymbol{\alpha}}_i^{(t+1)}$,
we have
\begin{equation}
	L(\boldsymbol{A} ^{(t+1)},\boldsymbol{\mu}^{(t+1)}, \boldsymbol{\pi}^{(t+1)}) \leq L(\boldsymbol{A} ^{(t)},\boldsymbol{\mu}^{(t)}, \boldsymbol{\pi}^{(t)}).
\end{equation}
Therefore the criterion \eqref{eq-C2} is decreasing.

In the finite sample setting, since there is finite number of partitions into $2^K$ classes, the decreasing sequence $L(\boldsymbol{A} ^{(t)},\boldsymbol{\mu}^{(t)}, \boldsymbol{\pi}^{(t)})$ also takes a finite number of values, which makes it converge to a stationary value.
Moreover, since the minima of the loss function is well-defined, the sequence $(\boldsymbol{A} ^{(t)},\boldsymbol{\mu}^{(t)}f, \boldsymbol{\pi}^{(t)})$ also converges.
\end{proof}

\subsection{Proof of Proposition 2}
\label{appendix-prop4}

\begin{proof}
Our proof directly follows that in \cite{celeux1992classification}. Since
\begin{align*}
L(\boldsymbol{U},\boldsymbol{\mu},\boldsymbol{\pi}) = & \sum_{\boldsymbol{\alpha}\in \{0,1\}^K} \sum_{i=1}^n u_{i\boldsymbol{\alpha}} \Big(l(\boldsymbol{x_i},\boldsymbol{\mu}_{\boldsymbol{\alpha}}) + h(\pi_{\boldsymbol{\alpha}})\Big)\\
	\geq & \sum_{\boldsymbol{\alpha}\in \{0,1\}^K} \sum_{i=1}^n u_{i\boldsymbol{\alpha}} \min_{\boldsymbol{\alpha}'}\Big(l(\boldsymbol{x_i},\boldsymbol{\mu}_{\boldsymbol{\alpha}'}) + h(\pi_{\boldsymbol{\alpha}'})\Big)\\
	\geq & \sum_{\boldsymbol{\alpha}\in \{0,1\}^K} \sum_{i=1}^n \min_{\boldsymbol{\alpha}'}\Big(l(\boldsymbol{x_i},\boldsymbol{\mu}_{\boldsymbol{\alpha}'}) + h(\pi_{\boldsymbol{\alpha}'})\Big),
\end{align*}
where the RHS is attained when $\boldsymbol{U}$ is equivalent to some partition,
the Algorithm \ref{algo1} can be regarded as an alternating optimization algorithm to minimize $L(\boldsymbol{U},\boldsymbol{\mu},\boldsymbol{\pi})$.
Specifically, the Algorithm \ref{algo1} is in fact a grouped coordinate descent method.
Following the Theorem 2.2 of \cite*{bezdek1987local}, the Proposition \ref{prop-finite-2} is proved.
\end{proof}

%\bibliographystyle{chicago}
%\bibliography{bibref}

%\end{document}

\end{document}